\documentclass{lmcs}


\usepackage{xcolor} 
\usepackage[
    colorlinks=true,    
    linkcolor=blue,     
    citecolor=green,    
    urlcolor=red        
]{hyperref}
\hypersetup{
    colorlinks=true,
    linkcolor=blue,
    citecolor=green,
    urlcolor=red
}
\usepackage{amsfonts}
\usepackage{amsmath}
\usepackage{amssymb}

\usepackage{float}
\usepackage{tikz}
\usetikzlibrary{intersections, calc, arrows, automata, positioning,shapes, trees}
\usepackage{bbding}
\usepackage{pifont}
\usepackage{pgf, verbatim}

\usepackage{caption}
\usepackage{subcaption}

\usepackage[capitalise]{cleveref}
\crefname{thm}{Theorem}{Theorems}   
\Crefname{thm}{Theorem}{Theorems}  
\crefname{prop}{Proposition}{Propositions}   
\Crefname{prop}{Proposition}{Propositions}  
\crefname{lem}{Lemma}{Lemmas}   
\Crefname{lem}{Lemma}{Lemmas}  
\crefname{defi}{Definition}{Definitions}   
\Crefname{defi}{Definition}{Definitions}  
\crefname{cor}{Corollary}{Corollaries}   
\Crefname{cor}{Corollary}{Corollaries}  
\crefname{exa}{Example}{Examples}   
\Crefname{exa}{Example}{Examples}  
\Crefname{figure}{Figure}{Figures}
\crefname{figure}{Figure}{Figures}
\Crefname{remark}{Remark}{Remarks}
\crefname{remark}{Remark}{Remarks}
\usepackage{setspace}
\usepackage[linesnumbered,ruled]{algorithm2e}

   \def\@citecolor{blue}%
   \def\@urlcolor{blue}%
   \def\@linkcolor{blue}%
\def\orcidID#1{\smash{\href{http://orcid.org/#1}{\protect\raisebox{-1.25pt}{\protect\includegraphics{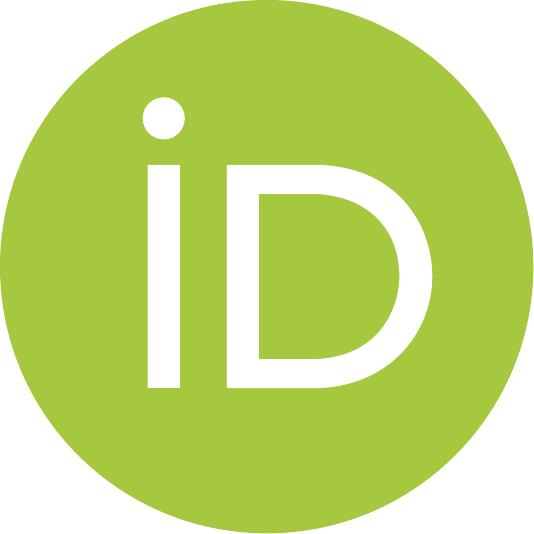}}}}}

\newcommand{\fork}{{\mathsf{fork}}} 
\newcommand{\hide}{{\mathsf{hide}}} 
\newcommand{\prune}{{\mathsf{prune}}} 
\newcommand{\relabel}{{\mathsf{relabel}}} 

\newcommand{\suc}{{\mathsf{succ}}} 
\newcommand{\supp}{{\mathsf{support}}}
\newcommand{\wide}{{\mathsf{wide}}} 
\newcommand{\pfun}{\mathrel{\ooalign{\hfil$\mapstochar\mkern5mu$\hfil\cr$\to$\cr}}} 
\newcommand{\Dist}{\mathrm{Distr}} 
\newcommand{\suchthat}{\;\ifnum\currentgrouptype=16 \middle\fi|\;} 
\newcommand{\Paths}{{\rm Paths}} 

\newcommand{\T}{{\mathcal T}} 
\newcommand{\D}{{\mathcal D}} 
\newcommand{\A}{{\mathcal A}}
\newcommand{\B}{{\mathcal B}}
\newcommand{\M}{{\mathcal M}}
\newcommand{\C}{{\mathcal C}} 
\newcommand{\G}{{\mathcal G}} 
\newcommand{\N}{{\mathcal N}} 
\newcommand{\U}{{\mathcal U}} 
 
\newcommand{\lang}{{\mathcal L}} 
\newcommand{\m}{{\sf m}} 
\newcommand{\Act}{{\mathsf{Act}}} 
\newcommand{\Prob}{{\sf P}} 
\newcommand{\nat}{{\mathbb N}} 
\newcommand{\QDiamond}{{Q_{\text{\ding{73}}}}} 

\definecolor{OliveGreen}{rgb}{0,0.6,0}

\iffalse
\newcommand{\rev}[1]{{\color{blue}#1}}
\else
\newcommand{\rev}[1]{{\color{black}#1}}
\fi
\theoremstyle{plain} 

\usepackage{lineno}

\begin{document}\sloppy

\title[]{On Good-for-MDPs Automata}

\author{Sven Schewe\lmcsorcid{0000-0002-9093-9518}}
\author{Qiyi Tang\lmcsorcid{0000-0002-9265-3011}}
\author{Tansholpan Zhanabekova\lmcsorcid{0000-0002-4941-2554}}

\address{University of Liverpool, UK}	
\email{\{sven.schewe, qiyi.tang, t.zhanabekova\}@liverpool.ac.uk}  





\begin{abstract}
  \noindent Nondeterministic good-for-MDPs (GFM) automata are for MDP model checking and reinforcement learning what good-for-games (GFG) automata are for reactive synthesis: a more compact alternative to deterministic automata that displays nondeterminism, but only so much that it can be resolved locally, such that a syntactic product can be analysed.
  
  GFM has recently been introduced as a property for reinforcement learning, where the simpler B\"uchi acceptance conditions it allows to use is key. However, while there are classic and novel techniques to obtain automata that are GFM, there has not been a decision procedure for checking whether or not an automaton is GFM.
  We show that GFM-ness is decidable and provide an {\sf EXPTIME} decision procedure as well as a {\sf PSPACE}-hardness proof.

We also compare the succinctness of GFM automata with other types of automata with restricted nondeterminism. The first natural comparison point are GFG automata.
Deterministic automata are GFG, and GFG automata are GFM, but not vice versa.
This raises the question of how these classes relate in terms of succinctness.
GFG automata are known to be exponentially more succinct than deterministic automata, but the gap between GFM and GFG automata as well as the gap between ordinary nondeterministic automata and those that are GFM have been open.
We establish that these gaps are exponential, and sharpen this result by showing that the latter gap remains exponential when restricting the nondeterministic automata to separating safety or unambiguous reachability automata.
  
\end{abstract}

\maketitle

	\section{Introduction}
	Omega-automata \cite{Thomas90,Perrin04} are formal acceptors of $\omega$-regular properties, which often result from translating a formula from temporal logics like LTL~\cite{Pnueli77}, as a specification for desired model properties in quantitative model checking and strategy synthesis \cite{Baier08}, and reinforcement learning~\cite{Sutton18}.
	
	Especially for reinforcement learning, having a simple B\"uchi acceptance mechanism has proven to be a breakthrough \cite{HahnPSSTW19}, which led to the definition of the ``good-for-MDPs'' (GFM) property in \cite{HPSSTWBP2020}.
	Just like for good-for-games automata in strategy synthesis for strategic games \cite{Henzin06}, there is a certain degree of nondeterminism allowed when using a nondeterministic automaton on the syntactic product with an MDP to learn how to control it, or to apply quantitative model checking. 
	\rev{Moreover, the degree of freedom available when controlling MDPs is higher than when controlling games: unlike in an adversarial game, where a single bad resolution of nondeterminism can be exploited immediately and irrevocably, a trace in an MDP almost surely eventually enters an end-component from which it can never escape, so that the same finite set of states and choices keeps recurring forever. 
    This gives a nondeterministic automaton effectively unlimited further opportunities to resolve its nondeterminism correctly once such an end-component is reached.}
    \rev{In particular, this always allows for using nondeterministic automata with a B\"uchi acceptance condition: both for the classically used limit-deterministic (or semi-deterministic) B\"uchi automata -- introduced by Courcoubetis and Yannakakis, building on earlier work by Vardi \cite{Vardi85,DBLP:conf/lics/VardiW86,Courco95}, to solve the qualitative probabilistic model-checking problem of deciding whether the executions of a Markov chain or MDP satisfy a given LTL formula with probability one, and refined since into more efficient constructions \cite{Hahn15,Sicker16b} -- and for alternative GFM automata like the slim automata from \cite{HPSSTWBP2020}.}

	This raises the question of whether or not an automaton is GFM.
	While \cite{HPSSTWBP2020} has introduced the concept, there is not yet a decision procedure for checking the GFM-ness of an automaton, let alone for the complexity of this test.
	
	We will start by showing that the problem of deciding GFM-ness is {\sf PSPACE}-hard by a reduction from 
	the NFA universality problem \cite{StockmeyerM73}.
	We then define the auxiliary concept of \emph{qualitative GFM}, QGFM, which relaxes the requirements for GFM to qualitative properties, and develop an automata-based {\sf EXPTIME} decision procedure for QGFM.
	This decision procedure is constructive in that it can provide a counter-example for QGFM-ness when such a counter-example exists.
	We then use it to provide a decision procedure for GFM-ness that uses QGFM queries for all states of the candidate automaton.
	Finally, we show that the resulting criterion for GFM-ness is also a necessary criterion for QGFM-ness, which leads to a collapse of the two concepts. This entails that the {\sf EXPTIME} decision procedure we developed to test QGFM-ness can be used to decide GFM-ness, while our {\sf PSPACE}-hardness proofs extend to QGFM-ness.

We then close this article by studying the succinctness of GFM automata, with a focus on a comparison to good-for-games (GFG) automata.
Like GFM automata, they need to be able to resolve their nondeterminism on-the-fly, but the restriction is stricter: they need to be able to resolve their nondeterminism on-the-fly in every context, including in two-player games, while it only needs to be possible for GFM automata in the context of finite MDPs.
This makes the restrictions posed on GFM automata milder.
One of the implications is that GFM automata are less complex than GFG automata, as GFM B\"uchi automata can recognise all regular languages.
But are they also \emph{significantly} more succinct?

We show that this is the case, establishing that GFM automata are exponentially more succinct than GFG automata, where exponentially more succinct means that a translation of automata with $n$ states (or transitions) require automata with $2^{\Omega(n)}$ states in the worst case.

To show this, we develop a family $\G_n$ of automata with $n+2$ states and $3n+7$ transitions, 
and show that the smallest language equivalent deterministic B\"uchi automaton (DBW) has at least $2^{n-1}$ states. As GFG B\"uchi automata are only quadratically more succinct than DBWs \cite{KuperbergS15}, this implies that GFG automata recognising the language are in $\Omega(2^{n/2})$.

We then continue to show a similar gap between B\"uchi automata with general nondeterminism and those that are GFM. 
\rev{As a single family of automata cannot witness this gap for both reachability and safety languages at once, we produce two families of automata, $\mathcal R_n$ and $\mathcal S_n$, with $n+2$ states and $2n+5$ transitions and $n+1$ states and $2n$ transitions, respectively: $\mathcal R_n$ witnesses the gap for reachability automata, and $\mathcal S_n$ witnesses it for safety automata -- and, as we show below, for the smaller class of separating safety automata as well. In both cases, we show that a language equivalent GFM automaton requires at least $2^n$ states to recognise the language in question.}

This is not only close to the known translation to slim automata \cite{HPSSTWBP2020} that results in $3^n$ states (and no more than two successors for any state letter pair), but it provides more: each $\mathcal R_n$ is a run-unambiguous reachability automaton, while each $\mathcal S_n$ is a safety automaton, which is run-unambiguous and separating. \rev{Here, an automaton is \emph{run-unambiguous} if every word admits at most one run, accepting or not (this is stronger than being \emph{unambiguous}, which only asks for at most one \emph{accepting} run), and \emph{separating} if distinct states recognise pairwise disjoint languages, which in particular implies unambiguity.}
Moreover, we show that the languages recognised by $\mathcal R_n$ are also recognised by a separating B\"uchi automaton $\mathcal R_n'$ with $n+2$ states and $3(n+2)$ transitions.

That the succinctness gap extends to these smaller classes of automata is not only interesting from a principle point of view: unambiguous automata appear naturally both as a standard translation from LTL, and in model checking Markov chains \cite{BaierKKMW23}.
Being unambiguous is thus a dissimilar type of restriction to nondeterminism, very different in nature, and yet with fields of application related to GFM.
\rev{It is also interesting to note that separation is more expensive ($\Omega((n-1)!)$, \cite{KarmarkarJC13}) than translating a nondeterministic B\"uchi automaton into one that is GFM, which costs only $O(3^n)$ via the slim automata construction already mentioned above \cite{HPSSTWBP2020} (or, alternatively, via limit-determinisation \cite{Hahn15,Sicker16b}).}

	\section{Preliminaries}
	
	We write $\nat$ for the set of nonnegative integers. 
	Let $S$ be a finite set. 
	We denote by $\Dist(S)$ the set of probability distributions on~$S$. For a distribution $\mu \in \Dist(S) $ we write $\supp(\mu) = \{s \in S \suchthat \mu(s) > 0 \}$ for its support. The cardinal of $S$ is denoted \rev{by} $|S|$. We use $\Sigma$ to denote a finite alphabet. We denote by $\Sigma^{*}$ the set of finite words over $\Sigma$ and $\Sigma^{\omega}$ the set of $\omega$-words over $\Sigma$. We use the standard notions of prefix and suffix of a word. By $w\alpha$ we denote the concatenation of a finite word $w$ and an $\omega$-word $\alpha$. 
	If $L \subseteq \Sigma^{\omega}$ and $w \in \Sigma^{*}$, the residual language (left quotient of $L$ by $w$), denoted by $w^{-1}L$ is defined as $\{\alpha \in \Sigma^{\omega} \mid w\alpha \in L\}$.
	
	\subsection{Automata}\label{subsection:preliminaries-automata} 
	A nondeterministic B\"uchi \textbf{word automaton} (NBW) is a tuple $\A= (\Sigma, Q, q_0, \delta, F)$, where $\Sigma$ is a finite alphabet, $Q$ is a finite set of states, $q_0 \in Q$ is the initial state, $\delta: Q \times \Sigma \to 2^Q$ is the transition function, and $F \subseteq Q$ is the set of final (or accepting) states. 
	{An NBW is \emph{complete} if $\delta(q, \sigma) \neq \emptyset$ for all $q \in Q$ and $\sigma \in \Sigma$.
		Unless otherwise mentioned, we consider complete NBWs in this paper.} 
	A run $r$ of $\A$ on $w \in \Sigma^{\omega}$ is an $\omega$-word $q_0, w_0, q_1, w_1, \ldots \in (Q \times \Sigma)^{\omega}$ such that $q_i \in \delta(q_{i-1} , w_{i-1})$ for all $i > 0$.
	An NBW $\A$ accepts exactly those runs, in which at least one of the infinitely often occurring states is in~$F$. 
	A word in $\Sigma^\omega$ is accepted by the automaton if it has an accepting run, and the language of an automaton, denoted $\lang(\A)$, is the set of accepted words in $\Sigma^{\omega}$.  
    
    \rev{We use common three-letter abbreviations to name types of automata throughout this paper. The first letter (D, N, U, or A) says whether the automaton is deterministic, nondeterministic, universal, or alternating (the latter two only become relevant for the tree automata introduced in \cref{section:QGFM-procedure}). The second letter denotes the acceptance condition (B for B\"uchi, C for co-B\"uchi, or P for parity, if neither B\"uchi nor co-B\"uchi). The third letter says whether the automaton is an $\omega$-word automaton (W) or a tree automaton (T, S, or N -- again only relevant later, marking symmetric and nearly symmetric tree automata, respectively, with T used otherwise). For example, an NBW, as just defined, is a nondeterministic B\"uchi $\omega$-word automaton.}
	An example of an NBW is given in \cref{fig:NBW}.
	
	\begin{figure}
		\centering
		\begin{tikzpicture} [xscale=.7,shorten >=1pt,node distance = 1cm, on grid, auto,initial text = {}]
			\usetikzlibrary{positioning,arrows,automata}
			\node[initial below,state] (q0) at (0,0) {$q_0$};
			\node[accepting,state] (q1) at (-4,0) {$q_1$};
			\node[accepting,state] (q2) at (4,0) {$q_2$};
			\node[state] (q3) at (0,2) {$q_3$};
			
			\path [->] (q0) edge [above] node {$a$} (q1)
			(q0) edge [above] node {$a$} (q2)
			(q0) edge [left] node {$b$} (q3)
			(q1) edge [above] node {$b$} (q3)
			(q2) edge [above] node {$a$} (q3)
			(q1) edge [loop left] node{$a$} (q1)
			(q2) edge [loop right] node{$b$} (q2)
			(q3) edge [loop above] node{$a,b$} (q3);
		\end{tikzpicture}
		\caption{An NBW over $\{a, b\}$. This NBW is complete and accepts the language $\{a^{\omega}, ab^{\omega}\}$.}
		\label{fig:NBW}
	\end{figure}
	
	Let $C \subset \nat$ be a finite set of colours. A nondeterministic parity \textbf{word automaton} (NPW) is a tuple $P= (\Sigma, Q, q_0, \delta, \pi)$, where $\Sigma$, $Q$, $q_0$ and $\delta$ have the same definitions as for NBW, and $\pi: Q  \to C $ is the priority (or colouring) function that maps each state to a priority (or colour). A run is accepting if and only if the highest priority occurring infinitely often in the infinite sequence is even. Similar to NBWs, a word in $\Sigma^\omega$ is accepted by an NPW if it has an accepting run, and the language of the NPW $P$, denoted $\lang(P)$, is the set of accepted words in $\Sigma^{\omega}$. An NBW is a special case of an NPW where $\pi(q) = 2$ for $q \in F$ and $\pi(q) = 1$ otherwise with $C = \{1, 2\}$. 
	
	A nondeterministic word automaton is \emph{deterministic} if the transition function $\delta$ maps each state and letter pair to a singleton set, a set consisting of a single state. 

	\rev{A state $q$ of a nondeterministic word automaton is a \emph{sink} if $\delta(q,\sigma) = \{q\}$ for all $\sigma \in \Sigma$, i.e., once entered, it can never be left. 
    A nondeterministic word automaton is called a \emph{safety automaton} if every state is final except for at most one rejecting sink, and a \emph{reachability automaton} if it has only a single final state, which is a sink.}
    
	A nondeterministic automaton is called \emph{good-for-games (GFG)} if it only relies on a limited form of nondeterminism: GFG automata can make their decision of how to resolve their nondeterministic choices on the history at any point of a run rather than using the knowledge of the complete word as a nondeterministic automaton normally would without changing their language. 
    \rev{Formally, an automaton is GFG if Eve can win the \emph{letter game} introduced by \cite{Henzin06}. 
    A play of the letter game starts with a token in the initial state and proceeds in infinitely many rounds. 
    In each round, Adam selects a letter from the input alphabet, and Eve responds by moving the token along a corresponding transition of the automaton. 
    Adam constructs a word, while Eve constructs a run for this word. 
    Eve wins the play if either the word is not in the language of the automaton, or the constructed run is accepting.
    We note that GFG automata are also known as history-deterministic automata in more recent literature.}

    We also consider automata that operate on finite words. 
    A \emph{nondeterministic finite automaton} (NFA) $\N$ is a tuple $(\Sigma, Q, q_0, \delta, F)$, where $F \subseteq Q$ is a set of \emph{final} states.
    Unlike NBWs, NFAs process \emph{finite} \rev{words} rather than $\omega$-\rev{words}.
    A finite word $u \in \Sigma^{*}$ is \emph{accepted} by an NFA $\N$ if there exists a finite run of $\N$ on $u$ that terminates in a final state.
    A \emph{deterministic finite automaton} (DFA) is an NFA in which the transition function is deterministic, defined analogously to deterministic $\omega$-automata.

	\subsection{Markov Decision Processes (MDPs)}\label{subsection:preliminaries-mdp}
	
	A \emph{(finite, state-labelled) Markov decision process} (MDP) is a tuple $\langle S, \Act, \Prob, \Sigma, \ell\rangle$ consisting of a finite set $S$ of states, a finite set $\Act$ of actions, a partial function 
	$\Prob: S \times \Act \pfun \Dist(S)$ 
	denoting the probabilistic transition and a labelling function $\ell: S \to \Sigma$.
	The set of available actions in a state $s$ is $\Act(s) = \{\m \in \Act \suchthat \Prob(s, \m) \text{ is defined}\}$. 
    An MDP is called \rev{a \emph{(labelled) Markov chain} (MC)} if $|\Act(s)| = 1$ for all $s \in S$\rev{, i.e., if there is no actual choice of action to make in any state}.
	
	\begin{figure}
			\resizebox{!}{3.5cm}{%
				\begin{tikzpicture} [xscale=.7,shorten >=1pt,node distance = 1cm, on grid, auto]
					\usetikzlibrary{positioning,arrows,automata}
					\tikzstyle{BoxStyle} = [draw, circle, fill=black, scale=0.4,minimum width = 1pt, minimum height = 1pt]
					
					\node[state,fill=orange!40] (s0) at (0,0) {$s_0$};
					\node[BoxStyle] (act) at (0,-1) {};
					\node[label] at (0.3,-0.8) {$\m$};
					\node[state,fill=orange!40] (s1) at (-2,-2) {$s_1$};
					\node[state,fill=green!30] (s2) at (2,-2) {$s_2$};
					\node[BoxStyle] at (-3.25,-2) {};
					\node[label] at (-3.3, -2.3) {$\m$};
					\node[label] at (-3, -1.65) {$1$};
					\node[BoxStyle] at (3.25,-2) {};		
					\node[label] at (3.3, -1.7) {$\m$};
					\node[label] at (3, -2.35) {$1$};
					
					\path [-] (s0) edge [above] node {} (act);
					
					\path [->] (act) edge [above, midway] node {$\frac{1}{3}$} (s1)
					(act) edge [above, midway] node {$\frac{2}{3}$} (s2)
					(s1) edge [loop left, near end] node {} (s1)
					(s2) edge [loop right, near end] node {} (s2);
					
					\node[rectangle,draw,dashed, minimum width = 2.5cm, minimum height = 1.3cm] at (-2,-2) {};
					\node[rectangle,draw,dashed, minimum width = 2.5cm, minimum height = 1.3cm] at (2,-2) {};
					
				\end{tikzpicture}
			}
		\caption{An MDP with initial state $s_0$. The set of labels is $\{a, b\}$ and the labelling function for the MDP is as follows: $\ell(s_0) = \ell(s_1) = a$, $\ell(s_2) = b$. The labels are indicated by different colours. Since each state has only one available action $\m$, the MDP is actually an MC. There are two end-components in this MDP labelled with the two dashed boxes.}
		\label{fig:MDP}
	\end{figure}

	
	An infinite \emph{run (path)} of an MDP $\M$ is a sequence $s_0\m_1\ldots \in (S \times \Act)^\omega$ such that $\Prob(s_i, \m_{i+1})$ is defined and $\Prob(s_i, \m_{i+1})(s_{i+1}) > 0$ for all $i \ge 0$.
	A finite run is a finite such sequence.
	Let $\Omega(\M)$ (resp. $\Paths(\M)$) denote the set of infinite (resp. finite) runs in $\M$ and $\Omega(\M)_s$ (resp. $\Paths(\M)_s$) denote the set of infinite (resp. finite) runs in $\M$ starting from $s$.
	Abusing the notation slightly, for an infinite run $r = s_0\m_1s_1\m_2\ldots$ we write $\ell(r) = \ell(s_0)\ell(s_1)\ldots \in \Sigma^{\omega}$.
	
	A \emph{strategy} for an MDP is a function $\mu: \Paths(\M) \to \Dist(\Act)$ that, given a finite run~$r$, returns a probability distribution on all the available actions at the last state of $r$. 
	A \emph{memoryless} (positional) strategy for an MDP is a function $\mu: S \to \Dist(\Act)$ that, given a state $s$, returns a probability distribution on all the available actions at that state. 
	The set of runs $\Omega(\M)_s^{\mu}$ is a subset of $\Omega(\M)_s$ that \rev{corresponds} to strategy $\mu$ and initial state~$s$.
	Given a memoryless/finite-memory strategy $\mu$ for $\M$, an MC $(\M)_{\mu}$ is induced \cite[Section~10.6]{Baier08}.

    \rev{A strongly connected component (SCC) is a maximally strongly connected set of states. 
    A bottom SCC (BSCC) is an SCC from which no state outside it is reachable.}
	A \emph{sub-MDP} of $\M$ is an MDP $\M' = \langle S', \Act', \Prob', \Sigma, \ell' \rangle$, where $S' \subseteq S$, $\Act' \subseteq \Act$ is such that $\Act'(s) \subseteq \Act(s)$ for every $s \in S'$, and $\Prob'$ and $\ell'$ are analogous to $\Prob$ and $\ell$ when restricted to $S'$ and $\Act'$.
    In particular, $\M'$ is closed under probabilistic transitions, i.e. for all $s \in S'$ and $\m \in \Act'$ we have that $\Prob'(s, \m)(s') > 0$ implies that $s' \in S'$. 
    \rev{An \emph{end-component} \cite{Alfaro1998,Baier08} of an MDP $\M$ is a sub-MDP $\M'$ of $\M$ whose underlying graph is strongly connected: every state of $\M'$ can reach every other using only the retained actions of $\M'$, and -- since $\M'$ is a sub-MDP -- none of these actions has, in $\M$, a transition from a state in $\M'$ to a state not in $\M'$.}
    An example MDP is presented in \cref{fig:MDP}.
	
	A strategy $\mu$ and an initial state $s \in S$ induce a standard probability measure on sets of infinite runs, see, e.g., \cite{Baier08}. Such measurable sets of infinite runs are called events or objectives. We write $\mathrm{Pr}_s^{\mu}\rev{(E)}$ for the probability of an event $E \subseteq sS^{\omega}$ of runs starting from $s$.
	
	\begin{thm}\label{theorem:end-component-properties}(End-Component Properties \cite{Alfaro1998,Baier08}). 
		Once an end-component $E$ of an MDP is entered, there is a strategy that visits every state-action combination in $E$ with probability one and stays in $E$ forever. Moreover, for every
		strategy the union of the end-components is visited with probability one.
	\end{thm}

	\subsection{The Product of MDPs and Automata}\label{subsection:preliminaries-product-mdp-gfm}
	
	Given an MDP $\M = \langle S, \Act, \Prob, \Sigma, \ell \rangle$ with initial state $s_0 \in S$ and an NBW $\N = \langle \Sigma, Q, \delta, q_0, F \rangle$, we want to compute an optimal strategy satisfying the objective that the run of $\M$ is in the language of $\N$. We define the semantic satisfaction probability for $\N$ and a strategy $\mu$ from state $s \in S$ as:
	$\mathrm{PSem}_{\N}^{\M}(s, \mu) = \mathrm{Pr}_s^{\mu}\{ r \in \Omega_s^{\mu}: \ell(r) \in \lang(\N)\}$ and $\mathrm{PSem}_{\N}^{\M}(s) = \sup_{\mu} \big( \mathrm{PSem}_{\N}^{\M}(s, \mu) \big)$. In the case that $\M$ is an MC, we simply have $\mathrm{PSem}_{\N}^{\M}(s) = \mathrm{Pr}_s\{ r \in \Omega_s: \ell(r) \in \lang(\N)\}$.
	
	\begin{figure*}[t]
		\centering
		\scalebox{0.82}{
			\begin{tikzpicture} [xscale=.7,shorten >=1pt,node distance = 1cm, on grid, auto]
				\usetikzlibrary{positioning,arrows,automata}
				\tikzstyle{BoxStyle} = [draw, circle, fill=black, scale=0.4,minimum width = 1pt, minimum height = 1pt]
				
				\node[state,fill=orange!40] (s0q0) at (0,0) {$s_0, q_0$};
				\node[BoxStyle] (act1) at (-2.3,0) {};
				\node[label] at (-1.8,0.3) {$(\m,q_1)$};
				\node[BoxStyle] (act2) at (2.3,0) {};
				\node[label,color=red] at (1.8,0.3) {$(\m, q_2)$};
				
				\node[state,fill=orange!40] (s1q1) at (-5,1) {$s_1, q_1$};
				\node[BoxStyle] at (-6.55,1) {};
				\node[label] at (-7, 0.61) {$(\m,q_1)$};
				\node[label] at (-6, 1.41) {$1$};
				
				\node[state,fill=green!30] (s2q1) at (-5,-1) {$s_2, q_1$};
				\node[BoxStyle] at (-7.3,-1) {};
				\node[label] at (-7,-1.3) {$(\m, q_3)$};
				\node[state,fill=green!30] (s2q3) at (-10,-1) {$s_2, q_3$};
				\node[BoxStyle] at (-11.55,-1) {};
				\node[label] at (-12, -1.39) {$(\m,q_3)$};
				\node[label] at (-11, -0.6) {$1$};
				
				\node[state,fill=orange!40] (s1q2) at (5,1) {$s_1, q_2$};
				\node[BoxStyle] at (7.3,1) {};
				\node[label] at (7,1.3) {$(\m, q_3)$};

				\node[state,fill=orange!40] (s1q3) at (10,1) {$s_1, q_3$};				\node[BoxStyle] at (11.55,1) {};
				\node[label] at (12,1.39) {$(\m,q_3)$};
				\node[label] at (11,0.56) {$1$};
				
				\node[state,fill=green!30] (s2q2) at (5,-1) {$s_2, q_2$};
				\node[BoxStyle] at (6.55,-1) {};
				\node[label] at (7,-0.61) {$(\m,q_2)$};
				\node[label] at (6,-1.45) {$1$};
				
				\path [-] (s0q0) edge [above] node {} (act1);
				\path [-,red,line width=2pt] (s0q0) edge [above] node {} (act2);
				
				\path [->] 
				(act1) edge [above, midway] node {$\frac{1}{3}$} (s1q1)
				(act1) edge [below] node {$\frac{2}{3}$} (s2q1)
				(s1q1) edge [loop left, near end] node {} (s1q1)
				(s2q1) edge [above, near end] node {$1$} (s2q3)
				(s2q3) edge [loop left, near end] node {} (s2q3)
				(act2) edge [above, midway] node {$\frac{1}{3}$} (s1q2)
				(act2) edge [below, midway] node {$\frac{2}{3}$} (s2q2)
				(s2q2) edge [loop right, near end] node {} (s2q2)
				(s1q2) edge [above, near end] node {$1$} (s1q3)
				(s1q3) edge [loop right, near end] node {} (s1q3);
				
				\node[rectangle,draw,dashed, minimum width = 3cm, minimum height = 1.3cm, line width=2pt] at (-5.7,1) {}; 
				\node[rectangle,draw,dashed, minimum width = 3cm, minimum height = 1.3cm] at (-10.7,-1) {}; 
				\node[rectangle,draw,dashed, minimum width = 3cm, minimum height = 1.3cm] at (10.7,1) {}; 
				\node[rectangle,draw,dashed, minimum width = 3cm, minimum height = 1.3cm, line width=2pt] at (5.7,-1) {}; 
				
			\end{tikzpicture}
		}
		\bigskip
		\caption{An example of a product MDP $\M \times \N$ with initial state $(s_0, q_0)$ and $F^{\times} = \{(s_1, q_1), (s_2, q_1), (s_1, q_2), (s_2, q_2)\}$ where $\M$ is the MDP (MC) in \cref{fig:MDP} and $\N$ is the NBW in \cref{fig:NBW}. The states $(s_0, q_0), (s_1, q_1), (s_1, q_2)$ and $(s_1, q_3)$ are labelled with $a$ while all the other states are labelled with $b$. Again, the four end-components of the MDP are labelled with dashed boxes; the upper left and lower right end-components are accepting (highlighted in thick dashed boxes).}
		\label{fig:product-MDP}
	\end{figure*}

	The \emph{product} of $\M$ and $\N$ is an MDP $\M \times \N = \langle S \times Q, \Act \times Q, \Prob^{\times}, \Sigma, \ell^{\times} \rangle$ augmented with the initial state $(s_0, q_0)$ and the B\"uchi acceptance condition $F^{\times}
	= \{(s, q) \in S \times Q \suchthat q \in F \}$. The labelling function $\ell^{\times}$ maps each state $(s, q) \in S \times Q$ to $\ell(s)$. 
	
	We define the partial function $\Prob^{\times}: (S \times Q) \times (\Act \times Q) \pfun \Dist(S \times Q)$ as follows: for all $(s, \m)$ \rev{such that $\Prob(s,\m)$ is defined}, 
    $s' \in S$ and $q, q' \in Q$,  we have  
	$\Prob^{\times}\big((s, q), (\m, q')\big) \big( (s', q') \big) = 
	\Prob(s, \m)(s')
	$ for all $q' \in \delta(q, \ell(s))$.
	When $\N$ is complete, there always exists a state $q'$ such that $q' \in \delta(q, \ell(s))$.

	
	We define the syntactic satisfaction probability for the product MDP and a strategy $\mu^{\times}$ from a state $(s, q)$ as:
	$\mathrm{PSyn}_{\N}^{\M} \big( (s, q), \mu^{\times}\big) = \mathrm{Pr}_{s, q}^{\mu^{\times}}\{ r \in \Omega_{s, q}^{\mu^{\times}}: \ell^{\times}(r) \in \lang(\N) \}$
    and $\mathrm{PSyn}_{\N}^{\M}(s) = \sup_{\mu^{\times}} (\mathrm{PSyn}_{\N}^{\M}\big( (s,q_0), \mu^{\times})\big)$. 
	The set of actions is $\Act$ in the MDP $\M$ while it is $\Act \times Q$ in the product MDP. This makes $\mathrm{PSem}$ and $\mathrm{PSyn}$ potentially different. In general, $\mathrm{PSyn}_{\N}^{\M}(s) \le \mathrm{PSem}_{\N}^{\M}(s)$ for all $s \in S$, because accepting runs can only occur on accepted words. 
	If $\N$ is deterministic, $\mathrm{PSyn}_{\N}^{\M}(s) = \mathrm{PSem}_{\N}^{\M}(s)$ holds for all $s \in S$. 
	
	End-components and runs are defined for products just like for MDPs.
    A run of $\M \times \N$ is accepting if it satisfies the product's acceptance condition. 
    We then have $\mathrm{PSyn}_{\N}^{\M}( (s, q), \mu^{\times}) = \mathrm{Pr}_{s, q}^{\mu^{\times}}\{ r \in \Omega_{s, q}^{\mu^{\times}}: r \text{ is an accepting run}\}$.
    An accepting end-component of $\M \times \N$ is an end-component which contains some states in $F^{\times}$. 
	
	An example of a product MDP is presented in \cref{fig:product-MDP}. It is the product of the MDP in \cref{fig:MDP} and the NBW in \cref{fig:NBW}. Since $\ell(r)$ is in the language of the NBW for every run $r$ of the MDP,  we have $\mathrm{PSem}_{\N}^{\M}(s_0) = 1$. However, the syntactic satisfaction probability $\mathrm{PSyn}_{\N}^{\M}(s_0) = \frac{2}{3}$ is witnessed by the memoryless strategy which chooses the action $(\m,q_2)$ at the initial state. We do not need to specify the strategy for the other states since there is only one available action for any remaining state. According to the following definition, the NBW in \cref{fig:NBW} is not GFM as witnessed by the MDP in \cref{fig:MDP}. 

    \begin{defi}[\cite{HPSSTWBP2020}]\label{def:GFM}
		An NBW $\N$ is good-for-MDPs (GFM) if, for all finite MDPs $\M$ with initial state $s_0$, $\mathrm{PSyn}_{\N}^{\M}(s_0) = \mathrm{PSem}_{\N}^{\M}(s_0)$ holds.  
    \end{defi}
    
	\section{{\sf PSPACE}-Hardness}\label{section:PSPACE-hardness}
	We show that the problem of checking whether or not a given NBW is GFM is {\sf PSPACE}-hard.
	The reduction is from the NFA universality problem, which is known to be {\sf PSPACE}-complete~\cite{StockmeyerM73}. Given an NFA $\A$, the NFA universality problem asks whether $\A$ accepts every string, that is, whether $\lang(\A) = \Sigma^{*}$. 
	We additionally require the NFA $\A$ to be non-empty. 
	We first give an overview of how this reduction works. Given a non-empty complete NFA $\A$\rev{, an NFA whose language is non-empty,} we construct an NBW $\A_f$ (\cref{def:NFA-to-NBA}) which can be shown to be GFM (\cref{lemma:NBW-GFM}). Using this NBW $\A_f$, we then construct another NBW $\fork(\A_f)$ (\cref{def:fork-NBA}). We complete the argument by showing in \cref{lemma:reduction-from-NFA-universality} that the NBW $\fork(\A_f)$ is GFM if, and only if, $\A$ accepts the universal language.
	
	We start with the small observation that `for all finite MDPs' in \cref{def:GFM} 
	can be replaced by `for all finite MCs'.
	\begin{thm}\label{theorem:MC-MDP-interchangable}
		An NBW $\N$ is GFM iff, for all finite MCs $\M$ with initial state $s_0$, $\mathrm{PSyn}_{\N}^{\M}(s_0) = \mathrm{PSem}_{\N}^{\M}(s_0)$ holds.
	\end{thm}
	
	\begin{proof}
		\textbf{`if':} \rev{We show that if the syntactic and semantic acceptance probabilities of $\N$ coincide on every finite MC, they also coincide on every finite MDP $\M$.		Since $\N$ is an NBW, it has a language-equivalent deterministic parity word automaton (DPW) $\mathcal P$ \cite{Piterman07}; by determinism, a word is accepted by $\mathcal P$ if, and only if, it is accepted by $\N$.		The product $\M \times \mathcal P$ is then a finite MDP equipped with $\mathcal P$'s parity objective, for which an optimal \emph{memoryless} strategy $\mu$ is guaranteed to exist \cite{BA1995}; applying $\mu$ turns $\M \times \mathcal P$ into an MC $\M' = (\M \times \mathcal P)_\mu$.		Although $\mu$ is memoryless on the product, it tracks the state of $\mathcal P$ as it goes, so it corresponds to a finite-memory strategy for $\N$ on $\M$ itself, using $\mathcal P$'s states as memory; and since $\mu$ is optimal for $\mathcal P$'s parity objective on $\M \times \mathcal P$, this finite-memory strategy is optimal for $\N$ on $\M$ too, giving $\mathrm{PSyn}_{\N}^{\M}(s_0) = \mathrm{PSyn}_{\N}^{\M'}(s_0)$.		As $\M'$ is a finite MC, our assumption yields $\mathrm{PSyn}_{\N}^{\M'}(s_0) = \mathrm{PSem}_{\N}^{\M'}(s_0)$; and since $\mathcal P$ exactly captures $\N$'s acceptance, $\mathrm{PSem}_{\N}^{\M'}(s_0) = \mathrm{PSem}_{\N}^{\M}(s_0)$ as well.		Chaining these equalities gives $\mathrm{PSyn}_{\N}^{\M}(s_0) = \mathrm{PSem}_{\N}^{\M}(s_0)$, as required.}
		
		\textbf{`only if':} MCs are just special cases of MDPs.
	\end{proof}  

	\begin{figure}[t]
		\begin{minipage}{0.45\textwidth}
			\centering
			\begin{tikzpicture} [->, node distance = 2cm, auto,initial text = {}]
				\node[initial left, accepting,state] (q0) at (1,2) {$q_0$};
				
				\path (q0) edge [loop right] node {$a,b$} (q0);
				\node at (1, .8) {(a) A universal NFA $\B$};
			\end{tikzpicture}
		\end{minipage}
		\begin{minipage}{0.45\textwidth}
			\centering
			\begin{tikzpicture} [->, node distance = 2cm, auto,initial text = {}]
				\node[initial left,state] (q0) at (0,2) {$q_0'$};
				\node[accepting, right of = q0, state] (f)  {$f'$};
				
				\path (q0) edge [loop above] node {$a,b$} (q0)
				(f) edge [loop above, color = red] node {$\$$} (f)
				(q0) edge [color = red] node {$\$$} (f)
				(f) edge [bend left, color = red] node {$a,b$} (q0);
				\node at (1, .8) {(b) An NBW $\B_f$};
			\end{tikzpicture}
		\end{minipage}
		\caption{(a) $\B$ is a complete universal NFA. Let $\Sigma = \{a, b\}$. We have $\lang(\B)=\Sigma^*$. (b) On the right is the corresponding complete NBW $\B_f$. The new final state $f$ and the added transitions are highlighted in red. We have $\lang(\B_f)= \{ w_1\$w_2\$w_3\$\ldots\}$ where  $w_i \in \Sigma^*$.  }
		\label{fig:B-and-Bf}
	\end{figure}
	
	Given a complete NFA $\A$, we construct an NBW $\A_f$ by introducing a new letter $\$$ and a new state. As an example, given an NFA (DFA) $\B$ in \cref{fig:B-and-Bf}(a), we obtain an NBW $\B_{f}$ in \cref{fig:B-and-Bf}(b). It is easy to see that $\lang(\B) = \Sigma^{*}$ where $\Sigma = \{a, b\}$.
	
	\begin{defi}\label{def:NFA-to-NBA}
		Given a complete NFA $\A = (\Sigma,Q,q_0,\delta,F)$, we define the 
		NBW $\A_f = (\Sigma_\$,Q_f,q_0,\delta_f,\{f\})$ with
		$\Sigma_\$ = \Sigma \cup \{\$\}$ and $Q_f = Q \cup \{f\}$ for a fresh letter $\$\notin\Sigma$ and a fresh state $f\notin Q$, and with
		$\delta_f(q,\sigma) = \delta(q,\sigma)$ for all $q\in Q$ and $\sigma \in \Sigma$,
		$\delta_f(q,\$) = \{f\}$ for all $q \in F$, 
		$\delta_f(q,\$) = \{q_0\}$ for all $q \in Q\setminus F$, and
		$\delta_f(f,\sigma) = \delta_f(q_0, \sigma)$ for all $\sigma \in \Sigma_{\$}$.
	\end{defi}
	
	
	The language of $\A_f$ consists of all words of the form $w_1\$w_1'\$w_2\$w_2'\$w_3\$w_3'\$\ldots$ such that, for all $i\in \nat$, $w_i \in {\Sigma_\$}^*$ and $w_i' \in \lang(\A)$. 
	This provides the following lemma.
	
	\begin{lem}\label{lemma:NFA-languageInclusion}
		Given two NFAs $\A$ and $\B$, $\lang(\B)\subseteq \lang(\A)$ if, and only if, $\lang(\B_f)\subseteq \lang(\A_f)$. 
	\end{lem} 
	
	The following lemma simply states that the automaton $\A_f$ from the above construction is GFM. This lemma is technical and is key to prove~\cref{lemma:reduction-from-NFA-universality}, the main lemma, of this section. 
    
	\begin{lem}\label{lemma:NBW-GFM}
		For every complete NFA $\A$, $\A_f$ is GFM.
	\end{lem}
	
	\begin{proof}
		Consider an arbitrary MC $\M$ with initial state $s_0$. 
        Due to \cref{theorem:MC-MDP-interchangable}, we only need to show that $\A_f$ is good for $\M$, that is, $\mathrm{PSem}_{\A_f}^{\M}(s_0) = \mathrm{PSyn}_{\A_f}^{\M}(s_0)$. It suffices to show $\mathrm{PSyn}_{\A_f}^{\M}(s_0) \ge \mathrm{PSem}_{\A_f}^{\M}(s_0)$ since by definition the converse $\mathrm{PSem}_{\A_f}^{\M}(s_0) \ge \mathrm{PSyn}_{\A_f}^{\M}(s_0)$ always holds.
		
		First, we construct a language equivalent deterministic B\"uchi automaton (DBW) $\D_f$ by first determinising the NFA $\A$ to a DFA $\D$ by a standard subset construction and then obtain $\D_f$ by \cref{def:NFA-to-NBA}. Since $\lang(\A_f) = \lang(\D_f)$, we have that $\mathrm{PSem}_{\A_f}^{\M}(s_0) = \mathrm{PSem}_{\D_f}^{\M}(s_0)$. In addition, since $\D_f$ is deterministic, we have $\mathrm{PSem}_{\D_f}^{\M}(s_0) = \mathrm{PSyn}_{\D_f}^{\M}(s_0)$.
        
		It remains to show $\mathrm{PSyn}_{\A_f}^{\M}(s_0) \ge \mathrm{PSyn}_{\D_f}^{\M}(s_0)$. 
        We consider the MC $\M' = \M \times \D_f$; naturally, $\M$ and $\M'$ have the same probability of success,
        that is, $\mathrm{PSyn}_{\A_f}^{\M}(s_0) = \mathrm{PSyn}_{\A_f}^{\M'}(s_0')$, where $s_0'$ is the initial state of $\M'$.

        It suffices to show that there exists a strategy for~$\M' \times \A_f$ such that, for every accepting run~$r$ of~$\M \times \D_f = \M'$, the corresponding run induced by this strategy in~$\M' \times \A_f$ is accepting.
  
        \rev{Consider an accepting run of $\M' = \M \times \D_f$. Since $\M'$ is a finite MDP, the set of state-action pairs that any infinite run visits infinitely often always forms an end-component (this holds with certainty for the run under consideration, and is a purely graph-theoretic fact, not a probabilistic one); as we are considering a run that is accepting, i.e., one that visits an accepting state infinitely often, this end-component is necessarily an accepting one. 
        Before entering an accepting end-component of $\M'$, any strategy to resolve the nondeterminism in $\M' \times \A_f$ (thus $\A_f$) can be used.
		This will not block $\A_f$, as it is a complete automaton.
        Moreover, reading a $\$$ always takes $\A_f$ from its current state to either $f$ (if that state is accepting in $\A$) or back to $q_0$ (otherwise); and since $\delta_f(f,\sigma) = \delta_f(q_0,\sigma)$ for every $\sigma \in \Sigma_\$$, $\A_f$ behaves exactly as if it were at $q_0$ immediately after any $\$$, regardless of which of the two states it actually lands in.}
        
		Once an accepting end-component of $\M'$ is entered,
		there must exist a word of the form $\$w\$$, where $w \in \lang(\D)$ (and thus $w \in \lang(\A)$), which is a possible initial sequence of letters produced from some state $\rev{u}$ of $\M'$ in this end-component.
		We fix such a word $\$w\$$; such a state $\rev{u}$ of the end-component in $\M'$ from which this word $\$w\$$ can be produced; and strategy of the NBW $\A_f$ to follow the word $w\$$ from $q_0$ (and $f$) to the accepting state $f$.
		(Note that the first $\$$ always leads to $q_0$ or $f$.) 
		
		In an accepting end-component we can be in two modes: \emph{tracking}, or \emph{not tracking}. \rev{Intuitively, to \emph{track} a word means that the GFM strategy tries to follow the trace of this word in the end-component: it steers $\A_f$ along exactly one sequence of states that reading the word from $q_0$ (or $f$) would induce, for as long as the model keeps producing the letters of that word in order.}
		If we are \emph{not tracking} and reach $\rev{u}$, we start to \emph{track} $\$w\$$:
		we use $\A_f$ to reach an accepting state after reading $\$w\$$ (ignoring what happens in any other case) with the strategy we have fixed. Note that after reading the first $\$$, we are in either $q_0$ or $f$, such that when starting from $u$, it is always possible, with a fixed probability $p>0$, to read $\$w\$$ and thus accept.
		If we read an unexpected letter (where the `expected' letter is always the next letter from $\$w\$$) or the end of the word $\$w\$$ is reached, we move to \emph{not tracking}.
		
		The choices of the automata when \emph{not tracking} can be resolved arbitrarily.
		
		Once in an accepting end-component of $\M'$, tracking is almost surely started infinitely often, and it is thus almost surely successful infinitely often.
		Thus, we have
		$\mathrm{PSyn}_{\A_f}^{\M}(s_0) = \mathrm{PSyn}_{\A_f}^{\M'}(s_0')\geq \mathrm{PSyn}_{\D_f}^{\M}(s_0)$.
	\end{proof}

\rev{\noindent\textbf{Remark.} An NBW $\N = (\Sigma,Q,q_0,\delta,F)$ is \emph{residual}, or \emph{semantically deterministic}, if every one of its transitions is \emph{language-preserving}: for every $q \in Q$, $\sigma \in \Sigma$, and $q' \in \delta(q,\sigma)$, we have $\lang(\N_{q'}) = \sigma^{-1}\lang(\N_q)$. 
This is exactly the notion of a \emph{residual transition} we use in \cref{ssec:decidingGFM}, applied to every transition of $\N$ at once.
The lemma above is an instance of a more general fact: every residual, or semantically-deterministic, B\"uchi automaton is GFM \cite[Theorem~6.4]{DBLP:journals/lmcs/RadiKL25}, and $\A_f$ is one such automaton.}

	
	Let $\B$ be a universal NFA in \cref{fig:B-and-Bf}(a) and $\B_f = (\Sigma_\$,Q_f',q_0',\delta_f',\{f'\})$ be the NBW in \cref{fig:B-and-Bf}(b). Assume without loss of generality that the state \rev{spaces} of $\A_f$, $\B_f$, and $\{q_{F}, q_{\A}, q_{\B}\}$ are pairwise disjoint. We now define the $\fork$ operation. An example of how to construct an NBW $\fork(\A_f)$ is shown in \cref{fig:fork}.
	
	\begin{defi}\label{def:fork-NBA}
		Given an NBW $\A_f = (\Sigma_\$,Q_f,q_0,\delta_f,\{f\})$, we define the NBW $\fork(\A_{f}) = (\Sigma_\$, Q_F, q_{F}, \delta_F,\{f,f'\})$ with 	
		\begin{itemize}
			\item $Q_F = Q_f \cup Q_{f'}'\cup\{q_{F}, q_{\A}, q_{\B}\}$;
			\item $\delta_F (q,\sigma) = \delta_f(q,\sigma)$ for all $q \in Q_f$ and $\sigma \in \Sigma_\$$;
			\item $\delta_F (q,\sigma) = \delta_{f'}(q,\sigma)$ for all $q \in Q_{f'}'$ and $\sigma \in \Sigma_\$$;
			\item $\delta_F (q_F,\sigma) = \{q_{\A},q_{\B}\}$ for all $\sigma \in \Sigma_\$$;
			\item $\delta_F (q_{\A},\sigma) = \{q_0\}$ for all $\sigma \in \Sigma_\$$;
			\item  $\delta_F (q_{\B},\$) = \{q_0'\}$, and
			$\delta_F (q_{\B},\sigma) = \emptyset$ for all $\sigma \in \Sigma$.
		\end{itemize}
	\end{defi}

	\begin{figure}[t]
		\begin{minipage}[c]{0.45\linewidth}
			\centering
			\begin{tikzpicture} [->, node distance = 2cm, auto]
				\usetikzlibrary{positioning,arrows,automata}
				\node[state] (s) at (2.5,2) {$q_{F}$};
				\node[state] (sa) at (1,0) {$q_{\A}$};
				\node[state] (sb) at (4,0) {$q_{\B}$};
				\node[state] (aq0) at (1,-2) {$q_0$};
				\node[state] (bq0) at (4,-2) {$q_0'$};
				
				\node[accepting, right of = bq0, state] (f)  {$f'$};
				
				\path (bq0) edge [loop below] node {$a,b$} (bq0)
				(f) edge [loop below, color = red] node {$\$$} (f)
				(bq0) edge [color = red] node {$\$$} (f)
				(f) edge [bend left, color = red] node {$a,b$} (bq0);       
				\node[rectangle,draw,dashed, minimum width = 2cm, minimum height = 2.2cm] (r1) at (1,-2.5) {};
				\node  at (1, -4) {$\A_f$};
				\node[rectangle,draw,dashed, minimum width = 3.2cm, minimum height = 2.2cm] (r2) at (5,-2.5) {};
				\node at (5, -4) {$\B_f$};
				
				\draw (s) edge [left]  node {$a,b,\$$} (sa)
				(s) edge  [right] node {$a,b,\$$} (sb)
				(sa) edge  [left] node  {$a,b,\$$} (aq0)
				(sb) edge  node {$\$$}(4,-1.5) ;
			\end{tikzpicture}
			\caption{Given a non-empty complete NFA $\A$, an NBW $\A_f$ and an NBW $\fork(\A_f)$ are constructed. In this example, $\Sigma = \{a, b\}$ and $\Sigma_{\$} = \{a, b, \$\}$. From the state $q_{\A}$ (resp. $q_{\B}$), the NBW $\fork(\A_f)$ transitions to the initial state of $\A_f$ (resp. $\B_f$).}
			\label{fig:fork}
		\end{minipage}\hfill
		\begin{minipage}[c]{0.5\linewidth}
			\centering
			\begin{tikzpicture} [->, node distance = 2cm, auto]
				\usetikzlibrary{positioning,arrows,automata}
				\tikzstyle{BoxStyle} = [draw, circle, fill=black, scale=0.4,minimum width = 1pt, minimum height = 1pt]
				
				\node[state,fill=orange!40] (s) at (2.5,2) {$s_0$};
				\node[BoxStyle] (bs) at (2.5, 1) {};
				\node[state,fill=orange!40] (sa) at (0.5,0) {$s_1$};
				\node[BoxStyle] (bsa) at (0.5, -1) {};
				\node[state,fill=blue!30] (sb) at (4.5,0) {$s_2$};
				\node[BoxStyle] (bsb) at (4.5, -1) {};
				\node[rectangle,draw,dashed, minimum width = 1.5cm, minimum height = 2cm] (r1) at (0.5,-2.5) {};
				\node at (1, -4.4) {\parbox[t]{1in}{generate $w_a{\cdot}\$$ repeatedly with probability one}};
				\node[rectangle,draw,dashed, minimum width = 1.5cm, minimum height = 2cm] (r2) at (4.5,-2.5) {};
				\node at (4.3, -4.4) {\parbox[t]{1in}{generate $w_b{\cdot}\$$ repeatedly with probability one}};
				
				\path [-] (s) edge [above] node {} (bs);
				
				\path [->] (bs) edge [above] node {$\frac{1}{2}$} (sa)
				(bs) edge [above] node {$\frac{1}{2}$} (sb)
				(sa) edge [left] node {$\frac{1}{2}$} (r1)
				(bsa) edge [above] node {$\frac{1}{2}$} (r2)
				(sb) edge [right] node {$\frac{1}{2}$} (r2)
				(bsb) edge [above] node {$\frac{1}{2}$} (r1);
				
			\end{tikzpicture}
			\caption{An example MC in the proof of \cref{lemma:reduction-from-NFA-universality}. Assume $\Sigma_{\$} = \{a, b, \$\}$. The labelling of the MC is as follows: $\ell(s_0) = \ell(s_1) = a$ and $\ell(s_2) = \$$.}
			\label{fig:fork-MC}            
		\end{minipage}
	\end{figure}
	
	Following from \cref{lemma:NFA-languageInclusion} and \cref{lemma:NBW-GFM}, we have:
	
	\begin{lem}\label{lemma:reduction-from-NFA-universality}
		The NBW $\fork(\mathcal A_f)$ is GFM if, and only if, $\lang(\A) = \Sigma^{*}$.
	\end{lem}
	
	
	\begin{proof}
        \rev{Recall that $\B$ is the simple one-state universal NFA and $\B_f$ is its corresponding complete NBW, shown in \cref{fig:B-and-Bf}(a) and (b) respectively.}
		We first observe that $\lang\big(\fork(\A_f)\big) = \{\sigma \sigma' w \mid \sigma,\sigma' \in \Sigma_\$,\ w\in \lang(\A_f)\} \cup 
		\{\sigma \$ w \mid \sigma \in \Sigma_\$,\ w\in \lang(\B_f)\}$.

		\textbf{`if':} 
		When $\lang(\A) = \Sigma^{*} = \lang(\B)$ holds, \cref{lemma:NFA-languageInclusion} provides
		$\{\sigma \sigma' w \mid \sigma,\sigma' \in \Sigma_\$,\ w\in \lang(\A_f)\} \supset 
		\{\sigma \$ w \mid \sigma \in \Sigma_\$,\ w\in \lang(\B_f)\}$, and therefore
		$\lang(\fork(\A_f)) = \{\sigma \sigma' w \mid \sigma,\sigma' \in \Sigma_\$,\ w\in \lang(\A_f)\}$.
		
		As $\A_f$ is GFM by \cref{lemma:NBW-GFM}, this provides the GFM strategy ``move first to $q_{\A}$, then to $q_0$, and henceforth follow the GFM strategy of $\A_f$ for $\fork(\A_f)$''.
		Thus, $\fork(\A_f)$ is GFM in this case.
		
		\textbf{`only if':} Assume $\lang(\A)\neq \Sigma^{*} = \lang(\B)$, that is, $\lang(\A) \subset \lang(\B)$. There must exist words $w_a \in \lang(\A)$ and $w_b\in \lang(\B) \setminus \lang(\A)$. We now construct an MC which witnesses that $\fork(\A_f)$ is not GFM. 
		
		The MC emits an $a$ at the first step and then either an $a$ or a $\$$ with a chance of $\frac{1}{2}$ at the second step. An example is provided in \cref{fig:fork-MC}. 
		
		After these two letters, it then moves to one of two cycles (independent of the first two chosen letters) with equal chance of $\frac{1}{2}$;
		one of these cycles repeats a word $w_a {\cdot} \$$ infinitely often, while the other repeats a word $w_b {\cdot} \$$ infinitely often.
		
		It is easy to see that the semantic chance of acceptance is $\frac{3}{4}$  -- failing if, and only if, the second letter is $a$ and the word $w_b\$$ is subsequently repeated infinitely often -- whereas the syntactic chance of satisfaction is $\frac{1}{2}$:
		when the automaton first moves to $q_{\A}$, it accepts if, and only if, the word $w_a\$$ is later repeated infinitely often, which happens with a chance of $\frac{1}{2}$; when the automaton first moves to $q_{\B}$, it will reject when $\$$ is not the second letter, which happens with a chance of $\frac{1}{2}$. 
	\end{proof}
	
	It follows from \cref{lemma:reduction-from-NFA-universality} that the NFA universality problem is polynomial-time 
	reducible to the problem of whether or not a given NBW is GFM.
	\begin{thm}\label{theorem:PSPACE-harness}
		The problem of whether or not a given NBW is GFM is {\sf PSPACE}-hard.
	\end{thm}
	
	Using the same construction of \cref{def:NFA-to-NBA}, we can show that the problem of minimising a GFM NBW is {\sf PSPACE}-hard. 

       \begin{figure}[htb]
            \centering
            \begin{tikzpicture}[xscale=.7,shorten >=1pt,node distance = 1cm, on grid, auto]
            \usetikzlibrary{positioning,arrows,automata}
            
            \node[state] (A) at (0,0) {$q_0$};
            \node[state, accepting] (B) [right of = A]  at (3,0) {$f$};
            \path[->] (A) edge [out=30,in=150,looseness=1] node [midway, above]{$\$$} (B);
            \path[->] (B) edge [out=-150,in=-30,looseness=1] node [midway, above]{$\Sigma$} (A);    
            \path[->] (A) edge [loop left] node [midway, left] {$\Sigma$}  (A);
            \path[->,dashed] (A) edge [loop below] node [midway, below] {$\$ \,$ \textcolor{red}{\textcircled{\raisebox{-0.9pt}{3}}}}  (A);
            \path[->, dashed] (B) edge [out=-120,in=-60,looseness=1] node [midway, below]{$\$\,$ \textcolor{red}{\textcircled{\raisebox{-0.9pt}{2}}}} (A);     
            \path[->,dashed] (B) edge [loop right] node [midway, right] {$\$\,$ \textcolor{red}{\textcircled{\raisebox{-0.9pt}{1}}}}  (B);
            \end{tikzpicture}
            \caption{ There are six language equivalent NBWs with two states. All six NBWs have three common transitions (shown as the solid paths). Each NBW has different transitions shown as the dashed paths: \textcircled{\raisebox{-0.9pt}{1}}, \textcircled{\raisebox{-0.9pt}{1}}\textcircled{\raisebox{-0.9pt}{2}}, 
            \textcircled{\raisebox{-0.9pt}{1}}\textcircled{\raisebox{-0.9pt}{3}}, 
            \textcircled{\raisebox{-0.9pt}{1}}\textcircled{\raisebox{-0.9pt}{2}}\textcircled{\raisebox{-0.9pt}{3}}, 
            \textcircled{\raisebox{-0.9pt}{2}},
            \textcircled{\raisebox{-0.9pt}{2}}\textcircled{\raisebox{-0.9pt}{3}}. 
            For example, \textcircled{\raisebox{-0.9pt}{1}} represents the NBW with the dashed transition \textcircled{\raisebox{-0.9pt}{1}} only, and \textcircled{\raisebox{-0.9pt}{2}}\textcircled{\raisebox{-0.9pt}{3}} represents the NBW with both dashed transitions \textcircled{\raisebox{-0.9pt}{2}} and \textcircled{\raisebox{-0.9pt}{3}}.  } 
            \label{fig:reduction-NFA-universality}
        \end{figure}
 
	\begin{restatable}{thm}{theoremMinimisation}\label{theorem:minimisation-GFM}
		Given a GFM NBW and a bound $k$, the problem whether there is a language equivalent GFM NBW with at most $k$ states is {\sf PSPACE}-hard. It is {\sf PSPACE}-hard even for (fixed) $k = 2$.
	\end{restatable}
	\begin{proof}
		Using the construction of \cref{def:NFA-to-NBA}, {\sf PSPACE}-hardness follows from a reduction from the problem \rev{of} whether a non-empty complete NFA is universal. 

       We show that for any non-empty complete NFA~$\A$, a minimal \rev{GFM} NBW equivalent to~$\A_f$ has exactly two states if and only if~$\A$ is universal.
        To this end, let $\A$ be a non-empty complete NFA.
        The following observations hold.
        
        \begin{itemize}
        \item
        Any \rev{GFM} NBW equivalent to $\A_f$ must have at least two states: one accepting and one non-accepting.
        Indeed, since $\A$ is non-empty, $\A_f$ accepts some infinite word and therefore requires an accepting state.
        At the same time, words containing only finitely many occurrences of $\$$ must be rejected\rev{;} we therefore require a non-accepting state.
        
        \item
        Consider a two-state \rev{GFM} NBW equivalent to $\A_f$.
        There cannot exist a word $w \in \Sigma^{+}$ that has a run from the accepting state back to itself.
        Otherwise, this would yield an accepting run that visits the accepting state infinitely often while containing only finitely many occurrences of $\$$, contradicting the definition of $\A_f$.
        
        Hence, starting from the accepting state, any finite word over $\Sigma^{+}$ must lead to the non-accepting state and remain there.
        Blocking is not possible, since every accepting prefix admits an accepting infinite continuation.
        It follows that either state transitions to the non-accepting state on every letter in $\Sigma$.
        
        Furthermore, for a word beginning with a letter in $\Sigma$ to be accepted, there must be a transition on $\$$ from the non-accepting state to the accepting state.
        
        \item
        Finally, in any two-state \rev{GFM} NBW equivalent to $\A_f$, the accepting state must have some transitions on $\$$, either to itself or to the non-accepting state.
        Indeed, any run that reaches the accepting state must not block upon reading $\$$.
        
        \end{itemize}

        Taken together, these constraints imply that any two-state \rev{GFM} NBW equivalent to $\A_f$ must coincide with one of the six automata shown in \cref{fig:reduction-NFA-universality}.
        Moreover, all these automata recognise the same language as $\A_f$ precisely when $\A$ is universal.
	\end{proof}

	\section{Decision Procedure for Qualitative GFM}\label{section:QGFM-procedure}
	In this section, we first define the notion of qualitative GFM (QGFM) and then provide an {\sf EXPTIME} procedure that decides QGFM-ness.
	
	The definition of QGFM is similar to GFM but we only need to consider MCs with which the semantic chance of success is one:
	\begin{defi}\label{def:QGFM}
		An NBW $\N$ is qualitative good-for-MDPs (QGFM) if, for all finite MDPs $\M$ with initial state $s_0$ and $\mathrm{PSem}_{\N}^{\M}(s_0) = 1$, $\mathrm{PSyn}_{\N}^{\M}(s_0) = 1$ holds.  
	\end{defi}
	
	Similar to \cref{theorem:MC-MDP-interchangable}, we can also replace `MDPs' by `MCs' in the definition of QGFM:
	\begin{thm}\label{theorem:MC-MDP-interchangable-QGFM}
		An NBW $\N$ is QGFM iff, for all finite MCs $\M$ with initial state $s_0$ and $\mathrm{PSem}_{\N}^{\M}(s_0) = 1$, $\mathrm{PSyn}_{\N}^{\M}(s_0) = 1$ holds.
	\end{thm}
	
	To decide QGFM-ness, we make use of the well known fact that qualitative acceptance, such as $\mathrm{PSem}_{\N}^{\M}(s_0) = 1$,
	does not depend on the probabilities for an MC $\M$. This can, for example, be seen by considering the syntactic product with a DPW $\mathcal D$, as $\mathrm{PSem}_{\mathcal D}^{\M}(s_0) = \mathrm{PSyn}_{\mathcal D}^{\M}(s_0)$ trivially holds for a deterministic automaton $\D$,
	where changing the probabilities does not change the end-components of the product MC $\M \times \D$, and the acceptance of these end-components is solely determined by the highest colour of the states (or transitions) occurring in it, and thus also independent of the probabilities:
	the probability is one if, and only if, an accepting end-component can be reached almost surely, which is also independent of the probabilities. As a result, we can search for the (regular) tree that stems from the unravelling of an MC, while disregarding probabilities. See \cref{fig:MC-tree} for an example of such a tree.

    \begin{figure}[htb]
		\centering
		\resizebox{!}{3.5cm}{%
		\begin{tikzpicture} [xscale=.7,shorten >=1pt,node distance = 1cm, on grid, auto]
					\usetikzlibrary{positioning,arrows,automata}
					\tikzstyle{BoxStyle} = [draw, circle, fill=black, scale=0.4,minimum width = 1pt, minimum height = 1pt]
					
					\node[state,fill=orange!40] (s0) at (0,0) {$s_0$};
					\node[state,fill=orange!40] (s1) at (-2,-1) {$s_1$};
					\node[state,fill=green!30] (s2) at (2,-1) {$s_2$};
					\node[state,fill=orange!40] (s12) at (-2,-2.1) {$s_1$};
					\node[state,fill=green!30] (s22) at (2,-2.1) {$s_2$};
					\node[label] at (-2,-2.7) {$\vdots$};
					\node[label] at (2,-2.7) {$\vdots$};
					
					\path [->] (s0) edge [above, midway] node {} (s1)
					(s0) edge [above, midway] node {} (s2)
					(s1) edge [left] node {} (s12)
					(s2) edge [right] node {} (s22);
				\end{tikzpicture}
			}
		\caption{The tree that stems from unravelling of the MC with initial state $s_0$ in \cref{fig:MDP}, while disregarding probabilities.}
		\label{fig:MC-tree}
	\end{figure}

	This observation has been used in the synthesis of probabilistic systems before \cite{Schewe2006}.
	The set of directions (of a tree) $\Upsilon$ could then, for example, be chosen to be the set of states of the unravelled finite MC; this would not normally be a full tree.

	\rev{This shift from MCs to trees is what lets us use \emph{tree automata} to decide QGFM-ness. A tree automaton reads a (possibly infinite) branching structure and accepts or rejects it, much like a word automaton does for a single infinite sequence; we recall trees and tree automata formally in \cref{app:def-tree-automata} below. The plan is to construct, for a candidate NBW $\C$, a tree automaton that accepts exactly the trees on which the GFM automaton $\G$ succeeds semantically -- i.e., with probability one -- and a second tree automaton that accepts exactly the trees on which $\C$ itself \emph{fails} to succeed syntactically with probability one. By definition, $\C$ is QGFM precisely when no tree is accepted by both automata, so deciding QGFM-ness reduces to a language-intersection emptiness check between two tree automata. Turning this idea into an actual algorithm, however, takes several further transformations of these automata, which is why this problem takes up the remainder of this section.
	
	We proceed as follows. \Cref{app:def-tree-automata} recalls the necessary background on trees and tree automata, including the \emph{symmetric} automata that let us describe transitions without committing to concrete directions. \Cref{ssec:QGFM-overview} then gives a high-level overview of the pipeline of constructions -- summarised in \cref{fig:algorithms} -- that reduces QGFM-ness to the emptiness check sketched above, and states the resulting {\sf EXPTIME} upper bound. Finally, \cref{app:constructions} carries out this pipeline step by step, together with the intuition behind each construction, and proves the {\sf EXPTIME} bound.}

\subsection{Tree Automata}\label{app:def-tree-automata}
\paragraph{\textbf{Trees.}}
Consider a finite set $\Upsilon = \{\upsilon_1, \ldots, \upsilon_n\}$ of directions. An $\Upsilon$-tree is a prefixed closed set $T \subseteq \Upsilon^{*}$.
A \emph{run} or \emph{path} $r$ of a tree is a set $r \subseteq T$ such that $\epsilon \in T$ and, for every $w \in r$, there exists a unique $\upsilon \in \Upsilon$ such that $w {\cdot} \upsilon \in r$.
For a node $w \in T$, $\suc(w) \subseteq \Upsilon$ denotes exactly those directions $\upsilon$, such that $w {\cdot} \upsilon \in T$.
A tree is called \emph{closed} if it is non-empty (i.e.\ if it contains the empty word $\varepsilon$) and every node of the tree has at least one successor.

Given a finite alphabet $\Sigma$, a $\Sigma$-labelled $\Upsilon$-tree is a pair $\langle T, V \rangle$ where $T$ is a closed $\Upsilon$-tree and $V:T \to \Sigma$ maps each node of $T$ to a letter in $\Sigma$.

An $\Upsilon$-tree $T$ and a labelled $\Upsilon$-tree $\langle T, V\rangle$ are called \emph{full} if $T=\Upsilon^*$.
\smallskip

\paragraph{\textbf{Tree automata.}}
The infinite tree automata in this paper use transition-based acceptance conditions.
An alternating parity tree automaton is a tuple $\A = (\Sigma, Q, q_0, \delta, \alpha)$, where $Q$ denotes a finite set of states, $q_0 \in Q$ denotes an initial state, $\delta: Q \times \Sigma \to \mathbb{B}^{+}\big( \Upsilon \times Q \times C \big)$ denotes,
for a finite set $C \subset \nat$ of colours, a transition function and the acceptance condition $\alpha$ is a transition-based parity condition.
\rev{$\mathbb{B}^{+}\big( \Upsilon \times Q \times C \big)$ are positive Boolean formulas: formulas built from elements in $\Upsilon \times Q \times C$ using $\land$ and $\lor$. 
For a set $S \subseteq \Upsilon \times Q \times C$ and a formula $\theta \in \mathbb{B}^{+}\big( \Upsilon \times Q \times C \big)$, we say that $S$ satisfies $\theta$ iff assigning \textbf{true} to elements in $S$ and assigning \textbf{false} to elements not in $S$ makes $\theta$ true.
For technical convenience, we use \emph{complete} automata, i.e.\ $\mathbb B^+$ does not contain \emph{true} or \emph{false}, so that our run trees are closed.}
An alternating automaton runs on $\Sigma$-labelled $\Upsilon$-trees.
The acceptance mechanism is defined in terms of run trees.

A run of an alternating automaton $\A$ on a $\Sigma$-labelled $\Upsilon$-tree $\langle T,V \rangle $ is a tree $ \langle T_r, r \rangle$ in which each node is labelled with an element of $T \times Q$. Unlike $T$, in which each node has at most $|\Upsilon|$ children, the tree $T_r$ may have nodes with many children and may also have leaves (nodes with no children). Thus, $T_r \subset \nat^{*}$ and a path in $T_r$ may be either finite, in which case it ends in a leaf, or infinite. Formally, $\langle T_r, r \rangle$ is the $(T \times Q)$-labelled tree such that:
\begin{itemize}
	\item $\epsilon \in T_{r}$ and $r(\epsilon) = (\epsilon, q_0)$;
	\item Consider a node $w_r \in T_r$ with $r(w_r) = (w, q)$ and $\theta = \delta\big(q, V(w)\big)$. Then there is a (possibly empty) set $S \subseteq \Upsilon \times Q \times C$ such that $S$ satisfies $\theta$, and for all $(c, q', i) \in S$, the following hold: 
	if $c \in \suc(w)$, then there is a $j \in \nat$ such that $w_r {\cdot} j \in T_r$, $r(w_r {\cdot} j) = (w {\cdot} c, q')$ and the transition in $T_r$ between $w_r$ and $w_r {\cdot} j$ is assigned colour $i$.
\end{itemize}  
An infinite path of the run tree fulfils the parity condition, if the highest colour of the transitions (edges) appearing infinitely often on the path is even. A run tree is accepting if all infinite paths fulfil the parity condition. A tree $T$ is accepted by an automaton $\A$ if there is an accepting run tree of $\A$ over $T$. We denote by $\lang(\A)$ the language of the automaton $\A$; i.e., the set of all labelled trees that $\A$ accepts.

The acceptance of a tree can also be viewed as the outcome of
a game, where player {\it accept} chooses, for a pair $(q, \sigma) \in Q \times \Sigma$, a set of \rev{elements in $\Upsilon \times Q \times C$} satisfying $\delta(q, \sigma)$, and player {\it reject} chooses one of these \rev{elements}, which is executed. The tree is accepted iff player {\it accept} has a strategy enforcing a path that fulfils the parity condition. One of the players has a memoryless winning strategy, that is, a strategy where the moves only depend on the state of the automaton and the node of the tree, and, for player {\it reject}, on the choice of player {\it accept} in the same move.

A \emph{nondeterministic} tree automaton is a special alternating tree automaton, where the image of $\delta$ consists only of such formulas that, when rewritten in disjunctive normal form, contains exactly one element of $\{ \upsilon \} \times Q \times C$ for every $\upsilon \in \Upsilon$ in every disjunct.
A nondeterministic tree automaton is \emph{deterministic} if the image of $\delta$ consists only of such formulas that, when rewritten in disjunctive normal form, contain only one disjunct.
An automaton is \emph{universal} when the image of $\delta$ consists only of formulas that can be rewritten as conjunctions.

An automaton is
called a B\"uchi automaton if $C = \{1, 2\}$ and a co-B\"uchi automaton if $C = \{0, 1\}$.

\paragraph{\textbf{Symmetric Automata.}}
\rev{Since our trees arise from unravelling Markov chains, their direction set $\Upsilon$ is not fixed in advance and can be as large as the state space of the model; we therefore want automata whose transitions do not need to name concrete directions of $\Upsilon$, but only speak of successors in the abstract.} Symmetric automata are alternating automata that abstract from the concrete directions in exactly this way, replacing them by operators $\Omega = \{\Box, \Diamond, \boxtimes_i\}$, where $\Box$ and $\Diamond$ are standard operators that mean `send to all successors ($\Box$)' and `send to some successor ($\Diamond$)', respectively.

	

\rev{A symmetric (alternating) parity tree automaton is a tuple $\mathcal S = (\Sigma, Q, q_0, \delta, \alpha)$ that differs from ordinary alternating automata only in the transition function, which now maps into $\mathbb{B}^{+}\big( \Omega \times Q \times C \big)$ instead of $\mathbb{B}^{+}\big( \Upsilon \times Q \times C \big)$: $\delta: Q \times \Sigma \to \mathbb{B}^{+}\big( \Omega \times Q \times C \big)$. To define the run tree of $\mathcal S$ on a labelled tree $\langle T,V\rangle$ exactly as for ordinary alternating automata above, we first translate $\delta(q,\sigma)$, at each node $w$ of $T$, into an ordinary transition formula over $\Upsilon \times Q \times C$: every element of $\delta(q,\sigma)$ is syntactically substituted, using $w$'s actual (finite) set of successors $\suc(w) \subseteq \Upsilon$ in $T$, according to the following rules:}
\begin{itemize}	
\item \rev{$(\Box,q,i)$ is substituted by $\bigwedge\limits_{c \in \suc(w)}(c,q,i)$;	$(\Diamond,q,i)$ is substituted by $\bigvee\limits_{c \in \suc(w)}(c,q,i)$} 	
\item	\rev{$(\boxtimes_j,q,i)$ is substituted by $\bigvee\limits_{\upsilon \in \suc(w)}(\boxtimes_j^\upsilon,q,i)$, where, in turn,		each disjunct $(\boxtimes_j^\upsilon,q,i)$ is substituted by $(\upsilon,q,j) \wedge \bigwedge\limits_{\upsilon \neq c \in \suc(w)}(c,q,i)$.}
\end{itemize}
\rev{The resulting formula over $\Upsilon \times Q \times C$, specific to the node $w$, is then used exactly as $\delta\big(q,V(w)\big)$ was used above to define the run tree.}

An alternating parity tree automaton is \emph{nearly} symmetric if there is $\boxtimes_j^{\upsilon}$ in the transition function, that is, $\delta: Q \times \Sigma \to \mathbb{B}^{+}\big( (\Omega \cup \{\boxtimes_j^{\upsilon} \suchthat j \in C \land \upsilon \in \Upsilon\} ) \times Q \times C \big)$. 

While the $\boxtimes_j$ and $\boxtimes_j^v$ operators are less standard, they have a simple semantics: $(\boxtimes_j,q,i)$ sends $q$ in every direction (much like $\Box$), but uses different colours: $i$ in all directions but one, and $j$ into one (arbitrary) direction; $\boxtimes_j^v$ is similar, but fixes the direction to $v$ into which $q$ is sent with a different colour.

So, while nondeterministic tree automata send exactly one copy to each successor, symmetric automata can send several copies to the same successor.
On the other hand, symmetric automata cannot distinguish between left and right and can send copies to successor nodes only in either a universal or an existential manner.

\rev{Recall the three-letter naming convention for automata from \cref{subsection:preliminaries-automata}. For tree automata, the first letter can also be U (universal) or A (alternating), and the third letter is S if the automaton is symmetric, N if it is nearly symmetric, and T otherwise. For example, a DCT is a deterministic co-B\"uchi tree automaton.}

\subsection{Overview of the Decision Procedure}\label{ssec:QGFM-overview}
\rev{With the notions in hand, we now make the idea sketched precise, and outline the pipeline of constructions that turns it into an actual {\sf EXPTIME} algorithm. \Cref{fig:algorithms} summarises this pipeline.}
	
	\begin{figure*}[t]
		\centering
		\scalebox{1}{
			\begin{tikzpicture}[xscale=.7,>=latex',shorten >=1pt,node distance=3cm,on grid,auto]
				
				\node[label] (C) at (0.2,0) {$\C$};
				\node[label] (TC) at (1.8,0) {$\T_{\C}$};
				\node[label] (WTC) at (5.5,0) {$\widetilde{\T}_{\C}$};
				\node[label] (UC) at (8.5,0) {$\U_{\C}$};
				\node[label] (UCw) at (11.05,  0) {$\U_{\C}^{w}$};
				\node[label] (UCp) at (13.6, 0) {$\U_{\C}^{p}$};
				\node[label] (UCP) at (15.6, 0) {$\U_{\C}'$};
				
				\node[label] (DPC) at (19.3,0) {$\D_{\C}$};
				
				\path[->] (C) edge node [midway, above] {} (TC);
				\path[->] (C) edge node [midway, below] {} (TC);	
				\path[->] (TC) edge node [midway, above] {{ dualise \&}} (WTC);
				\path[->] (TC) edge node [midway, below] {{complement}} (WTC);	
				\path[->] (WTC) edge node [midway, above] {strategy} (UC);	
				\path[->] (WTC) edge node [midway, below] {explicit} (UC);				
				\path[->] (UC) edge node [midway, above] {widen} (UCw);	
				\path[->] (UCw) edge node [midway, above] {prune} (UCp);
				\path[->] (UCp) edge node [midway, above] {$\Upsilon$} (UCP);
				\path[->] (UCp) edge node [midway, below] {free} (UCP);
				
				\path[->, red] (UCP) edge node [midway, above] { determinise} (DPC);
				
				\node[label] (G) at (0.2,-2.5) {$\G$}; 
				\node[label] (NG) at (5.5,-2.5) {$\T_{\G}'$};
				\node[label] (NGD) at (8.5,-2.5) {$\D_{\G}$};
				\node[label] (DGw) at (11.05,-2.5) {$\D_{\G}^w$};
				\node[label] (DGp) at (13.6,-2.5) {$\D_{\G}^p$};
				
				\node[label] (DGP) at (19.3, -2.5) {$\D_{\G}'$};

				\path[->] (G) edge node [midway, above] {} (NG);
				\path[->] (NG) edge node [midway, above] {strategy} (NGD);	
				\path[->] (NG) edge node [midway, below] {explicit} (NGD);
				\path[->] (NGD) edge node [midway, above] {widen} (DGw);
				\path[->] (DGw) edge node [midway, above] {prune} (DGp);
				
				\path[->, red] (C) edge node [midway, right, xshift=0.0cm, yshift=0.2cm] {$\lang(\C) = \lang(\G)$,} (G);
				\path[->, red] (C) edge node [midway, right, xshift=0.0cm, yshift=-0.3cm] {$\G$ GFM} (G);
				\path[->, red] (C) edge node [midway, left] {} (G);
				
				\path[->] (DGp) edge node [midway, above] {$\Upsilon$} (DGP);
				\path[->] (DGp) edge node [midway, below] {free} (DGP);
				
				\path[<->, dashed, red] (WTC) edge node [midway, right, xshift=0.0cm, yshift=-0.3cm] {\color{gray}would be \sf{2-EXPTIME}} (NG);	
				\path[<->, dashed, red] (WTC) edge node [midway, right, xshift=0.0cm, yshift=0.2cm] {$\lang(\widetilde{\T}_{\C}) \cap \lang(\T_{\G}') \stackrel{?}{=} \emptyset$} (NG);	
				\path[<->, dashed] (DPC) edge node [midway, left, xshift=0.0cm, yshift=0.2cm] {$\lang(\D_{\C}) \cap \lang(\D_{\G}') \stackrel{?}{=} \emptyset$} (DGP);	
				\path[<->, dashed] (DPC) edge node [midway, left, xshift=-0.3cm, yshift=-0.3cm] {overall \sf{EXPTIME}} (DGP);

			\end{tikzpicture}
		}
		\caption{\label{fig:algorithms}%
			Flowchart of the algorithms. The first algorithm is in {\sf 2-EXPTIME}, which is to check the non-emptiness of the intersection of $\widetilde{\T}_{\C}$ and $\T_{\G}'$. The second algorithm is in {\sf EXPTIME}, which is to check the non-emptiness of the intersection of $\D_{\C}$ and $\D_{\G}'$. The steps that have exponential blow-up are highlighted in red.}
	\end{figure*}

	\paragraph{\textbf{From NBW to tree automata.}}
	\rev{For a given candidate NBW $\C$, we first construct a language-equivalent NBW $\G$ that we know to be GFM, such as a slim automaton from \cite{HPSSTWBP2020} or a suitable limit-deterministic automaton \cite{Hahn15,Sicker16b}. For all known constructions, $\G$ can be exponentially larger than $\C$; we use the slim automata from \cite{HPSSTWBP2020}, which have $O(3^{|Q|})$ states and transitions.

	We then show how to build, for an NBW $\N$, a (symmetric) alternating B\"uchi tree automaton $\T_{\N}$ that accepts the unravelling of an MC $\M$ (again disregarding probabilities, as discussed earlier) exactly when the syntactic product of $\N$ and $\M$ accepts almost surely. Applying this construction to $\C$ and to $\G$ yields the tree automata $\T_{\C}$ and $\T_{\G}$. 
    To facilitate our analysis, we tune the transition function of $\T_{\G}$ slightly which turns it into a language equivalent tree automaton $\T'_{\G}$.
    Since $\G$ is GFM and language-equivalent to $\C$, $\T_{\G}$ and $\T'_{\G}$ accept exactly those MCs $\M$
    on which $\C$ succeeds semantically, that is, $\mathrm{PSem}_{\G}^{\M}(s_0) = \mathrm{PSem}_{\C}^{\M}(s_0) = 1$.}

	\paragraph{\textbf{A first, expensive, decision procedure.}}
	\rev{Checking QGFM-ness for $\C$ now amounts to checking that every MC accepted by $\T_{\G}'$ is also accepted by $\T_{\C}$. 
    Alternating tree automata are closed under complementation at no extra cost -- dualising $\T_{\C}$ to $\widetilde{\T}_{\C}$ only swaps the roles of the two players, negates the transition function, and shifts colours down by one -- so this is the same as checking whether $\lang(\widetilde{\T}_{\C}) \cap \lang(\T_{\G}') = \emptyset$ holds, and any MC in the intersection witnesses that $\C$ is \emph{not} QGFM. That is, $\C$ is QGFM if, and only if, $\lang(\widetilde{\T}_{\C}) \cap \lang(\T_{\G}') = \emptyset$.

	Deciding this emptiness directly, however, means solving a parity game on trees whose branching degree is the direction set $\Upsilon$, which we may need to choose as large as the state space of the MC under consideration -- and hence, in the worst case, as large as the (possibly exponentially larger) automaton $\G$ itself. Solving such a game costs time exponential in the branching degree, so this already gives us a decision procedure for QGFM-ness, but only one that runs in {\sf 2-EXPTIME}. The remaining constructions exist to remove this dependency on $\Upsilon$ (and hence on $\G$): they confine the exponential blow-up to the syntactic material of $\C$ alone, while retaining enough of the structure of $\widetilde{\T}_{\C}$ and $\T_{\G}'$ for the intersection test above to remain valid.}

	\paragraph{\textbf{Widening and pruning.}}
	\rev{Starting from $\widetilde{\T}_{\C}$ and $\T_{\G}'$, we first make the choices of player {\it accept} explicit: rather than asking whether a winning strategy \emph{exists}, as usual when viewing the acceptance of a tree by a tree automaton as the outcome of a game played between player {\it accept} and player {\it reject}, we hand a candidate strategy to the automaton as part of its input and check only whether that particular candidate wins. This turns the existential choices of player {\it accept} into concrete data the later constructions can inspect, and it turns $\widetilde{\T}_{\C}$ into a universal co-B\"uchi tree automaton $\U_{\C}$, and $\T_{\G}'$ into an already deterministic B\"uchi tree automaton $\D_{\G}$, both of which still run on $\Upsilon$-trees.

	We then rein in the number of directions these two automata need to consider, in two steps. First, we \emph{widen} the tree: we add one fresh direction for every state $q$ of $\C$ that player {\it accept} might move to in $\U_{\C}$, plus one further direction reserved for $\D_{\G}$. This disentangles the decisions made across $\U_{\C}$ and $\D_{\G}$, which up to this point might have had to agree on a single, shared direction, even when they concern unrelated parts of the acceptance game. Second, we \emph{prune} the tree, discarding every direction except these newly distinguished ones. This shrinks the branching degree of the tree from the (a priori unbounded) size of the underlying MC down to a \emph{fixed} bound -- one more than the number of states of $\C$ -- which no longer depends on the size of $\G$ or of the MC being inspected.

	\Cref{sssec:widening-and-pruning} explains in detail why each of these two steps is needed. For now, the key point is that widening and pruning only discard information the automata no longer use: a non-empty intersection $\lang(\widetilde{\T}_{\C}) \cap \lang(\T_{\G}')$ after these transformations remains non-empty, even though the trees that witness it may change.}

	\paragraph{\textbf{From alternation to determinism.}}
	\rev{Once pruned, we simplify the resulting automata by dropping the now-unnecessary $\Upsilon$-component of their strategies (\cref{sssec:alternation-to-determinism} spells out this step), which leaves us with a universal co-B\"uchi tree automaton $\U_{\C}'$ and a deterministic B\"uchi tree automaton $\D_{\G}'$, both running on \emph{full} trees over this same fixed direction set, independent of $\Upsilon$ and hence of the size of the original MC. This closes the soundness/completeness gap that widening and pruning, on their own, left open: since $\Upsilon$ was never more than an arbitrary, sufficiently large set of directions to begin with, we may simply take it to already equal this fixed set, and under that identification, a strategy for the simplified automata already \emph{is} a strategy for the original automata $\widetilde{\T}_{\C}$ and $\T_{\G}'$ -- so any tree the simplified automata accept maps straight back to a common accepted tree for $\widetilde{\T}_{\C}$ and $\T_{\G}'$. 
    That is, a non-empty intersection $\lang(\U_{\C}') \cap \lang(\D_{\G}')$ implies a non-empty intersection $\lang(\widetilde{\T}_{\C}) \cap \lang(\T_{\G}')$.
    From $\U_{\C}'$, viewing the universal automaton as running independently along each branch of the tree lets us determinise it word by word, just as one determinises an ordinary universal co-B\"uchi word automaton, which yields a deterministic parity tree automaton $\D_{\C}$. The automaton $\D_{\G}'$ needs no further work: it has been deterministic ever since its choices were first made explicit, so we avoid a second exponential blow-up on this side of the construction.}

	\paragraph{\textbf{Complexity.}}
	\rev{Both $\D_{\C}$ and $\D_{\G}'$ can now be built from $\C$ in time exponential in $|\C|$ alone, and checking their intersection for emptiness costs no more.}
    This establishes membership in $\sf EXPTIME$ for QGFM-ness:
    \begin{restatable}[QGFM is in {\sf EXPTIME}]{thm}{theoremQGFMIntersectionMain}\label{theorem:QGFM-exptime}
    For an NBW $\C$ with $n$ states over an alphabet $\Sigma$, we can decide in time polynomial in $\max\{|\Sigma|,n!\}$ whether or not $\C$ is QGFM and, if $\C$ is not QGFM, construct an MC $\M$ with initial state $s_0$ and $\mathrm{PSyn}_{\C}^{\M}(s_0) \neq \mathrm{PSem}_{\C}^{\M}(s_0) = 1$.
    \end{restatable}

	
\subsection{Detailed Constructions}\label{app:constructions}

\rev{We now carry out, step by step, the pipeline sketched in \cref{fig:algorithms}, organised into the same five stages as the overview in \cref{ssec:QGFM-overview}: we build the tree automata $\T_{\C}$ and $\T_{\G}'$ that reason about (unravellings of) Markov chains directly (\cref{sssec:nbw-to-tree-automata}); we dualise $\T_{\C}$ to obtain a tree automaton $\widetilde{\T}_{\C}$ for the complement language.
The QGFM-ness check then reduces to checking the language emptiness of $\widetilde{\T}_{\C}$ and $\T_{\G}'$.
We make the strategies of $\widetilde{\T}_{\C}$ and $\T_{\G}'$ explicit and widen and prune the resulting trees to a fixed branching degree (\cref{sssec:widening-and-pruning}); we simplify the strategy labels and determinise into a deterministic parity tree automaton $\D_{\C}$, intersecting it with the deterministic B\"uchi tree automaton $\D_{\G}'$ (\cref{sssec:alternation-to-determinism}); and we conclude with the resulting {\sf EXPTIME} bound and the proof of \cref{theorem:QGFM-exptime}.}

An MC with an initial state induces a (state-labelled) \emph{tree} by unravelling and disregarding the probabilities. See \cref{fig:MC-tree} for an example. We say an MC with an initial state is accepted by a tree automaton if the tree induced by the MC and the initial state is accepted by that tree automaton.

\subsubsection{From NBW to Tree Automata}\label{sssec:nbw-to-tree-automata}
\rev{Before giving the formal definition, here is the idea behind it. $\mathrm{PSyn}_{\C}^{\M}(s_0) = 1$ says that, under some strategy for $\C$, the product MDP $\M \times \C$ visits an accepting state infinitely often \emph{almost surely}. 
For B\"uchi objectives on finite Markov chains, however, this is equivalent to a purely structural, non-probabilistic condition. 
By the properties of end-components, $\mathrm{PSyn}_{\C}^{\M}(s_0) = 1$ holds if, and only if, there is a strategy under which every trace of the product almost surely eventually reaches, and remains within, an \emph{accepting} end-component, that is, an end-component that contains an accepting state -- since, once such an end-component is entered, every state and action within it, including the accepting one, is visited infinitely often with probability one. Equivalently, from every reachable state of the product, it is with probability one that an accepting state is reached again.
Since only finitely many states exist, some accepting state is reached within a number of steps bounded by the size of the product \cite{Alfaro1998}. 
It is exactly this uniform bound, available from every state, that lets us capture $\mathrm{PSyn}_{\C}^{\M}(s_0) = 1$ using an ordinary tree automaton: at every node of the tree, $\T_\C$ only needs to point to some direction, via $\Diamond$, along which progress towards an accepting state is guaranteed within this bound, while continuing, via $\Box$, to make the same guarantee available along every other direction as well -- regardless of which one the model actually takes.}

From the candidate NBW $\C = \langle \Sigma, Q, \delta, q_0, F \rangle$, we build a symmetric alternating B\"uchi tree automaton (ABS) $\T_{\C} = \langle \Sigma, Q, \delta_{\T}, q_0, \alpha \rangle$ where $\delta_{\T}(q, \sigma) = \bigvee_{q' \in F \cap \delta(q, \sigma) } (\Box, q', 2) \lor \bigvee_{q' \in (Q \setminus F) \cap \delta(q, \sigma) } (\Box, q', 2) \land (\Diamond, q', 1)$ for $(q,\sigma) \in \supp(\delta)$ and the accepting condition $\alpha$ is a B\"uchi one (and thus a parity condition with two colours), which specifies that a run is accepting if, and only if, the highest colour that occurs infinitely often is even. 
In fact, all other accepting conditions in this section are parity ones.
A run tree is accepted by the ABS $\T_{\C}$ if colour $2$ is seen infinitely often on all infinite paths. Similar constructions can be found in \cite[Section~3]{Schewe2006}.

\begin{restatable}{lem}{lemmaWordToTreeCandidate}\label{lemma:word-to-tree-candidate}
The ABS $\T_{\C}$ accepts an MC $\M$ with initial state $s_0$ if, and only if, the syntactic product of $\M$ and $\C$ almost surely accepts, that is, $\mathrm{PSyn}_{\C}^{\M}(s_0) = 1$. 
\end{restatable}
\begin{proof}
Consider an arbitrary MC $\M$ with initial state $s_0$. Let $\Upsilon$ be the set of states in the MC $\M$.
Let $\langle T, V \rangle$ be the $\Sigma$-labelled $\Upsilon$-tree induced by $\M$ with $s_0$. 

\textbf{`if:'}
We have that $\mathrm{PSyn}_{\C}^{\M}(s_0) = 1$.
As B\"uchi MDPs have positional optimal strategies~\cite{Alfaro1998}, there is a positional strategy $\mu^{\times}$ for the product MDP, such that $\mathrm{PSyn}_{\C}^{\M}(s_0) = \mathrm{PSyn}_{\C}^{\M}\big( (s_0,q_0), \mu^{\times}\big) = 1$ holds.

For every reachable state in the MC $(\M\times\C)_{\mu^{\times}}$, there is a (not necessarily unique) shortest path to the next final state, and for each state $(s,q)$ of $(\M\times\C)_{\mu^{\times}}$, we pick one such direction $d_{(s,q)} \in \Upsilon$ of $\M$ (note that the direction in $\C$ is determined by $\mu^{\times}$).

The length of every path that follows these directions to the next final state is at most the number of states of $\M \times \C$.

\rev{Thus, we resolve the $\Diamond$-choice of $\T_\C$ at a node ending in MC-state $s$ with $\T_\C$ in state $q$ by taking the direction $d_{s,q}$. Since this bound of at most $|\M \times \C|$ steps to the next final state holds from \emph{every} reachable pair $(s,q)$, not just the one we start from, the same argument applies again as soon as the next final state is reached along any path of the run tree. 
Consequently, no path of the run tree ever goes more than $|\M \times \C|$ consecutive transitions without seeing colour $2$, so the dominating colour of every path is $2$: the run tree is accepting, and (the unravelling of) $\M$ is accepted by $\T_\C$.}

\textbf{`only if':}
\rev{Fix an accepting run tree $\langle T_r, r \rangle$ of $\T_\C$ on $\langle T,V \rangle$. 
For a pair $(s,q) \in \Upsilon \times Q$, consider every node $w_r \in T_r$ with $r(w_r) = (w,q)$ for some node $w \in T$ ending in $s$, that is, whose last direction is $s$. 
We denote by $Q_{s,q} \subseteq \delta(q,\ell(s))$ the set of all states $q'$ that are selected as part of the satisfying set at any such node $w_r$.}

For every run tree, this defines an MDP $\M' \subseteq \M \times \C$, and as the run tree we use is accepting, a final state in $\M'$ is reachable from every state in $\M'$.
(To see this, assume for contradiction that $\M'$ contains a state, from which no final state is reachable. Then we can henceforth follow the path from any such state defined by the $\Diamond$ choices, this path has henceforth only colour $1$, and is therefore rejecting. This contradicts the assumption that our run is accepting.)

\rev{Therefore, from every state $(s,q)$ of $\M'$, there is a shortest path within $\M'$ -- using, at each state $(s',q')$ along the way, one of the successors recorded in $Q_{s',q'}$ -- leading to a final state of $\M'$, i.e., to some state $(s',q')$ with $q' \in F$. We fix, for every state $(s,q)$ of $\M'$, the successor $q_{(s,q)} \in Q_{s,q}$ that is the first step of (one such) shortest path from $(s,q)$ to a final state.}
For the resulting positional strategy $\mu$,
the minimal distance to the next final state for $\M_\mu'$ is the same as for $\M'$, and thus finite for every reachable state.
Thus, all reachable end-components in $\M_\mu'$---and thus all reachable end-components in $(\M\times \C)_\mu$---contain a final state.
We therefore have $\mathrm{PSyn}_{\C}^{\M}\big( (s_0,q_0), \mu\big) = 1$.
\end{proof}

From an NBW $\G = \langle \Sigma, Q_{\G}, \delta_{\G}, q_0^{\G}, F_{\G} \rangle$ that is known to be GFM  and language equivalent to $\C$ ($\lang(\G) = \lang(\C)$), similar to $\C$, we can construct an ABS $\T_{\G}$. Instead, we construct a symmetric nondeterministic B\"uchi tree automaton (NBS) $\T_{\G}' = \langle \Sigma, Q_{\G}, \delta_{\T_{\G}'}, q_0^{\G}, \alpha_{\G} \rangle$ where $\delta_{\T_{\G}'}(q, \sigma) = \bigvee_{q' \in F_{\G} \cap \delta_{\G}(q, \sigma) } (\Box, q', 2) \lor \bigvee_{q' \in (Q_{\G} \setminus F_{\G}) \cap \delta_{\G}(q, \sigma) } (\boxtimes_1, q', 2)$ for all $(q,\sigma) \in \supp(\delta_{\G})$. The automaton $\T_{\G}'$ differs from $\T_{\G}$ only in the transition functions for the successor states that are non-accepting: in $\T_{\G}'$, a copy of the automaton in a non-accepting state $q'$ is sent to some successor with colour $1$ and all the other successors with colour $2$; in $\T_{\G}$, a non-accepting state $q'$ is sent to some successor with both colour $1$ and~$2$. 
The following lemma claims that these two tree automata are language equivalent:

\begin{restatable}{lem}{lemmaLanguageEquivalentNBTABT}\label{lemma:language-equivalent-NBT-ABT}
$\lang(\T_{\G}') = \lang(\T_{\G})$.
\end{restatable}
\begin{proof}
To prove $\lang(\T_{\G}) \subseteq \lang(\T_{\G}')$, we show that any tree rejected by $\T_{\G}'$ is rejected by $\T_{\G}$, that is, for a tree $t$, if player {\it reject} in $\T_{\G}'$ has a winning strategy, player {\it reject} in $\T_{\G}$ also has a winning strategy. It is trivially true as player {\it reject} in $\T_{\G}$ can just use the winning strategy by player {\it reject} in $\T_{\G}'$, since the set of choices of player {\it reject} in $\T_{\G}$ is a superset of that in $\T_{\G}'$.

Assume that player {\it reject} in $\T_{\G}$ has a memoryless winning strategy over an input tree.
We modify this winning strategy such that player {\it reject} sends a copy of the automaton in state $q$ with colour $1$ regardless of which colour is chosen in the original winning strategy, whenever player {\it reject} has the choice of sending with either colour $1$ or $2$. 
We argue that this modified strategy is also a winning strategy for player {\it reject}.
The run tree induced by the modified strategy remains the same with possibly more colour $1$'s on some edges of the run tree. For any tree rejected by $\T_{\G}$, player {\it reject} in $\T_{\G}'$ can use this modified winning strategy to win the rejection game.
It follows that $\lang(\T_{\G}') \subseteq \lang(\T_{\G})$.
\end{proof}

By \cref{lemma:word-to-tree-candidate} and 
\cref{lemma:language-equivalent-NBT-ABT}, 
we have:
\begin{restatable}{lem}{lemmaWordToTreeGFM}\label{lemma:word-to-tree-GFM}
The NBS $\T_{\G}'$ accepts an MC $\M$ with initial state $s_0$ if, and only if, the syntactic product of $\G$ and $\M$ almost surely accepts, that is, $\mathrm{PSyn}_{\G}^{\M}(s_0) = 1$. 
\end{restatable}

We dualise $\T_\C$ into a symmetric alternating co-B\"uchi tree automaton (ACS), which is a complement of the ABS $\T_{\C}$. In other words, this automaton accepts an MC $\M$ if, and only if, the syntactical product of $\C$ and $\M$ does \emph{not} almost surely accept.
Dualising the ABS $\T_{\C}$ consists of switching the roles of player {\it accept} and player {\it reject}, complementing the transition function and decreasing the colours by one.
Thus, we define the ACS $\widetilde{\T}_{\C} = \langle \Sigma, Q, \delta_{\widetilde{\T}}, q_0, \widetilde{\alpha} \rangle$ where $\delta_{\widetilde{\T}}(q, \sigma) = \bigwedge_{q' \in F \cap \delta(q, \sigma) } (\Diamond, q', 1) \land \bigwedge_{q' \in (Q \setminus F) \cap \delta(q, \sigma) } (\Diamond, q', 1) \lor (\Box, q', 0)$ for $(q,\sigma) \in \supp(\delta)$.
Since $0$ and $1$ are the only colours, $\widetilde{\T}_{\C}$ is a co-B\"uchi automaton, and a run tree is accepted when colour $1$ is seen only finitely often on all infinite paths.

\begin{restatable}{lem}{lemmaWordToTreeCandidateComplement}\label{lemma:word-to-tree-candidate-complement}
The ACS 
$\widetilde{\T}_{\C}$ 
accepts the complement of $\lang(\T_\C)$. 
\end{restatable}

By \cref{lemma:word-to-tree-candidate-complement} and \cref{lemma:word-to-tree-GFM}, together with \cref{def:QGFM,theorem:MC-MDP-interchangable-QGFM}, the candidate NBW $\C$ is QGFM if, and only if, $\lang(\widetilde{\T}_{\C}) \cap \lang(\T_{\G}') = \emptyset$: any MC accepted by both automata is one on which $\G$ -- and hence $\C$, since $\lang(\G) = \lang(\C)$ -- succeeds semantically, while $\C$ itself fails to succeed syntactically, witnessing that $\C$ is not QGFM.


\subsubsection{Widening and Pruning}\label{sssec:widening-and-pruning}
\paragraph{\textbf{Strategy of the ACS $\widetilde{\T}_{\C}$.}} Define a function $\delta': Q \to \{\Box\} \cup \Upsilon$ such that a state $q \in Q$ is mapped to an $\upsilon \in \Upsilon$ if $q \in F$ and is mapped to either $\Box$ or an $\upsilon \in \Upsilon$ if $q \not\in F$. Let $\Delta_{\C}$ be the set of such functions. Intuitively, a function $\delta'$ can be considered as the strategy of player {\it accept} at a tree node. Since an input tree is accepted by a parity automaton iff player {\it accept} has a memoryless winning strategy on the tree, we can define the strategy of $\widetilde{\T}_{\C}$ for a $\Sigma$-labelled $\Upsilon$-tree $\langle T, V\rangle$ as an $\Delta_{\C}$-labelled tree $\langle T,f_{\C} \rangle$ where $f_{\C}: T \to \Delta_{\C}$. A strategy $\langle T, f_{\C} \rangle$ induces a single run $\langle T_r, r \rangle$ of $\widetilde{\T}_{\C}$ on $\langle T,V \rangle$. The strategy is winning if the induced run is accepting. Whenever the run $\langle T_r, r \rangle$ is in state $q$ as it reads a node $w \in T$, it proceeds according to $f_{\C}(w)$. Formally, $\langle T_r, r \rangle$ is the $(T \times Q)$-labelled tree such that:
\begin{itemize}
    \item $\epsilon \in T_{r}$ and $r(\epsilon) = (\epsilon, q_0)$;
    \item Consider a node $w_r \in T_r$ with $r(w_r) = (w, q)$. Let $\sigma = V(w)$. For all $q' \in \delta(q, \sigma)$ and $c = f_{\C}(w)(q')$ we have the following two cases: 
    \begin{enumerate}
        \item if $c = \Box$, then for each successor node $w'$ of $w$, there is a $j \in \nat$ such that $w_r {\cdot} j \in T_r$, $r(w_r {\cdot} j) = (w', q')$, the transition in $T_r$ between $w_r$ and $w_r {\cdot} j$ is assigned colour $0$;
        \item if $c \in \Upsilon$ and $w {\cdot} c \in \suc(w)$, then there is a $j \in \nat$ such that $w_r {\cdot} j \in T_r$, $r(w_r {\cdot} j) = (w {\cdot} c, q')$ and the transition in $T_r$ between $w_r$ and $w_r {\cdot} j$ is assigned colour $1$.
    \end{enumerate}  
\end{itemize}

\begin{restatable}{lem}{lemmaStrategyAndStrategyLabelledTreeCandidate}\label{lemma:strategy-and-strategy-labelled-tree-candidate}
The ACS $\widetilde{\T}_{\C}$ accepts a $\Sigma$-labelled $\Upsilon$-tree $\langle T, V \rangle$ if, and only if, there exists a winning strategy $\langle T, f_{\C} \rangle$ for the tree.
\end{restatable}

\paragraph{\textbf{Strategy of the NBS $\T_{\G}'$.}} Let $\Delta_\G \subseteq (\{\Box\} \cup \Upsilon) \times Q_\G$ such that $(\Box, q) \in \Delta_\G$ if $q \in F_{\G}$ and $(\upsilon, q) \in \Delta_\G$ for some $\upsilon \in \Upsilon$ if $q \not\in F_{\G}$. An item in the set $\Delta_{\G}$ can be considered as the strategy of player {\it accept} at a tree node in the acceptance game of $\T_{\G}'$. 
Intuitively, by choosing both, the elements that are going to be satisfied and the directions in which the existential requirements are going to be satisfied, an item in $\Delta_\G$ resolves all the nondeterminism in $\delta_{\T_{\G}'}$.
A tree is accepted if, and only if, player {\it accept} has a memoryless winning strategy. 
Thus, we can define player {\it accept}'s strategy of $\T_{\G}'$ on a $\Sigma$-labelled $\Upsilon$-tree $\langle T, V \rangle$ as a $\Delta_{\G}$-labelled tree $\langle T, f_{\G} \rangle$ with $f_{\G}: T \to \Delta_{\G}$. A strategy $\langle T, f_{\G} \rangle$ induces a single run $\langle T_r, r \rangle$ and it is winning if, and only if, the induced run is accepting.

Whenever the run $\langle T_r, r \rangle$ is in state $q$ as it reads a node $w \in T$, it proceeds according to $f_{\G}(w)$. Formally, $\langle T_r, r \rangle$ is the $(T \times Q_{\G})$-labelled tree such that:
\begin{itemize}
    \item $\epsilon \in T_{r}$ and $r(\epsilon) = (\epsilon, q_0^{\G})$;
    \item Consider a node $w_r \in T_r$ with $r(w_r) = (w, q)$. Let $\sigma = V(w)$ and $(c, q') = f_{\G}(w)$.  
    We have the following two cases: 
    \begin{enumerate}
        \item if $c = \Box$, then $q' {\in} F_\G$ and, for each successor node $w'$ of $w$, there is a $j {\in} \nat$ s.t.\ $w_r {\cdot j} \in T_r$ and $r(w_r {\cdot j}) = (w', q')$, and the transition in $T_r$ between $w_r$ and $w_r {\cdot j}$ is assigned colour~$2$;
        \item if $c \in \Upsilon$ and $w {\cdot} c \in \suc(w)$, then $q'\notin F_\G$ and there is a $j \in \nat$ such that $w_r {\cdot j} \in T_r$, $r(w_r {\cdot j}) = (w {\cdot c}, q')$, and the transition in $T_r$ between $w_r$ and $w_r {\cdot j}$ is assigned colour~$1$.
        Also, for all successor node $w'$ of $w$ with $w' \not= w {\cdot c}$, there is a $j \in \nat$ s.t.\ $w_r {\cdot j} \in T_r$, $r(w_r {\cdot j}) = (w', q')$, and the transition in $T_r$ between $w_r$ and $w_r {\cdot j}$ is assigned the colour~$2$.
    \end{enumerate}  
\end{itemize}

\begin{restatable}{lem}{lemmaStrategyAndStrategyLabelledTreeGFM}\label{lemma:strategy-and-strategy-labelled-tree-GFM}
The NBS $\T_{\G}'$ accepts a $\Sigma$-labelled $\Upsilon$-tree $\langle T, V \rangle$ if, and only if, there exists a winning strategy $\langle T, f_{\G} \rangle$ for the tree.
\end{restatable}

\paragraph{\textbf{Making Strategies Explicit.}} \rev{So far, we have only said that $\widetilde{\T}_{\C}$ and $\T_{\G}'$ accept a tree when player {\it accept} \emph{has} a winning strategy, without fixing what that strategy actually looks like as an object we could compute with -- an existential claim buried inside a two-player game is not yet something a tree-automaton construction can manipulate directly. We therefore turn this existential search into an explicit input: instead of asking $\widetilde{\T}_{\C}$ to find a good strategy $f_{\C}$ on its own, we hand it a \emph{candidate} $f_{\C}$ as an extra component of the tree's labelling, and only ask whether this particular candidate happens to be winning. Once player {\it accept}'s choices are pinned down this way, nothing existential is left inside the automaton's own transitions -- what remains is exactly the universal ($\Box$) part of the original acceptance game, so the alternating $\widetilde{\T}_{\C}$ turns into a purely universal automaton. Applying the same idea to $\T_{\G}'$, whose only real freedom was existential to begin with, leaves no choice at all once that freedom is fixed by the label, so $\T_{\G}'$ turns into a fully deterministic automaton.

In place of two automata whose acceptance is defined through a game, we obtain two automata whose transitions can be inspected and manipulated directly. This is also what makes fixing the number of directions possible in the first place: to bound the branching degree, the widening step below needs to know, at every node and for every state, exactly which direction the strategy currently sends that state to, so that it can manufacture one fresh, private copy of that direction per state. As long as a winning strategy is only an implicit, existential fact about the acceptance game, there is no such concrete ``current direction'' to point to -- pinning one down would mean solving the game first. Once the strategy has been written into the tree as data, however, that direction is simply the value $f_{\C}(w)$ (or $f_{\G}(w)$) read off at the node $w$, so widening -- and, with it, the reduction to a fixed branching degree -- becomes a purely local, syntactic rewriting of the tree, rather than a global operation on winning strategies.}

Formally, annotating input trees with functions from $\Delta_{\C}$ and $\Delta_{\G}$ realises this idea, transforming the ACS $\widetilde{\T}_{\C}$ into a universal co-B\"uchi tree automaton (UCT) and the NBS $\T_{\G}'$ into a nearly symmetric deterministic B\"uchi tree automaton (DBN).

We first define a UCT $\mathcal{U}_{\C} = \langle \Sigma \times \Delta_{\C} \times \Delta_{\G}, Q, \delta_{\mathcal{U}_{\C} }, q_0,  \alpha_{\U_{\C}} \rangle$ that runs on $(\Sigma \times \Delta_{\C} \times \Delta_{\G})$-labelled $\Upsilon$-trees. For all $q \in Q$ and $(\sigma, \delta', g) \in \Sigma \times \Delta_{\C} \times \Delta_{\G}$, define $\delta_{\U_{\C} }(q, \langle \sigma, \delta', g \rangle) = \bigwedge_{q' \in \delta(q, \sigma) \land \delta'(q') = \Box} (\Box, q', 0) \land \bigwedge_{q' \in \delta(q, \sigma) \land \upsilon = \delta'(q') \in \Upsilon} (\upsilon, q', 1)$. Then, we have:

\begin{restatable}{lem}{lemmaAccStrategyExplicitCandidate}\label{lemma:acc-strategy-explicit-candidate}
The UCT $\mathcal{U}_{\C}$ accepts a tree $\langle T, (V, f_{\C}, f_{\G}) \rangle$ if, and only if, the ACS $\widetilde{\T}_{\C}$ accepts $\langle T, V \rangle$ with $\langle T,  f_{\C} \rangle$ as the winning strategy.
\end{restatable}

Next, we define a DBN $\D_{\G} = \langle \Sigma \times \Delta_{\C} \times \Delta_{\G}, Q_{\G}, \delta_{\D_{\G}}, q_0^{\G},  \alpha_{\D_{\G}} \rangle$ that runs on $(\Sigma \times \Delta_{\C} \times \Delta_{\G})$-labelled $\Upsilon$-trees,
where $\delta_{\D_{\G}}\big( q, \langle \sigma, \delta', (c, q') \rangle \big) = \left \{ \begin{array}{ll}
(\Box, q', 2) & \mbox{if $c = \Box$}\\
(\boxtimes_1^{c}, q', 2) & \mbox{if $c \in \Upsilon$}
\end{array}
\right .
$ for all $q \in Q_{\G}$ and $\big( \sigma, \delta', (c, q') \big) \in \Sigma \times \Delta_{\C} \times \Delta_{\G}$. Then, we have:

\begin{restatable}{lem}{lemmaAccStrategyExplicitGFM}\label{lemma:acc-strategy-explicit-GFM}
The DBN $\D_{\G}$ accepts a tree $\langle T, (V, f_{\C}, f_{\G}) \rangle$ if, and only if, the NBS $\T_{\G}'$ accepts $\langle T, V \rangle$ with $\langle T, f_{\G} \rangle$ as the winning strategy.
\end{restatable}

\rev{At this point, $\mathcal{U}_{\C}$ and $\D_{\G}$ both run on $\Upsilon$-trees, where $\Upsilon$ is the (a priori unbounded) direction set induced by the Markov chain $\M$ under consideration -- recall that, as discussed above, $\Upsilon$ can, for instance, be chosen as the state space of $\M$. This is inconvenient for two reasons. First, an automaton whose branching degree grows with the size of the model it inspects is hard to analyse with standard techniques, which typically assume a fixed, small alphabet of directions. Second, and more importantly, the decisions of player {\it accept} in $\mathcal{U}_{\C}$ and in $\D_{\G}$ are, at this stage, only recorded implicitly, namely as a concrete direction $\upsilon \in \Upsilon$ chosen by the strategies $f_{\C}$ and $f_{\G}$: two different states $q_1, q_2 \in Q$ of the candidate automaton $\C$ may well want to move to the \emph{same} direction $\upsilon$, and this direction may, in addition, coincide with the direction chosen for $\G$. When several of these choices collide in one concrete direction $\upsilon$, the picture of ``who moves where'' is entangled, and it is exactly this entanglement that stands in the way of intersecting and later determinising the automata.

We resolve this in two steps. In the first step (\emph{widening}), we do not remove any information, but instead \emph{add} directions: we equip every existing direction $\upsilon \in \Upsilon$ with $|\QDiamond|$ private copies, one for every $q \in \QDiamond = Q \cup \{\text{\ding{73}}\}$ -- one dedicated copy for every state of $\C$ that might need to move in this step, plus one further copy reserved for $\G$. A strategy can now record its decision for state $q$ by \emph{which of the copies} of $\upsilon$ it uses, rather than by $\upsilon$ alone, so the decisions for different states $q$, and the decision made on behalf of $\G$, no longer need to coincide even when they used to point in the same direction of $\Upsilon$. Because the additional copies only ever carry colours that cannot spoil an already accepting run (\cref{lemma:widen-candidate,lemma:widen-GFM}), widening changes nothing about which trees are accepted, and it is only a bookkeeping device that untangles the decisions before we discard what is no longer needed.

In the second step (\emph{pruning}), we exploit the fact that, once the identity of the state $q \in \QDiamond$ tells us everything the automata need to know, the concrete direction $\upsilon \in \Upsilon$ that goes with it becomes irrelevant: for every $q \in \QDiamond$, only the single successor selected by the (now widened) strategy matters, and every other successor can safely be dropped. This shrinks a tree whose branching degree could be as large as $|\Upsilon|$ down to one with the small, \emph{fixed} branching degree $|\QDiamond| = |Q|+1$, independent of the size of the model $\M$ -- exactly the kind of tree that standard automata-theoretic machinery, such as determinisation via the word-level view of universal automata, can handle. Since pruning only removes directions that the automata no longer inspect, an accepting run survives pruning (\cref{lemma:prune-candidate,lemma:prune-GFM}); the converse need not hold, as a pruned tree may originate from several different unpruned trees, only some of which were accepting to begin with. We return to this asymmetry, and to how it is repaired, in the paragraph on the language intersection emptiness check below.}
\smallskip

\paragraph{\textbf{Tree Widening.}} The operator $\hide_Y$ maps a tree node $w \in (X \times Y)^*$ to tree node $w' \in X^*$ such that $w'$ is obtained by replacing each $\langle x, y \rangle$ in $w$ by $x$. The operator $\wide_Y$ maps a $\Sigma$-labelled $\Upsilon$-tree $\langle T, V \rangle$ to a $\Sigma$-labelled $(\Upsilon \times Y)$-tree $\langle T', V' \rangle$ such that, for every node $w \in T'$, we have $V'(w) = V(\hide_Y(w))$. Let $Q_{\text{\ding{73}}} = Q \cup \{\text{\ding{73}}\}$. We widen the input trees with $\QDiamond$\rev{, formalising the extra, state-indexed copies of each direction described above}. We then define a UCT and a DBN that accept the widened trees. 

Let a UCT which runs on $(\Sigma \times \Delta_{\C} \times \Delta_{\G})$-labelled $(\Upsilon \times \QDiamond)$-trees $\U_{\C}^{w} = \langle \Sigma \times \Delta_{\C} \times \Delta_{\G}, Q, \delta_{\U_{\C}^{w}}, q_0, \alpha_{\U_{\C}^{w}} \rangle$. For all $q \in Q$ and $(\sigma, \delta', g) \in \Sigma \times \Delta_{\C} \times \Delta_{\G}$, define $\delta_{\U_{\C}^{w}}(q, \langle \sigma, \delta', g\rangle) = \bigwedge_{q' \in \delta(q, \sigma) \land \delta'(q') = \Box} (\Box, q', 0) \land \bigwedge_{q' \in \delta(q, \sigma) \land \upsilon = \delta'(q') \in \Upsilon} \big( (\upsilon, q'), q', 1 \big)$ for all $q \in Q$ and $(\sigma, \delta', g) \in \Sigma \times \Delta_{\C} \times \Delta_{\G}$.

\begin{restatable}{lem}{lemmaWidenCandidate}\label{lemma:widen-candidate}
The UCT $\mathcal{U}_{\C}$ accepts a $(\Sigma \times \Delta_{\C} \times \Delta_{\G})$-labelled tree $\langle T, (V, f_{\C}, f_{\G}) \rangle$ if, and only if, the UCT $\U_{\C}^{w}$ accepts a $(\Sigma \times \Delta_{\C} \times \Delta_{\G})$-labelled tree $\wide_{\QDiamond}\big( \langle T, (V, f_{\C}, f_{\G}) \rangle \big)$. 
\end{restatable}
\begin{proof}
Consider the run trees of $\mathcal{U}_{\C}$ and $\U_{\C}^{w}$ on an arbitrary input tree $\langle T, (V, f_{\C}, f_{\G}) \rangle$ and the corresponding widened tree $\wide_{\QDiamond}\big( \langle T, (V, f_{\C}, f_{\G}) \rangle \big)$, respectively. For every tree path in the run tree of $\mathcal{U}_{\C}$, there is the same tree path in the run tree of $\U_{\C}^{w}$. In other words, the run tree of $\mathcal{U}_{\C}$ is a subtree of that of $\U_{\C}^{w}$. Furthermore, consider the paths that are in the run tree of $\U_{\C}^{w}$ but are not in the run tree of $\U_{\C}$. For those paths, there are more colour $0$'s, which will not affect the acceptance of the run tree of the UCT~$\U_{\C}^{w}$. Thus, the run tree of $\mathcal{U}_{\C}$ is accepting if, and only if, the run tree of $\U_{\C}^{w}$ is accepting.
\end{proof}

Define a DBN $\D_{\G}^{w} = \langle \Sigma \times \Delta_{\C} \times \Delta_{\G}, Q_{\G}, \delta_{\D_{\G}^{w}}, q_0^{\G}, \alpha_{\D_{\G}^{w}} \rangle$ that runs on $(\Sigma \times \Delta_{\C} \times \Delta_{\G})$-labelled $(\Upsilon \times \QDiamond)$-trees, where $\delta_{\D_{\G}^w}\big( q, \langle \sigma, \delta', (c, q') \rangle \big)$ is $(\Box, q', 2)$ if $c = \Box$ and $(\boxtimes_1^{(c,\text{\ding{73}})}, q', 2)$ if $c \in \Upsilon$
for all $q \in Q_{\G}$, $\big( \sigma, \delta', (c, q') \big) \in \Sigma \times \Delta_{\C} \times \Delta_{\G}$. Similar to \cref{lemma:widen-candidate}, we have:
\begin{restatable}{lem}{lemmaWidenGFM}\label{lemma:widen-GFM}
The DBN $\D_{\G}$ accepts a $(\Sigma \times \Delta_{\C} \times \Delta_{\G})$-labelled tree $\langle T, (V, f_{\C}, f_{\G}) \rangle$ if, and only if, the DBN $\D_{\G}^w$ accepts a $(\Sigma \times \Delta_{\C} \times \Delta_{\G})$-labelled tree $\wide_{\QDiamond}\big( \langle T, (V, f_{\C}, f_{\G} ) \rangle \big)$. 
\end{restatable}
\begin{proof}
The proof is similar to that of \cref{lemma:widen-candidate}. Consider the run trees of $\D_{\G}$ and $\D_{\G}^{w}$ on an arbitrary input tree $\langle T, (V, f_{\C}, f_{\G}) \rangle$ and the corresponding widened tree $\wide_{\QDiamond} \big(\langle T, (V, f_{\C}, f_{\G}) \rangle \big)$, respectively. For every tree path in the run tree of $\D_{\G}$, there is the same tree path in the run tree of $\D_{\G}^{w}$. In other words, the run tree of $\D_{\G}$ is a subtree of that of $\D_{\G}^{w}$. Furthermore, consider the paths that are in the run tree of $\D_{\G}^{w}$ but are not in the run tree of $\D_{\G}$. For those paths, there are more colour $2$'s, which will not affect the acceptance of the run tree of the DBN~$\D_{\G}^{w}$. Thus, the run tree of $\mathcal{D}_{\G}$ is accepting if, and only if, the run tree of $\D_{\G}^{w}$ is accepting.
\end{proof}

\paragraph{\textbf{Tree Pruning.}} Given a widened tree $\wide_{\QDiamond}\big( \langle T, (V, f_{\C}, f_{\G}) \rangle \big)$, the operator $\prune$ keeps, for all $w \in \wide_{\QDiamond}\big( \langle T, (V, f_{\C}, f_{\G}) \rangle \big)$ and all $q \in \QDiamond$, exactly one direction $(\upsilon_q, q)$ from $\suc(w)$, such that, for all $q \in \QDiamond$, the following holds:
\begin{itemize}
    \item if $q = \text{\ding{73}}$ and $f_{\G}(\hide_{\QDiamond}(w)) = (\upsilon, q') \in \Upsilon \times Q_{\G}$, then $\upsilon_q = \upsilon$; and
    \item if $q \in Q$ and $f_{\C}(\hide_{\QDiamond}(w))(q) = \upsilon \in \Upsilon$ then $\upsilon_q =  \upsilon$.
\end{itemize}
If $q \in Q$ and $f_{\C}(\hide_{\QDiamond}(w))(q) = \Box$ or if $q = \text{\ding{73}}$ and $f_{\G}(\hide_{\QDiamond}(w)) \in \{ \Box \} \times Q_{\G}$, then $\upsilon_q$ can be chosen arbitrarily from $\suc(\hide_{\QDiamond}(w))$.

Next, all nodes $w$ in the pruned tree are renamed to $\hide_{\Upsilon}(w)$. Pruning the tree by operator $\prune$ gives us full $\QDiamond$-trees, and thus trees with fixed branching degree $|\QDiamond|$. 

Let us define a UCT $\mathcal{U}_{\C}^{p} = \langle \Sigma \times \Delta_{\C} \times \Delta_{\G}, Q, \delta_{\U_{\C}^{p}}, q_0, \alpha_{\U_{\C}^{p}} \rangle$ that runs on full $(\Sigma \times \Delta_{\C} \times \Delta_{\G})$-labelled $\QDiamond$-trees, where $\delta_{\mathcal{U}^{p}_{\C}}(q, \langle \sigma, \delta', g \rangle) =  \bigwedge_{q' \in \delta(q, \sigma) \land \delta'(q') = \Box} (\Box, q', 0) \land \bigwedge_{q' \in \delta(q, \sigma) \land \delta'(q') \in \Upsilon} ( q', q', 1 )$ for all $q \in Q$ and $(\sigma, \delta', g) \in \Sigma \times \Delta_{\C} \times \Delta_{\G}$. 

Consider the run trees of $\mathcal{U}_{\C}^{w}$ and $\U_{\C}^{p}$ on an arbitrary widened tree $\wide_{\QDiamond} \big(\langle T, (V, f_{\C}, f_{\G}) \rangle \big)$ and the corresponding pruned tree $\prune\Big( \wide_{\QDiamond} \big(\langle T, (V, f_{\C}, f_{\G}) \rangle \big) \Big)$, respectively. The run tree of $\mathcal{U}_{\C}^{p}$ is a subtree of that of $\U_{\C}^{w}$: 
\begin{restatable}{lem}{lemmaPruneCandidate}\label{lemma:prune-candidate}
If the UCT $\U_{\C}^{w}$ accepts a $(\Sigma \times \Delta_{\C} \times \Delta_{\G})$-labelled tree $\wide_{\QDiamond} \big(\langle T, (V, f_{\C}, f_{\G}) \rangle \big)$, then the UCT $\mathcal{U}_{\C}^p$ accepts the $(\Sigma \times \Delta_{\C} \times \Delta_{\G})$-labelled tree $\prune\Big( \wide_{\QDiamond} \big(\langle T, (V, f_{\C}, f_{\G}) \rangle \big) \Big)$. 
\end{restatable}

Let us define a DBN $\D_{\G}^{p} = \langle \Sigma \times \Delta_{\C} \times \Delta_{\G}, Q_{\G}, \delta_{\D_{\G}^{p}}, q_0^{\G}, \alpha_{\D_{\G}^{p}} \rangle$ that runs on full $(\Sigma \times \Delta_{\C} \times \Delta_{\G})$-labelled $\QDiamond$-trees, where $\delta_{\D_{\G}^{p}}\big( q, \langle \sigma, \delta', (c, q') \rangle \big)$
is $(\Box, q', 2)$ if $c = \Box$ and $(\boxtimes_1^{\text{\ding{73}}}, q', 2)$ if $c \in \Upsilon$
for all $q \in Q_{\G}$, $\big( \sigma, \delta', (c, q') \big) \in \Sigma \times \Delta_{\C} \times \Delta_{\G}$. Similar to \cref{lemma:prune-candidate}, we have: 
\begin{restatable}{lem}{lemmaPruneGFM}\label{lemma:prune-GFM}
If the DBN $\D_{\G}^{w}$ accepts a $(\Sigma \times \Delta_{\C} \times \Delta_{\G})$-labelled tree $\wide_{\QDiamond}(\langle T, (V, f_{\C}, f_{\G}) \rangle)$, then the DBN $\mathcal{D}_{\G}^p$ accepts the full $(\Sigma \times \Delta_{\C} \times \Delta_{\G})$-labelled tree $\prune\Big(\wide_{\QDiamond}\big( \langle T, (V, f_{\C}, f_{\G} ) \rangle \big)\Big)$. 
\end{restatable}
\begin{proof}
Consider the run trees of $\D_{\G}^{w}$ and $\D_{\G}^{p}$ on an arbitrary widened tree $\wide_{\QDiamond} \big(\langle T, (V, f_{\C}, f_{\G}) \rangle \big)$ and the corresponding pruned tree $\prune\Big( \wide_{\QDiamond} \big(\langle T, (V, f_{\C}, f_{\G}) \rangle \big) \Big)$, respectively. It follows from the fact that the run tree of $\D_{\G}^{p}$ is a subtree of the run tree of $\D_{\G}^{w}$.
\end{proof}

Note that a tree $\prune\Big(\wide_{\QDiamond}\big( \langle T, (V, f_{\C}, f_{\G} ) \rangle \big)\Big)$ being accepted by $\U_{\C}^{p}$ does not imply that the tree $\langle T, V \rangle$ is accepted by $\widetilde{\T}_{\C}$.
This is because some nodes in the tree $T$ may disappear after pruning.
For the same reason, that a tree $\prune\Big(\wide_{\QDiamond}\big( \langle T, (V, f_{\C}, f_{\G} ) \rangle \big)\Big)$ is accepted by $\D_{\G}^{p}$ does not imply that the tree $\langle T, V \rangle$ is accepted by $\T_{\G}'$.
\smallskip

\subsubsection{From Alternation to Determinism}\label{sssec:alternation-to-determinism}
\paragraph{\textbf{Strategy Simplification -- removing the direction $\Upsilon$.}}
After pruning, every state $q \in \QDiamond$ that is not sent to $\Box$ already has a single, distinguished successor of its own, and the only thing the leftover $\Upsilon$-label on that successor still records is \emph{which} of the original directions it happened to be -- information $\U_{\C}^p$ and $\D_{\G}^p$ never actually consult once the successor is fixed. Dropping it turns these automata into $\U_{\C}'$ and $\D_{\G}'$, which run on \emph{full} $\QDiamond$-trees: trees whose direction set is the same small, fixed set $\QDiamond$, no matter which Markov chain -- and hence which original direction set $\Upsilon$ -- we started from. This uniformity is what lets us close the loop between the simplified automata and the original ones: since $\Upsilon$ was never anything more than an arbitrary, sufficiently large set of directions to begin with, and $\widetilde{\T}_{\C}$ and $\T_{\G}'$ are defined the same way for any choice of $\Upsilon$, nothing stops us from simply taking $\Upsilon = \QDiamond$. Under that choice, a strategy for $\U_{\C}'$ and $\D_{\G}'$ already \emph{is} a strategy for the original automata $\widetilde{\T}_{\C}$ and $\T_{\G}'$, so any tree the simplified automata accept can be mapped straight back to a witness that $\widetilde{\T}_{\C}$ and $\T_{\G}'$ accept a common tree -- exactly the correctness argument the widening and pruning steps above still owed us, which we spell out next.

Formally, the $\Upsilon$-component of a strategy has, by now, done its job: it was only needed to pick out, during pruning, which single successor to keep for each $q \in \QDiamond$, and the automata $\U_{\C}^p$ and $\D_{\G}^p$ no longer inspect it once that successor has been fixed -- they only ever ask whether a state is sent to $\Box$ or into the (unique, now implicit) successor direction. We can therefore drop the $\Upsilon$-part of the strategies altogether, which simplifies the alphabet from $\Sigma \times \Delta_{\C} \times \Delta_{\G}$ to the smaller $\Sigma \times \Delta_{\C}' \times \Delta_{\G}'$ and removes any remaining reference to the (unbounded) direction set $\Upsilon$ of the original Markov chain:
Let $\Delta_{\C}'$ be the set of functions $Q \to \{\Box, \Diamond\}$ and $\Delta_{\G}' \subseteq Q_{\G}$. For $\Delta_{\G}'$, we remove $\{\Box\} \cup \Upsilon$ from $\Delta_{\G}$, as it is clear from the construction of the DBT $\D_{\G}^p$ and the definition of strategy $\Delta_{\G}$ that we have a $(\Box, q, 2)$ transition when $q \in F_{\G}$ is chosen and a $(\boxtimes^{\text{\ding{73}}}_1, q, 2)$ transition when $q \not\in F_{\G}$ is chosen.  

We first define the UCT $\U_{\C}' = \langle \Sigma \times \Delta_{\C}' \times \Delta_{\G}', Q, \delta_{\U_{\C}'}, q_0, \alpha_{\U_{\C}'} \rangle$ that runs on full $(\Sigma \times \Delta_{\C}' \times \Delta_{\G}')$-labelled $\QDiamond$-trees.
Let  $\delta_{\mathcal{U}_{\C}'}(q, \langle \sigma, \delta', g \rangle) = \bigwedge_{q' \in \delta(q, \sigma) \land \delta'(q') = \Box} (\Box, q', 0) \land \bigwedge_{q' \in \delta(q, \sigma) \land \delta'(q') = \Diamond} ( q', q', 1 )$ for all $q \in Q$ and $(\sigma, \delta', g) \in \Sigma \times \Delta_{\C}' \times \Delta_{\G}'$.
We next define the DBN that runs on full $(\Sigma \times \Delta_{\C}' \times \Delta_{\G}')$-labelled $\QDiamond$-trees $\D_{\G}' = \langle \Sigma \times \Delta_{\C}' \times \Delta_{\G}', Q_{\G}, \delta_{\D_{\G}'}, q_0^{\G}, \alpha_{\D_{\G}'} \rangle$ where $\delta_{\D_{\G}'}\big( q, \langle \sigma, \delta',  q' \rangle \big)$ 
is $(\Box, q', 2)$ if $q' \in F_{\G}$ and $(\boxtimes_1^{\text{\ding{73}}}, q', 2)$ if $q' \not\in F_{\G}$
for all $q \in Q_{\G}$ and $(\sigma, \delta', q') \in \Sigma \times \Delta_{\C}' \times \Delta_{\G}'$. 

Given a $(\Sigma \times \Delta_{\C} \times \Delta_{\G})$-labelled $\QDiamond$-tree $\langle T, (V, f_{\C}, f_{\G}) \rangle$, the operator $\relabel$ assigns all nodes $w \in T$ with old label $\big(\sigma, \delta', (c, q')\big)$ a new label $\big(\sigma, \delta'', q'\big)$, where $\delta''(q) $
is $\Box$ if $\delta'(q) = \Box$ and $\Diamond$ if $\delta'(q) \in \Upsilon$.
Relabelling gives us $(\Sigma \times \Delta_{\C}' \times \Delta_{\G}')$-labelled trees.

\begin{restatable}{lem}{lemmaStrategyLabelUpsilonFreeCandidateGFM}\label{lemma:strategy-label-Upsilon-free-candidate-GFM}
The UCT $\U_{\C}^p$ accepts a full $(\Sigma \times \Delta_{\C} \times \Delta_{\G})$-labelled $\QDiamond$-tree $\langle T, (V, f_{\C}, f_{\G}) \rangle$ if, and only if, the UCT $\U_{\C}'$ accepts a full $(\Sigma \times \Delta_{\C}' \times \Delta_{\G}')$-labelled $\QDiamond$-tree $\relabel \big( \langle T, (V, f_{\C}, f_{\G} ) \rangle \big)$. Similarly, the DBN $\D_{\G}^p$ accepts a full $(\Sigma \times \Delta_{\C} \times \Delta_{\G})$-labelled $\QDiamond$-tree $\langle T, (V, f_{\C}, f_{\G}) \rangle$ if, and only if, the DBN $\D_{\G}'$ accepts a full $(\Sigma \times \Delta_{\C}' \times \Delta_{\G}')$-labelled $\QDiamond$-tree $\relabel \big( \langle T, (V, f_{\C}, f_{\G} ) \rangle \big)$.
\end{restatable}

Crucially, $\U_{\C}'$ and $\D_{\G}'$ now run on \emph{full} $\QDiamond$-trees over the \emph{fixed} direction set $\QDiamond = Q \cup \{\text{\ding{73}}\}$, exactly as promised above. We make this precise next.

\paragraph{\textbf{Language Intersection Emptiness Check.}}\label{par:language-intersection}
\rev{As announced earlier, widening and pruning are sound but not, on their own, complete: a tree accepted by both $\widetilde{\T}_{\C}$ and $\T_{\G}'$ always gives rise to an accepted, pruned $\QDiamond$-tree, but a pruned tree in $\lang(\U_{\C}^p) \cap \lang(\D_{\G}^p)$ need not come from such a tree, since pruning is free to discard branches on which acceptance depended. We now close this gap in the direction we actually need: instead of asking whether every pruned tree comes from an original one, we show that whenever the pruned automata $\U_{\C}'$ and $\D_{\G}'$ (obtained after also simplifying strategies, above) accept a common tree, this already witnesses a common accepted tree for $\widetilde{\T}_{\C}$ and $\T_{\G}'$ -- because $\QDiamond$-trees are exactly the kind of regular tree that can be re-expanded, direction by direction, into a genuine model. Concretely, $\Upsilon$ was never more than an arbitrary, sufficiently large direction set to begin with, and $\widetilde{\T}_{\C}$ and $\T_{\G}'$ are defined uniformly for any such choice, so nothing stops us from simply taking $\Upsilon = \QDiamond$: under this choice, a full $\QDiamond$-tree accepted by $\U_{\C}'$ and $\D_{\G}'$ \emph{is}, once we undo the relabelling of strategy simplification, already a winning strategy for $\widetilde{\T}_{\C}$ and $\T_{\G}'$ on that very tree. This is enough to reduce the original, unbounded-branching problem to non-emptiness of $\lang(\U_{\C}') \cap \lang(\D_{\G}')$, a problem over trees of the small, fixed branching degree $|\QDiamond|$.}

Let a $\Sigma$-labelled $\Upsilon$-tree $\langle T, V \rangle$ be accepted by both the ACS $\widetilde{\T}_{\C}$ and the NBS $\T_{\G}'$.
Let $\langle T , f_{\C} \rangle$ be the winning strategy of $\widetilde{\T}_{\C}$ and $\langle T , f_{\G} \rangle$ be the winning strategy of $\T_{\G}'$ over this tree.
From the previous lemmas in \cref{app:constructions}, we have that if a tree $\langle T, V \rangle$ is accepted by both $\widetilde{\T}_{\C}$ and $\T_{\G}'$, then the tree $\langle T', V' \rangle = \prune\Big(\wide_{\QDiamond}\big( \langle T, (V, f_{\C}, f_{\G} ) \rangle \big)\Big)$ is accepted by both $\U_{\C}^{p}$ and $\D_{\G}^{p}$, and the tree $\relabel(\langle T', V' \rangle )$ is accepted by both $\U_{\C}'$ and $\D_{\G}'$: 

\begin{cor}\label{corollary:QGFM-main-theorem-impies-with-Upsilon}
We have that $\lang(\widetilde{\T}_{\C}) \cap \lang(\rev{\T'_\G}) \not= \emptyset$ implies $\lang(\U_{\C}') \cap \lang(\D_{\G}') \not= \emptyset$.
\end{cor}

As argued before, the other directions of \cref{lemma:prune-candidate} and \cref{lemma:prune-GFM} do not necessarily hold.
The other direction of \cref{corollary:QGFM-main-theorem-impies-with-Upsilon}, however, holds, because a tree in the intersection of $\lang(\U_{\C}')$ and $\lang(\D_{\G}')$ is accepted by both $\widetilde{\T}_{\C}$ and $\rev{\T'_\G}$ with $\Upsilon = \QDiamond$.

\begin{restatable}{prop}{propositionQGFMMainTheoremImpliedby}\label{proposition:QGFM-main-theorem-impliedby}
For any full ($\Sigma \times \Delta_{\C}' \times \Delta_{\G}'$)-labelled $\QDiamond$-tree $\langle T, (V, f_{\C}', f_{\G}') \rangle$ in $\lang(\U_{\C}')$ and $\lang(\D_{\G}')$, respectively, we have that the full $\Sigma$-labelled $\QDiamond$-tree $\langle T, V \rangle$ is accepted by $\widetilde{\T}_{\C}$ and $\rev{\T'_\G}$, respectively. 
\end{restatable}

\begin{proof}

If $\langle T, (V, f_{\C}', f_{\G}') \rangle$ is accepted by $\U_\C'$ then, due to the construction of $\U_\C'$, $\langle T,V\rangle$ is accepted by $\widetilde{\T}_{\C}$ with $\langle T, f_{\C} \rangle$ as the winning strategy, where $f_{\C}(w)(q) = \left \{ \begin{array}{ll}
\Box & \mbox{if $f_{\C}'(w)(q) = \Box$}\\
q & \mbox{if $f_{\C}'(w)(q) = \Diamond$}
\end{array}
\right .$ for all $w \in T$ and $q \in Q$.
This is because the run tree of $\widetilde{\T}_{\C}$ over $\langle T,V\rangle$ with this strategy is the same as the run tree of $\U_{\C}'$ over $\langle T, (V, f_{\C}', f_{\G}') \rangle$, and $\widetilde{\T}_{\C}$ and $\U_{\C}'$ have the same acceptance condition. 

Likewise, if $\langle T, (V, f_{\C}', f_{\G}') \rangle$ is accepted by $\D_\G'$, then $\langle T,V\rangle$ is accepted by
${\T}_{\G}$ with $\langle T, f_{\G} \rangle$ as the winning strategy, where $f_{\G}(w) = \left \{ \begin{array}{ll}
(\Box, q) & \mbox{if $f_{\G}'(w) = q \in F_{\G}$}\\
(\text{\ding{73}}, q) & \mbox{if $f_{\G}'(w) = q \not\in F_{\G}$}
\end{array}
\right .$ for all $w \in T$. 
This is because the run tree of ${\T}_{\G}$ over $\langle T,V \rangle$ with this strategy is the same as the run tree of $\D_{\G}'$ over $\langle T, (V, f_{\C}', f_{\G}') \rangle$, and ${\T}_{\G}$ and $\D_{\G}'$ have the same acceptance condition.
This concludes the proof.
\end{proof}

In the next step, we determinise $\U_{\C}'$ and obtain a deterministic parity tree automaton (DPT) $\D_{\C}$ such that $\lang(\U_{\C}') = \lang(\D_{\C})$.
For universal automata, this is a standard transformation that works on the word level, as the UCT can be viewed as individual UCW running along every individual branch of the tree.
It is therefore a standard operation in synthesis to determinise these automata by dualising them to nondeterministic word automata, determinising them to deterministic parity automata \cite{Piterman07}, and dualising the resulting DPW by increasing all colours by $1$, which defines a DPT $\D_{\C}$, which is language equivalent to $\U_\C'$.

To check whether the candidate NBW $\C$ is QGFM, it now suffices to check the emptiness of $\lang(\D_{\C}) \cap \lang(\D_{\G}')$.
To this end, we construct a deterministic tree automaton $\D$ that recognises the language $\lang(\D) = \lang(\D_{\C}) \cap \lang(\D_{\G}')$. 

The DPT $\D$ is constructed as follows: the new states are triples $(q_\C,q_\G,i)$, where $q_\C$ is a state of $\D_\C$, and updated according to the rules of $\D_\C$, $q_\G$ is a state from $\D_\G'$, and updated according to the rules of $\D_\G'$, and $i$ serves as memory for the highest colour that occurs in a run of $\D_\C$ between two accepting transitions (those with colour $2$) of $\D_\G'$.

That is, when we are in a state $(q_\C,q_\G,i)$, $\sigma \in \Sigma'$ (where $\Sigma'$ is the alphabet over which $\D$, $\D_\C$, $\D_\G'$ are interpreted), and $\upsilon \in \QDiamond$, then
$\big(\upsilon,(q_\C',q_\G',i'), j\big)$ is a conjunct of $\delta_{\D}\big((q_\C,q_\G,i),\sigma\big)$ if, and only if,
(1) $\delta_{\D_\C}(q_\C,\sigma)$ has a conjunct $(\upsilon, q_\C',i_\C)$,
(2) $\delta_{\D_\G'}(r,\sigma)$ has a conjunct $(\upsilon, q_\G',i_\G)$, and
(3) either $i_\G=1$, $i' = \max\{i, i_\C\}$, and $j = 1$,
or $i_\G=2$, $i'=i_\C$, and $j = i$.

This way, if there are only finitely many accepting transitions on a path in the run tree of $\D_\G'$, then the dominating colour of the same path in the run tree of $\D$ is $1$, and
if there are infinitely many accepting transitions on a path in the run tree of $\D_\G'$, then the dominating colour of the same path in the run tree of $\D$ is the dominating colour of the same path in the run tree of $\D_\C$.

With $\D$ in hand, we can now establish the complexity bound announced in \cref{theorem:QGFM-exptime}:
\begin{proof}[Proof of \cref{theorem:QGFM-exptime}]
The UCT $\U_{\C}'$ has $|Q|$ states. From $\U_{\C}'$, we can construct a language equivalent DPT $\D_\C$ with $O(|Q|!^2)$ states and $O(|Q|)$ colours \cite{Piterman07,Schewe09}. The GFM automaton of our choice $\G$ has $|Q_\G| = 3^{|Q|}$ states. The DBN $\D_{\G}'$ has $|Q_\G| = 3^{|Q|}$ states. By the construction of the DPT $\D$ described above, $\D$ has  $O\big(|Q|!^23^{|Q|}|Q|\big)$ states and $O(|Q|)$ colours.

We note that solving the language non-emptiness of $\mathcal D$ reduces to solving a parity game, where player {\it accept} and player {\it reject} take turns in selecting the letter on the node of the tree, and choosing the direction to follow.

This game has $O\big(|\Sigma \times \Delta_\C' \times \Delta_\G' \times Q_\D|\big)$ states (where $Q_\D$ is the set of states of $\D$) and $O(|Q|)$ colours.
Using one of the quasi polynomial time algorithms for solving parity games \cite{CaludeJKL017}, this can be solved in time polynomial in $\max\{|\Sigma|,|Q|!\}$.

Since the previous steps (constructing $\D_\C$, $\D_{\G}'$, and $\D$) are dominated in cost by this parity-game solve, deciding non-emptiness of $\D$ -- and hence QGFM-ness -- takes time polynomial in $\max\{|\Sigma|,|Q|!\}$ overall.
\end{proof}

	\section{Membership in EXPTIME for GFM}\label{section:EXPTIME-algorithm}
	In this section, we start out with showing a sufficient (\cref{lemma:GFM-NBW-if}) and necessary (\cref{lemma:GFM-NBW-onlyif}) criterion for a candidate NBW to be GFM in \cref{ssec:decidingGFM}.
	
	We show in \cref{ssec:QGFMisGFM} that this criterion is also sufficient and necessary for QGFM-ness.
	This implies that GFM-ness and QGFM-ness collapse, so that the {\sf EXPTIME} decision procedure from \cref{section:QGFM-procedure} can be used to decide GFM-ness, and the {\sf PSPACE}-hardness from \cref{section:PSPACE-hardness} extends to QGFM.
	
	\subsection{Key Criterion for GFM-ness}\label{ssec:decidingGFM}
	
	To establish a necessary and sufficient criterion for GFM-ness, we construct two safety automata, $\mathcal S$ and $\mathcal T$\rev{, that is, we make all states in those two automata final}. These automata can be viewed as NFAs where convenient.
	
	
	Given a candidate NBW $\C$, we define some notions for the states and transitions. 
    \rev{Given a state $q$ of $\C$, $\C_q$ is the automaton obtained from $\C$ by making $q$ the initial state.}
    We say a state $q$ of $\C$ is productive if $\lang(\C_q) \neq \emptyset$.
    A state $q$ of the NBW $\C$ is called QGFM if the automaton $\mathcal C_q$ is QGFM. A transition $(q, \sigma, r)$ is called residual if $\lang(\C_{r}) = \sigma^{-1}\lang(\C_{q})$ 
	\cite{KuperbergS15,BagnolKuperberg2018}.  
	In general, $\lang(\C_{r}) \subseteq \sigma^{-1}\lang(\C_{q})$ holds.
	See \cref{fig:NBW} for an example of non-residual transitions. Selecting either of the two transitions from $q_0$ will lose language: when selecting the transition to the left, the word $a\cdot b^\omega$ is no longer accepted. Likewise, when selecting the transition to the right, the word $a^\omega$ is no longer accepted. Thus, this automaton cannot make the decision to choose the left or the right transition, and neither $(q_0, a, q_1)$ nor $(q_0, a, q_2)$ is a residual transition.

	Now we are ready to define $\mathcal S$ and $\mathcal T$. In the NBW $\mathcal S$, we include the states from the candidate NBW $\C$ that are productive and QGFM at the same time. We only include the residual transitions (in $\C$) between those states. In the NBW $\mathcal T$, we include only the productive states of $\C$ and the transitions between them. 
    \rev{The initial state of $\mathcal S$ (respectively, $\mathcal T$) is $q_0$, the initial state of $\C$, whenever $q_0$ is among the states retained for $\mathcal S$ (respectively, $\mathcal T$); otherwise, we let $\mathcal S$ (respectively, $\mathcal T$) be the automaton recognising the empty language. 
    }
    \rev{Now we are ready to define the key criterion:
	\begin{equation*}		
    \lang(\mathcal S) = \lang(\mathcal T) \text{ and } \mathcal S \text{ is GFG.} 		
    \tag{\textbf{Key Criterion}}\label{eq:key-criterion}	\end{equation*}}
    
    We first show that this criterion is sufficient for $\C$ to be GFM.
    Although GFG as a general property can be subtle, $\mathcal S$ is a safety automaton, and GFG safety automata are determinisable by pruning which can be checked in polynomial time~\rev{\cite[Theorem~17]{DBLP:conf/fsttcs/BokerKS17}}.
    Similar to the proof of \cref{lemma:NBW-GFM}, to show the NBW $\C$ is GFM, we show there exists a strategy for $\C$ such that the syntactic and semantic chance of winning are the same for any MC. 
	
	
	

	\begin{restatable}{lem}{lemmaGFMNBW}\label{lemma:GFM-NBW-if}
		If \rev{the candidate NBW $\C$ satisfies the key criterion}, then $\C$ is GFM.
	\end{restatable}
	
	\begin{proof}
		As $\mathcal T$ contains all states and transitions from $\mathcal S$, $\lang(\mathcal S) \subseteq \lang(\mathcal T)$ always holds.
		We assume that $\lang(\mathcal S) \supseteq \lang(\mathcal T)$ holds and $\mathcal S$ is GFG.
			
		By \cref{theorem:MC-MDP-interchangable}, to show $\C$ is GFM, it suffices to show that $\C$ is good for an arbitrary MC $\M$ with initial state $s_0$. 
        We first determinise $\C$ into \rev{a language-equivalent deterministic parity word automaton (DPW) $\D$, which is possible for every NBW} \cite{Piterman07}. Since $\D$ is deterministic and $\lang(\D) = \lang(\C)$, we have $\mathrm{PSyn}_{\D}^{\M}(s_0) = \mathrm{PSem}_{\D}^{\M}(s_0) = \mathrm{PSem}_{\C}^{\M}(s_0)$. Since $\mathrm{PSem}_{\C}^{\M}(s_0) \ge \mathrm{PSyn}_{\C}^{\M}(s_0)$ always holds, we establish the equivalence of syntactic and semantic chance of winning for $\M \times \C$ by proving $\mathrm{PSyn}_{\C}^{\M}(s_0) \ge \mathrm{PSyn}_{\D}^{\M}(s_0)$.

        We consider the MC $\M' = \M \times \D$; naturally, $\M$ and $\M'$ have the same probability of success,
        that is, $\mathrm{PSyn}_{\C}^{\M}(s_0) = \mathrm{PSyn}_{\C}^{\M'}(s_0')$, where $s_0'$ is the initial state of $\M'$.
        It suffices to show that there exists a strategy for~$\M' \times \C$ such that, for every accepting run~$r$ of~$\M \times \D = \M'$, the corresponding run induced by this strategy in~$\M' \times \C$ is accepting.

		Consider an accepting run $r$ of $\M \times \D = \M'$. 
        It must enter an accepting end-component of $\M'$ eventually. 
		Let the run $r$ be $ (s_0, q_0^{\D}) (s_1, q_1^{\D}) \cdots$ and it enters the accepting end-component on reaching the state $(s_n, q_n^{\D})$.  
		\rev{This is where we use the assumption $\lang(\mathcal S) \supseteq \lang(\mathcal T)$: since $r$ is an accepting run of $\D$, the word $\ell(s_0 s_1 \cdots)$ it produces lies in $\lang(\D) = \lang(\C) = \lang(\mathcal T)$ (the last equality holds because $\mathcal T$ only omits non-productive states, which cannot occur on any accepting run), and hence, by our assumption, also in $\lang(\mathcal S)$. 
        In particular, its prefix $\ell(s_0 \cdots s_{n-1})$ is a prefix of a word in $\lang(\mathcal S)$, which guarantees that $\mathcal S$, following its GFG strategy, never blocks while $r$ is read up to $(s_n, q_n^{\D})$. 
        This is exactly what lets us define the strategy for $\C$ as follows:} before $r$ enters an accepting end-component of $\M'$, $\C$ follows the GFG strategy for $\mathcal S$ to stay within the states that are productive and QGFM. 
		Upon reaching an accepting end-component of $\M'$, the run $r$ is in state $(s_n, q_n^{\D})$, assume the run for $\M' \times \C$ is in state $(s_n, q_n^{\D}, q_n^{\C})$ at this point. We then use the QGFM strategy of $\C_{q_n^{\C}}$ from here since $q_n^{\C}$ is QGFM. 
		

		We briefly explain why this strategy for $\M' \times \C$ would lead to $\mathrm{PSyn}_{\C}^{\M'}(s_0) \ge \mathrm{PSyn}_{\D}^{\M}(s_0)$. \rev{Write $\widehat{\D} := \D_{q_n^{\D}}$ and $\widehat{\C} := \C_{q_n^{\C}}$ for the automata obtained from $\D$ and $\C$ by making $q_n^{\D}$ and $q_n^{\C}$, respectively, the initial state.} Since $(s_n, q_n^{\D})$ is in the accepting end-component of $\M' = \M \times \D$, we have $\mathrm{PSem}_{\widehat{\D}}^{\M}(s_n) = \mathrm{PSyn}_{\widehat{\D}}^{\M}(s_n) = 1$ by \cref{theorem:end-component-properties}. 		We also have $\mathrm{PSem}_{\widehat{\C}}^{\M'}(s_n) = 1$ so that $\mathrm{PSyn}_{\widehat{\C}}^{\M'}(s_n) = \mathrm{PSem}_{\widehat{\C}}^{\M'}(s_n) = 1$ as $q_n^{\C}$ is QGFM.		This is because $\lang(\widehat{\C}) = \lang(\widehat{\D}) = w^{-1}\lang(\C)$ where $w = \ell(s_0s_1\cdots s_{n-1})$, which can be proved by induction: 
        
		Induction Basis ($w=\varepsilon$): It trivially holds that $\lang(\C)=\lang(\C_{q_0^\C})=\varepsilon^{-1}\lang(\C)$ where $q_0^\C$ is the initial state of $\C$, and $q_0^{\C}$ is productive and QGFM if $\lang(\mathcal S) \neq \emptyset$.
        Otherwise, $\lang(\mathcal S) = \emptyset$ and there is nothing to show.
		
		Induction Step ($w \mapsto w {\cdot} \sigma$):
		If $w {\cdot} \sigma$ is not a prefix of a word in $\lang(\mathcal S)$, there is nothing to show.
		If $w {\cdot} \sigma$ is a prefix of a word in $\lang(\mathcal S)$, then obviously $w$ is a prefix of the same word in $\lang(\mathcal S)$, and $\mathcal S$ was in a state $q$, which is QGFM and with $\lang(\C_q)=w^{-1}\lang(\C)$.
		As $\mathcal S$, with the transition function denoted by $\delta_{\mathcal S}$, is following a GFG strategy, we have that the GFG strategy will select a state $r \in \delta_{\mathcal S}(q,\sigma) \neq \emptyset$.
		By construction, we have 
		$\lang(\C_{r}) = \sigma^{-1}\lang(\C_q)$, 
		which --- together with $\lang(\C_q) = w^{-1}\lang(\C)$, implies $\lang(\C_{r})=(w {\cdot} \sigma)^{-1} \lang(\C)$.
        \qedhere
	\end{proof}

	\begin{figure}[t]
		\centering
		\begin{tikzpicture} [->, node distance = 2cm, auto]
			\usetikzlibrary{positioning,arrows,automata}
			\tikzstyle{BoxStyle} = [draw, circle, fill=black, scale=0.4,minimum width = 1pt, minimum height = 1pt]
			
			\node[rectangle,draw,dashed, minimum width = 6cm, minimum height = 1.7cm] (r2) at (2.8,-0.35) {};
			\node at (5.4, -1) {$\M_{\Sigma}$};
			
			\node[state,fill=orange!40] (sa) at (1,0) {$s_a$};
			\node[BoxStyle] (bsa) at (1, -1) {};
			\node[state,fill=green!30] (sb) at (4,0) {$s_b$};
			\node[BoxStyle] (bsb) at (4, -1) {};
			\node[state] (s0q) at (1,-2.5) {$s_0^{q}$};
			\node[state] (s0t) at (4,-2.5) {$s_0^{t}$};
			\node[rectangle,draw,dashed, minimum width = 1.5cm, minimum height = 2cm] (r1) at (1,-3) {};
			\node at (1, -3.5) {$\M_{q}$};
			\node[rectangle,draw,dashed, minimum width = 1.5cm, minimum height = 2cm] (r2) at (4,-3) {};
			\node at (4, -3.5) {$\M_{t}$};

			\path [->] (sa) edge [near end, left] node {$\frac{1}{4}$} (s0q)
			(bsa) edge [near end, above] node {$\frac{1}{4}$} (sb)
			(bsa) edge [below] node [xshift=-1em,yshift=.7em] {$\frac{1}{4}$} (s0t)
			(sb) edge [near end, right] node {$\frac{1}{4}$} (s0t)
			(bsb) edge [below] node [xshift=1em,yshift=.7em] {$\frac{1}{4}$} (s0q)
			(bsb) edge [near end, above] node {$\frac{1}{4}$} (sa);
			
			\path[->] (bsa) edge [out=180,in=195,looseness=2] node [midway, left] {$\frac{1}{4}$} (sa);
			
			\path[->] (bsb) edge [out=0,in=-15,looseness=2] node [midway, right] {$\frac{1}{4}$} (sb);
			
		\end{tikzpicture}
		\caption{An example MC in the proof of \cref{lemma:GFM-NBW-onlyif}. In this example, $\Sigma = \{a, b\}$. Also, the state $q$ of the candidate NBW $\C$ is the only state that is not QGFM and the transition $t$ of the NBW $\C$ is the only non-residual transition. 
        }
		\label{fig:non-GFM-MC}  
	\end{figure}
	
	In order to show that this requirement is also necessary, we build an MC witnessing that $\C$ is not GFM in case the criterion is not satisfied. An example MC is given in \cref{fig:non-GFM-MC}.
	We produce the MC by parts. It has a central part denoted by $\M_{\Sigma}$. The state space of $\M_{\Sigma}$ is $S_{\Sigma}$ where $S_{\Sigma} = \{s_{\sigma} \mid \sigma \in \Sigma\}$. Each state $s_{\sigma}$ is labelled with $\sigma$ and there is a transition between every state pair.
	
	For every state $q$ that is not QGFM, we construct an MC $\M_q = \langle S_q,P_q,\Sigma,\ell_q\rangle$ 
	from \cref{section:QGFM-procedure} witnessing that $\C_q$ is not QGFM, 
	that is, from a designated initial state $s_0^q$, 
	$\mathrm{PSyn}_{\C_q}^{\M_q}(s_0^q) \neq \mathrm{PSem}_{\C_q}^{\M_q}(s_0^q) = 1$,
	and for every non-residual transition $t=(q,\sigma,r)$ that is not in $\mathcal S$ due to 
	$\lang(\C_r) \neq \sigma^{-1}\lang(\C_q)$,
	we construct an MC $\M_t  = \langle S_t,P_t,\Sigma,\ell_t\rangle$ such that, from an initial state $s_0^t$, there is only one ultimately periodic word $w_t$ produced, such that 
	$w_t \in \sigma^{-1}\lang(\C_q) \setminus \lang(\C_r)$.

	Finally, we produce an MC $\M$, whose states are the disjoint union of the MCs 
	$\M_\Sigma$, $\M_q$ and $\M_t$ from above. 
	The labelling and transitions within the MCs $\M_q$ and $\M_t$ are preserved while, from the states in $S_\Sigma$, $\M$ also transitions to all initial states of the individual $\M_q$ and $\M_t$ from above.
	It remains to specify the probabilities for the transitions from $S_\Sigma$: any state in $S_\Sigma$ transitions to its successors uniformly at random.
	
	\begin{lem}\label{lemma:GFM-NBW-onlyif}
		If the candidate NBW $\C$ does not satisfy the key criterion, then $\C$ is not GFM.
	\end{lem}
	
	\begin{proof}
		
		Assume $\lang(\mathcal S) \neq \lang(\mathcal T)$. There must exist a word $w = \sigma_0, \sigma_1, \ldots \in \lang(\T) \setminus \lang(S)$. 
		
		Let us use $\M$ with initial state  $s_{\sigma_0}$ as the MC which witnesses that $\C$ is not GFM. We first build the product MDP $\M \times \C$. There is a non-zero chance that, no matter how the choices of $\C$ (thus, the product MDP $\M \times \C$) are resolved, a state sequence $(s_{\sigma_0},q_0),(s_{\sigma_1},q_1),\ldots,(s_{\sigma_k},q_k)$ with $k \geq 0$ is seen,  and $\C$ selects a successor $q$ such that $(q_k,\sigma_k,q)$ is not a transition in $\mathcal S$.
		
		For the case that this is because $\C_q$ is not QGFM, we observe that there is a non-zero chance that the product MDP moves to $(s_0^q,q)$, such that $\mathrm{PSyn}_{\C}^{\M}(s_{\sigma_0}) < \mathrm{PSem}_{\C}^{\M}(s_{\sigma_0})$ follows.
		
		For the other case that this is because the transition $t = (q_k,\sigma_k,q)$ is non-residual, that is, 
		$\lang(\C_q) \neq {\sigma_{k}}^{-1}\lang(\C_{q_k})$,
		we observe that there is a non-zero chance that the product MDP moves to $(s_0^t,q)$, such that $\mathrm{PSyn}_{\C}^{\M}(s_{\sigma_0}) < \mathrm{PSem}_{\C}^{\M}(s_{\sigma_0})$ follows.
		
		
		For the case that $\mathcal S$ is not GFG, no matter how the nondeterminism of $\C$ is resolved, there must be a shortest word $w = \sigma_0,\ldots,\sigma_k$ ($k \ge 0$) such that $w$ is a prefix of a word in  $\lang(\mathcal S)$, but the selected control leaves $\mathcal S$. For this word, we can argue in the same way as above.
	\end{proof}
	
    \cref{lemma:GFM-NBW-if} and \cref{lemma:GFM-NBW-onlyif} suggest that GFM-ness of an NBW $\C$ can be decided in {\sf EXPTIME} by checking whether the key criterion 
    holds. \rev{This membership follows because each ingredient of the criterion can be checked within {\sf EXPTIME}: which states of $\C$ are productive is decided by a standard reachability check in NL; 
    whether a state $q$ is QGFM is decided by one call, per state, to the {\sf EXPTIME} procedure of \cref{section:QGFM-procedure} applied to $\C_q$ -- and {\sf EXPTIME} is closed under polynomially many such calls, so deciding QGFM-ness of every state of $\C$ is still in {\sf EXPTIME}; 
    whether a transition between productive, QGFM states is residual, and whether $\lang(\mathcal S) = \lang(\mathcal T)$ holds, are both language-equivalence checks between automata no larger than $\C$, hence decidable in {\sf PSPACE} $\subseteq$ {\sf EXPTIME}; and whether the safety automaton $\mathcal S$ is GFG is decidable in polynomial time by determinisation by pruning \cite[Theorem~17]{DBLP:conf/fsttcs/BokerKS17}, since GFG-ness of safety automata is tractable. 
    Composing these {\sf EXPTIME} (and cheaper) checks yields an {\sf EXPTIME} procedure for GFM-ness. 
    As we show next, however, this route turns out not to be necessary: once \cref{ssec:QGFMisGFM} establishes that QGFM-ness and GFM-ness coincide, the {\sf EXPTIME} procedure from \cref{section:QGFM-procedure} can be applied directly to $\C$ to decide GFM-ness, without first constructing $\mathcal S$ and $\mathcal T$ at all.}
	
	\subsection{QGFM = GFM}\label{ssec:QGFMisGFM}
	
	
	To show that QGFM = GFM, we show that the same criterion from the previous section is also sufficient and necessary for QGFM. By definition, if an NBW is GFM then it is QGFM. Thus, the sufficiency of the criterion follows from \cref{lemma:GFM-NBW-if}. We are left to show the necessity of the criterion. To do that, we build an MC $\M'$ witnessing that $\C$ is not QGFM in case the criterion of \cref{lemma:GFM-NBW-onlyif} is not satisfied. We sketch in \cref{fig:non-QGFM-MC} the construction of $\M'$.
	
	\begin{figure}[t]
		\centering
		\begin{tikzpicture} [->, node distance = 2cm, auto]
			\usetikzlibrary{positioning,arrows,automata}
			\tikzstyle{BoxStyle} = [draw, circle, fill=black, scale=0.4,minimum width = 1pt, minimum height = 1pt]
			
			\node[rectangle,draw,dashed, minimum width = 3.5cm, minimum height = 2cm] (r3) at (1,-0.3) {};
			\node at (1, -1) {$\M_{\Sigma}'$};
			\node[state,fill=orange!40] (sa) at (-0.1,0) {$s_a, p$};
			\node[state,fill=green!30] (sb) at (1.8,0) {$s_b, p$};

			\node[rectangle,draw,dashed, minimum width = 1.5cm, minimum height = 2cm] (r4) at (4,-0.3) {};
			\node at (4, -1) {$\M_{p, a}$};
			
			\node[rectangle,draw,dashed, minimum width = 1.5cm, minimum height = 2cm] (r5) at (6,-0.3) {};
			\node at (6, -1) {$\M_{p, b}$};
			\node at (8, -0.3) {$\cdots$};
			
			\node[state] (s0q) at (1,-2) {$s_0^{q}$};
			\node[state] (s0t) at (4,-2) {$s_0^{t}$};
			\node[rectangle,draw,dashed, minimum width = 1.5cm, minimum height = 2cm] (r1) at (1,-2.5) {};
			\node at (1, -3) {$\M_{q}$};
			\node at (2.5, -2.5) {$\cdots$};
			
			\node[rectangle,draw,dashed, minimum width = 1.5cm, minimum height = 2cm] (r2) at (4,-2.5) {};
			\node at (4, -3) {$\M_{t}$};
			\node at (5.5, -2.5) {$\cdots$};
			
			\path[->] (sa) edge [out=30,in=150,looseness=1] node [midway, left] {} (r4);
			\path[->] (sb) edge [out=30,in=150,looseness=1] node [midway, left] {} (r5);
		\end{tikzpicture}
		\caption{An illustration of the MC in the proof of \cref{lemma:QGFM-NBW-onlyif}. In this example, we have $\Sigma = \{a, b\}$. The new central part $\M_{\Sigma}'$ is obtained by removing the states that have no outgoing transitions in the cross product of $\M_{\Sigma}$ and $\D'$. 
        For each state $(s_{\sigma},p)$ of the central part, we construct an MC $\M_{p, \sigma}$ using $\D'$. There is a transition from each state $(s_{\sigma}, p)$ of $\M_{\Sigma}'$ to the initial state of $\M_{p, \sigma}$. The MCs $\M_q$ and $\M_t$ are as before. Whether there is a transition from a state from $\M_{\Sigma}'$ to the MCs $\M_q$ and $\M_t$ is determined by the overestimation provided by $\D'$. }
		\label{fig:non-QGFM-MC}  
	\end{figure}

	The principle difference between the MC $\M'$ constructed in this section and $\M$ from the previous section is that the new MC $\M'$ needs to satisfy that $\mathrm{PSem}_{\C}^{\M'}(s_0) = 1$ ($s_0$ is the initial state of $\M'$), while still forcing the candidate NBW $\mathcal C$ to make decisions that lose probability of success, leading to  $\mathrm{PSyn}_{\C}^{\M'}(s_0) < 1$.
	This makes the construction of $\M'$ more complex, 
	but establishes that qualitative and full GFM are equivalent properties.
	
	The MC will also be constructed by parts and it has a central part. It will also have the MCs $\M_q$ for each non-QGFM state $q$ and $\M_t$ for each non-residual transition $t$ from the previous section. We now describe the three potential problems of $\M$ of the previous section before presenting the possible remedies.     \begin{enumerate}        
    \item 	The first potential problem is in the central part as it might contain prefixes that cannot be extended to words in $\lang(\C)$. Such prefixes should be excluded.    \rev{For example, suppose $\Sigma = \{a,b\}$ and $\delta(q_0,b)$ leads only to a non-productive state of $\C$ (one from which no accepting run of $\C$ is possible), while $\delta(q_0,a)$ leads to a productive state. The naive central part $\M_\Sigma$ nonetheless still moves to the state labelled $b$ with probability $\tfrac12$ at the very first step, since it transitions uniformly between all states of $\Sigma$; from there, no continuation can ever be accepted by $\C$, so already $\mathrm{PSem}_{\C}^{\M}(s_0) \le \tfrac12 < 1$. Prefixes like this one -- and the states of the central part they lead to -- must therefore be excluded.}                
    \item 	The second problem is caused by the transitions to all MCs $\M_q$ and $\M_t$ from every state in the central part\rev{, regardless of whether $\C$ could actually be in state $q$ (or have just taken transition $t$) given the prefix read so far}.                
    \rev{For example, suppose the non-QGFM state $q_a$ of $\C$ is only reachable by first reading $a$. The naive construction still offers, from every state of the central part -- including the one labelled $b$ -- a transition to the initial state of $\M_{q_a}$. Taking it would run the dynamics of $\M_{q_a}$, which are built on the assumption that $\C$ is really in state $q_a$, on top of a history under which $\C$ could never actually be in $q_a$, severing the link between the word produced by $\M'$ and any genuine run of $\C$. Such transitions must therefore be restricted to central-part states whose history is consistent with $\C$ actually reaching $q$ (or taking $t$).}                
    \item 	Removing the transitions to some of the $\M_q$ and $\M_t$, as required to solve the second problem, can in turn create a third problem: a state of the central part might then be left with no transition to any $\M_q$ or $\M_t$ at all, if no non-QGFM state or non-residual transition happens to be reachable given that particular prefix. 
    \rev{For example, if $\C$ has no non-QGFM state and no non-residual transition reachable after reading $b$, then, once the second problem is fixed, the state labelled $b$ loses all its transitions to any $\M_q$ or $\M_t$, so that remaining in the central part forever becomes the only option -- which would prevent the construction from having semantic probability of one. We must therefore also ensure that every state of the central part retains some way of eventually leaving it with a word in the language of the automaton.}  
    \end{enumerate}
    
    We construct two automata, $\mathcal D$ and $\mathcal D'$, to address all the problems we have identified.
    For the construction of the MC, it is required that all finite runs starting from the initial state $s_0$ --- before transitioning to $\M_q$ or $\M_t$ --- that can be extended to a word in the language of $\mathcal{C}$ are retained.
	The language of all such initial sequences is a safety language, and it is easy to construct an automaton that (1) recognises this safety language and (2) retains the knowledge of how to complete each word in the language of $\C$.
	To this end, we first determinise $\C$ into a DPW $\mathcal D$ \cite{Piterman07}. 
    We then remove all non-productive states from $\mathcal D$ and mark all remaining states as final, yielding the safety automaton $\mathcal D'$. 	
    We now use these two automata to address each problem in turn:
    \begin{enumerate}
        \item To address the first problem, we build a cross product MC of $\M_{\Sigma}$ (the central part of $\M$ in \cref{ssec:decidingGFM}) and $\D'$. 
        We then remove all the states in the product MC without any outgoing transitions and make the resulting MC the new central part denoted by $\M_{\Sigma}'$. 
        Every state in $\M_{\Sigma}'$ is of the form $(s_\sigma, p)$ where $s_{\sigma} \in S_{\Sigma}$ is from $\M_{\Sigma}$ and $p \in \D'$.  
        \item  
        The states of the deterministic automaton $\D$ (and thus those of $\D'$) also provide information about the possible states of $\C$ that could be after the prefix we have seen so far. 
        To address the second problem, we use this information to overestimate whether $\C$ could be in some state $q$, or use a transition $t$, which in turn is good enough for deciding whether or not to transition to the initial states of $\M_q$ (resp.\ $\M_t$) from every state of the new central part.
        \item 	To address the third problem, we build, for every state $p$ of $\mathcal D'$ and every letter $\sigma \in \Sigma$ such that $\sigma$ can be extended to an accepted word from state $p$, an MC $\M_{p, \sigma}$ that produces a single lasso word $w_{p, \sigma}$ with probability one. 
        The word $\sigma {\cdot} w_{p, \sigma}$ 
        will be accepted from state $p$, that is, $\sigma {\cdot} w_{p, \sigma} \in \lang(\D_p)$. 
	    From every state $(s_\sigma, p)$ of the central part, there is a transition to the initial state of $\M_{p,\sigma}$.
    \end{enumerate}
	    
		
	
	The final MC $\M'$ transitions uniformly at random, from a state $(s_{\sigma}, p)$ in $\M_{\Sigma}'$, to one of its successor states, which comprise 
	the initial state of $\M_{p, \sigma}$ and the initial states of some of the individual MCs $\M_q$ and $\M_t$. 
	
	\begin{lem}\label{lemma:QGFM-NBW-onlyif}  
		If the candidate NBW $\C$ does not satisfy the key criterion, then $\C$ is not QGFM.
	\end{lem}
	
	\begin{proof}

        The proof of the difference in the probability of winning in case $\lang(\mathcal S) \neq \lang(\mathcal T)$ or in case $\mathcal S$ is not GFG \rev{is} the same as in \cref{lemma:GFM-NBW-onlyif}.

        \rev{We additionally have to show that $\mathrm{PSem}_{\C}^{\M'}(s_0) = 1$.} But this is easily provided by the construction: 
        when we move on to some $\M_{p, \sigma}$, $\M_q$, or $\M_t$, we have \emph{sure}, \emph{almost-sure}, and \emph{sure} satisfaction, respectively, of the property by construction\rev{ -- that is, respectively, every word produced is accepted, a produced word is accepted with probability one, and every word produced is accepted, as we now detail.
        Concretely, $\M_{p,\sigma}$ is built from $\D'$ specifically so that every word it can produce, together with the prefix read so far, lies in $\lang(\C)$ -- that is, \emph{sure} satisfaction, since no exception is possible, regardless of probability. Each $\M_q$ was built in \cref{section:QGFM-procedure} to witness $\mathrm{PSem}_{\C_q}^{\M_q}(s_0^q) = 1$ -- only \emph{almost-sure} satisfaction, since, as $q$ is not QGFM, this need not hold for literally every word $\M_q$ could produce.
        And each $\M_t$ deterministically produces a single, fixed ultimately periodic word $w_t \in \sigma^{-1}\lang(\C_q) \setminus \lang(\C_r)$ -- so that $\sigma w_t$ is accepted by $\C_q$ -- and satisfaction is again \emph{sure}: trivially, as there is only one possible outcome to begin with.}
	    Staying forever in the central part of the new MC, meanwhile, happens with probability zero.
	\end{proof}
	
	By definition, if a candidate NBW $\C$ is GFM, it is QGFM. Together with \cref{lemma:GFM-NBW-if} and \cref{lemma:QGFM-NBW-onlyif}, we have that $\lang(\mathcal S) = \lang(\mathcal T)$ and $\mathcal S$ is GFG iff the candidate NBW $\C$ is QGFM. Thus, we have
	\begin{thm}\label{theorem:QGFM-GFM-equivalence}
		\rev{Every} NBW $\C$ is GFM if, and only if, $\C$ is QGFM.
	\end{thm}
	
		
		

	By~\cref{theorem:QGFM-exptime} and \cref{theorem:QGFM-GFM-equivalence}, we have:

	\begin{cor}\label{corollary:GFM-EXPTIME}
		The problem of whether or not a given NBW is GFM is in {\sf EXPTIME}.
	\end{cor}
	
	Likewise, by \cref{theorem:PSPACE-harness}, \cref{theorem:minimisation-GFM}, and \cref{theorem:QGFM-GFM-equivalence}, we have:
	
	\begin{cor}\label{corollary:PSPACE-harness}
		The problem of whether or not a given NBW is QGFM is {\sf PSPACE}-hard.
		Given a QGFM NBW and a bound $k$, the problem \rev{of} whether there is a language equivalent QGFM NBW with at most $k$ states is {\sf PSPACE}-hard. It is {\sf PSPACE}-hard even for (fixed) $k = 2$.
	\end{cor}

    
\section{Succinctness of GFM Automata}\label{sec:succinct}

In this section, we show that GFM automata are exponentially more succinct than GFG automata.
We also show that general nondeterministic automata are exponentially more succinct than GFM automata.
This holds even when we restrict the class of general nondeterministic automata to run-unambiguous reachability or safety automata, run-unambiguous separating safety automata, or separating automata recognising a reachability language.

\subsection{GFM vs GFG} \label{subsec:GfMvsGfG}


We first recall the standard result that NFAs are exponentially more succinct than deterministic finite automata (DFAs), and use a textbook example of a family of languages that witnesses this, namely
\[\lang_n = \{0, 1\}^{*}1\{0, 1\}^{n-1}\$\ ,\]
the language of those words $w \in \{0, 1\}^{*}\$$, where the $n$-th letter before $\$$ in $w$ is $1$.
NFAs recognising this language require $n+2$ states; see \cref{fig:NFA} for the NFA $\A_n$ defined below.
Any DFA recognising the same language requires more than $2^n$ states.
The intuition is that a DFA needs to store the previous $n$ letters in their order of occurrence in order to recognise this language.
Let $\Sigma = \{0, 1, \$\}$. Define the NFA $\A_n = (\Sigma, Q, q_0, \delta_\A, \{f\})$ where $Q = \{q_0, \ldots, q_n, f\}$ and $\delta_\A$ is defined as follows:
$\delta_\A (q_0, 0) = \{q_0\}$; 
$\delta_\A (q_0, 1) = \{q_0, q_1\}$; 
$\delta_\A (q_i, \sigma) = \{q_{i+1}\}$ for all $i \in \{1, \ldots, n-1\}$ and $\sigma \in \{0, 1\}$;
$\delta_\A (q_n, \$) = \{f\}$.

\begin{figure}[t]
\centering
\begin{tikzpicture} [->, node distance = 2cm, auto,initial text = {}]
    \node[initial left,state] (q0) at (0,1) {$q_0$};
    \node[state] (q1) at (2,1) {$q_1$};
    \node[state] (q2) at (4,1) {$q_2$};
    \node[state] (q3) at (6,1) {$q_3$};
    \node (qi) at (7,1) {$\ldots$};
    \node[state] (qn) at (8,1) {$q_n$};
    \node[state, accepting] (qf) at (10,1) {$f$};

    \path (q0) edge [loop above] node {$0,1$} (q0);
    \path (q0) edge node {$1$} (q1);
    \path (q1) edge node {$0,1$} (q2);
    \path (q2) edge node {$0,1$} (q3);
    \path (qn) edge node {$\$$} (qf);
\end{tikzpicture}
\caption{The NFA $\A_n$ recognising $\lang_n$.}
\label{fig:NFA}
\end{figure}

\begin{figure}[t]
\centering
\begin{tikzpicture} [->, node distance = 2cm, auto,initial text = {}]
    \node[initial left,state] (q0) at (0,1) {$q_0$};
    \node[state] (q1) at (2,1) {$q_1$};
    \node[state] (q2) at (4,1) {$q_2$};
    \node[state] (q3) at (6,1) {$q_3$};
    \node (qi) at (7,1) {$\ldots$};
    \node[state] (qn) at (8,1) {$q_n$};
    \node[state, accepting] (qf) at (10,1) {$f$};

    \path (q0) edge [loop above] node {$0,1,\$$} (q0);
    \path (q0) edge [bend left] node [above] {$1$} (q1);
    \path (q1) edge [bend left] node [below]{$0,1$} (q2);
    \path (q1) edge [bend left] node [above] {$\$$} (q0);
    \path (q2) edge [bend left] node [below] {$0,1$} (q3);
    \path (q2) edge [bend left] node [below, near start, yshift=0.2em]{$\$$} (q0);
    \path (q3) edge [bend left]  node [below, near start, yshift=0.2em]{$\$$} (q0);
    \path (qn) edge node {$\$$} (qf);
    \path (qn) edge [bend left] node [below, near start, yshift=0.1em] {$0,1$}  (q0);
    \path (qf) edge [bend left] node [below, near start] {$0,1,\$$}  (q0);
    \path (qf) edge [bend right] node [below, near start] {$1$}  (q1);


\end{tikzpicture}
\caption{The NBW $\G_n$ recognising $\lang_n^{\omega}$.}
\label{fig:NBA}
\end{figure}




We then build an NBW from the standard NFA for this language that recognises the $\omega$-language of infinite words that contain infinitely many sequences of the form $1\, \{0,1\}^{n-1} \$$,
such that the resulting automaton has the same states as the NFA, and almost the same structure (cf.\ \cref{fig:NFA,fig:NBA}).
We call this language $\lang_n^\omega$ and show that the resulting NBW that recognises it is GFM.
Any DBW or GFG NBW recognising the same language would require exponentially more states than this GFM NBW.

For any $n$, there is an NBW with $n+2$ states recognising $\lang_{n}^\omega$. 
Define the NBW $\G_n = (\Sigma, Q, q_0, \delta, \{f\})$ where $Q = \{q_0, \ldots, q_n, f\}$ and  $\delta$ is defined as follows:
$\delta (q_0, 0) = \{q_0\}$; 
$\delta (q_0, 1) = \{q_0, q_1\}$; 
$\delta (q_i, \sigma) = \{q_{i+1}\}$ for all $i \in \{1, \ldots, n-1\}$ and $\sigma \in \{0, 1\}$;
$\delta (q_n, \$) = \{f\}$;
$\delta (q_n, \sigma) = \{q_0\}$ for $\sigma \in \{0, 1\}$; 
$\delta (f, 0) = \{q_0\}$; 
$\delta (f, 1) = \{q_0,\rev{q_1}\}$; 
$\delta (q, \$) = \{q_0\}$ for all $q_n \neq q \in Q$.

\begin{lem}\label{lem:gn_correct}
$\G_n$ recognises $\lang_n^\omega$.
\end{lem}
\begin{proof}
\textbf{`$\lang(\G_n) \subseteq \lang_n^\omega$':} 
It is easy to see that by the construction of $\G_n$, to visit $f$ infinitely often, $\G_n$ has to take the path $q_0, 1, q_1, \sigma_1, q_2, \ldots, q_{n}, \$, f$ on the word $1\sigma_1\cdots\sigma_{n-1}\$ \in \lang_n$ where $\sigma_i \in \{0,1\}$ for all $i \in \{1, \ldots, n-1\}$ infinitely often.

\textbf{`$\lang_n^\omega \subseteq \lang(\G_n)$':} 
We only need to show there is an accepting run of $\G_n$ on any word in $\lang_n^\omega$.
A word $w \in \lang_n^\omega$ is of the form $w_1w_1'w_2w_2'w_3w_3'\ldots$ such that, for all $i\in \nat$, $w_i \in \{0, 1, \$\}^*$ and $w_i' \in \lang_n$.
Let $w_i' = u_i1v_i$ where $u_i \in \{0, 1\}^*$ and $v_i \in \{0, 1\}^{n-1}\$$ for all $i > 0$. We can build the run as follows.

For odd indices $i$, $\G_n$ always stays in $q_0$ on $w_iu_i$, transitions to $q_1$ on $1$ and to $f$ on $v_i$;
(at the beginning, $\G_n$ starts from $q_0$ and stays in $q_0$ on $w_1u_1$, transitions to $q_1$ on $1$ and to $f$ on $v_1$.)
For even indices $i$, $\G_n$ starts from the state $f$, then traverses to $q_0$ and stays in $q_0$ on $w_iw_i'$. 


For every $w_i' \in \lang_n$ such that $i$ is odd, the run we build will visit the accepting state $f$. 
As there are infinitely many such $w_i'$s, the run will visit the accepting state $f$ infinitely often, and thus, is accepting.  
\end{proof}

\rev{$\G_n$ is obtained from $\A_n$ by first giving $f$ the same successor states as $q_0$ -- so that $\G_n$ can immediately restart searching for another word in $\lang_n$ upon accepting one -- and then, from the resulting automaton, adding a transition to $q_0$ from every state and letter on which it would otherwise block.}
The intuition why $\G_n$ is GFM is that, if there is any sequence in $\lang_n$ in a BSCC of an MC, then it can almost surely try infinitely often to facilitate it for acceptance.
Each of these attempts will succeed with a positive probability $p>0$, so that almost surely infinitely many of these infinite attempts will succeed.
The formal proof is similar to that of \cref{lemma:NBW-GFM} and \cref{lemma:GFM-NBW-if}.

\begin{restatable}{lem}{lemGnGFM}\label{lem:gfm}
$\G_n$ is GFM.
\end{restatable}
\begin{proof}

Consider an arbitrary MC $\M$ with initial state $s_0$. We show that $\G_n$ is good for $\M$, that is, $\mathrm{PSem}_{\G_n}^{\M}(s_0) = \mathrm{PSyn}_{\G_n}^{\M}(s_0)$. It suffices to show $\mathrm{PSyn}_{\G_n}^{\M}(s_0) \ge \mathrm{PSem}_{\G_n}^{\M}(s_0)$ since by definition the converse $\mathrm{PSem}_{\G_n}^{\M}(s_0) \ge \mathrm{PSyn}_{\G_n}^{\M}(s_0)$ always holds.

First, we construct a language equivalent deterministic B\"uchi automaton (DBW) $\D_n$. 
We start with the NFA $\A_n'$, see \cref{fig:NFA2}, which recognises the language $\{w \in \{0, 1\}^{*} \mid \text { the $n$-th to the last letter in $w$ is $1$}\}$. 
The NFA $\A_n'$ is very similar to $\A_n$: it has the same states $q_i$ for $i \in \{0, \ldots, n\}$ and the same transitions between them; 
the only difference is that $q_n$ is accepting without any outgoing transitions as there is no $f$ state in $\A_n'$. 
We then determinise $\A_n'$ to a DFA $\B_n'$ by a standard subset construction and then obtain $\B_n$ from $\B_n'$ by adding a fresh $f$ state, adding a $\$$ transition from all accepting states to $f$ and making $f$ the only accepting state. 
It is easy to see that $\B_n$ is deterministic and $\lang(\A_n) = \lang(\B_n)$.  
Now we obtain the DBW $\D_n$ from the DFA $\B_n$ in the same way as we obtain the NBW $\G_n$ from the NFA $\A_n$.

\begin{figure}[ht]
\centering
\begin{tikzpicture} [->, node distance = 2cm, auto,initial text = {}]
    \node[initial left,state] (q0) at (0,1) {$q_0$};
    \node[state] (q1) at (2,1) {$q_1$};
    \node[state] (q2) at (4,1) {$q_2$};
    \node[state] (q3) at (6,1) {$q_3$};
    \node (qi) at (7,1) {$\ldots$};
    \node[state, accepting] (qn) at (8,1) {$q_n$};

    \path (q0) edge [loop above] node {$0,1$} (q0);
    \path (q0) edge node {$1$} (q1);
    \path (q1) edge node {$0,1$} (q2);
    \path (q2) edge node {$0,1$} (q3);
\end{tikzpicture}
\caption{The NFA $\A_n'$, which recognises the language $\{w \in \{0, 1\}^{*} \mid \text { the $n$-th to the last letter in $w$ is $1$}\}$.}
\label{fig:NFA2}
\end{figure}

Since $\lang(\G_n) = \lang(\D_n)$, we have that $\mathrm{PSem}_{\G_n}^{\M}(s_0) = \mathrm{PSem}_{\D_n}^{\M}(s_0)$. In addition, since $\D_n$ is deterministic, we have $\mathrm{PSem}_{\D_n}^{\M}(s_0) = \mathrm{PSyn}_{\D_n}^{\M}(s_0)$.
  
It remains to show $\mathrm{PSyn}_{\G_n}^{\M}(s_0) \ge \mathrm{PSyn}_{\D_n}^{\M}(s_0)$. 
We consider the MC $\M' = \M \times \D_n$; naturally, $\M$ and $\M'$ have the same probability of success,
that is, $\mathrm{PSyn}_{\G_n}^{\M}(s_0) = \mathrm{PSyn}_{\G_n}^{\M'}(s_0')$, where $s_0'$ is the initial state of $\M'$.

It suffices to show that there exists a strategy for~$\M' \times \G_n$ such that, for every accepting run~$r$ of~$\M \times \D_n = \M'$, the corresponding run induced by this strategy in~$\M' \times \G_n$ is accepting.

Similar to the proofs of \cref{lemma:NBW-GFM} and \cref{lemma:GFM-NBW-if},
\rev{this strategy for $\M' \times \G_n$ is to \emph{track} a fixed word $w \in \lang(\B_n)$ (and thus $w \in \lang(\A_n)$) once an accepting end-component of $\M' = \M \times \D_n$ is entered: as in the proof of \cref{lemma:NBW-GFM}, tracking simply means that the strategy for $\G_n$ tries to follow the trace of $w$ in the end-component, steering $\G_n$ along the sequence of states that reading $w$ from $q_0$ would induce, and falling back to an arbitrary choice as soon as the model produces a letter that deviates from $w$, until the next opportunity to restart tracking arises.}

Once in an accepting end-component of $\M \times \D_n$, tracking is almost surely started infinitely often, and it is thus almost surely successful infinitely often.
Thus, we have
$\mathrm{PSyn}_{\G_n}^{\M}(s_0) = \mathrm{PSyn}_{\G_n}^{\M'}(s_0')\geq \mathrm{PSyn}_{\D_n}^{\M}(s_0)$.
\end{proof}

\rev{As noted in the remark after \cref{lemma:NBW-GFM}, $\G_n$ is also a residual (semantically-deterministic) automaton, so this is again an instance of the same general fact that every residual (semantically-deterministic) NBW is GFM  \cite[Theorem~6.4]{DBLP:journals/lmcs/RadiKL25}.} 

We obtain a lower bound of $2^{n-1}$ states for DBWs recognising the same language:
\begin{thm}\label{cor:Gn-exponential-Dn}
GFM B\"uchi automata are exponentially more succinct than DBWs.   
\end{thm}
\begin{proof}
We show that any DBW recognising $\rev{\lang_n^\omega}$ must have at least $2^{n-1}$ states.

Assume, towards a contradiction, that there exists a DBW $\B$ recognising $\rev{\lang_n^\omega}$ with strictly fewer than $2^{n-1}$ states.
Consider the $2^n$ distinct words in $\{0,1\}^n$.
Starting from the initial state $q_0$ of $\B$, by the pigeonhole principle, at least three of these words lead to the same state.

Among these three words, there exist two—denoted $u_1$ and $v_1$—whose runs from $q_0$ either both visit an accepting state or both avoid accepting states entirely.
Since $u_1 \neq v_1$, let $i \ge 0$ be the first position at which they differ, and let $w_1 = 0^{i}\$$.
By construction, exactly one of the words $u_1 w_1$ and $v_1 w_1$ belongs to $\rev{\lang_n}$.
Without loss of generality, assume $u_1 w_1 \in \rev{\lang_n}$ and $v_1 w_1 \notin \rev{\lang_n}$.

Both $u_1 w_1$ and $v_1 w_1$ reach the same state of $\B$.
From this state, the argument can be iterated inductively.
At each step $k \ge 1$, we select two distinct words $u_k$ and $v_k$ that reach the same state of $\B$ and whose runs either both visit some accepting states or both avoid accepting states, together with a finite word $w_k$ such that
$u_k w_k \in \rev{\lang_n}$ and $v_k w_k \notin \rev{\lang_n}$.

This yields two infinite words
\[
u_1 w_1 u_2 w_2 u_3 w_3 \cdots
\quad\text{and}\quad
v_1 w_1 v_2 w_2 v_3 w_3 \cdots
\]
such that exactly one of them belongs to $\rev{\lang_n^\omega}$.
By construction, after reading each finite sequence
$u_k w_k$ and $v_k w_k$,
the automaton $\B$ reaches the same state and the two corresponding runs have either both visited some accepting states or both not visited any accepting state.
Hence, the two infinite runs either both visit accepting states infinitely often or both visit them only finitely often -- a contradiction.
\end{proof}

Noting that the DBW can be chosen minimal with at least $2^{n-1}$ states, it is now easy to infer that GFM B\"uchi automata are exponentially more succinct than GFG B\"uchi automata by using the result that GFG B\"uchi automata are only up to quadratically more succinct than DBW for the same language \cite{KuperbergS15}.
The size of GFG automata is therefore in $\Omega(\sqrt{2^n})=\Omega(2^{n/2})$.

\begin{cor}\label{cor:Gn-exponential}
GFM B\"uchi automata are exponentially more succinct than GFG B\"uchi automata.   
\end{cor}

\subsection{GFM vs General Nondeterminism}
\label{section:GfMvsNondeterminism}

We show that general nondeterministic automata are exponentially more succinct than GFM automata.
As we can see in the following, this even holds when the languages are restricted to reachability or safety languages.
Moreover, using the same languages, we observe that unambiguous automata --- and even more restricted subclasses thereof --- can be exponentially more succinct than GFM automata.



An automaton is \emph{unambiguous} if every word has at most one accepting run, and \emph{run-unambiguous} if every word has at most one run. 
An automaton is \emph{separated/separating} if different states have disjoint languages; separating automata are therefore unambiguous.

The reachability languages we use are the extension of $\lang_n$ to reachability:
\[\lang_{n}^r = \{w \{0,1,\$\}^{\omega} \mid w \in \lang_n\}\ ,\]
the language of infinite words that contain at least one $\$$, as well as a $1$ that occurs exactly $n$ letters before this first $\$$.
In particular, words where the first $\$$ is one of the first $n$ letters is not part of $\lang_n^r$.



An NBW $\mathcal R_n$ that recognises $\lang_n^r$ is shown in \cref{fig:NRA}.
The automaton is a reachability automaton as there is only one accepting state, which is a sink. 

\begin{figure}[ht]
\centering
\begin{tikzpicture} [->, node distance = 2cm, auto,initial text = {}]
    \node[initial left,state] (q0) at (0,1) {$q_0$};
    \node[state] (q1) at (2,1) {$q_1$};
    \node[state] (q2) at (4,1) {$q_2$};
    \node[state] (q3) at (6,1) {$q_3$};
    \node (qi) at (7,1) {$\ldots$};
    \node[state] (qn) at (8,1) {$q_n$};
    \node[state, accepting] (qf) at (10,1) {$f$};

    \path (q0) edge [loop above] node {$0,1$} (q0);
    \path (q0) edge node {$1$} (q1);
    \path (q1) edge node {$0,1$} (q2);
    \path (q2) edge node {$0,1$} (q3);
    \path (qn) edge node {$\$$} (qf);
    \path (qf) edge [loop above] node {$0,1,\$$} (qf);

\end{tikzpicture}
\caption{The reachability NBW $\mathcal R_n$ recognising $\lang_n^r$.}
\label{fig:NRA}
\end{figure}

The safety languages we use are the safety closure of an extension of words in $\lang_n$ by infinitely many $\$$s: 
\[\lang_{n}^s = \{w \$^{\omega} \mid w \in \lang_n\} \cup  \{0, 1\}^{\omega}\ , \]

An NBW $\mathcal S_n$ that recognises $\lang_n^s$ is shown in \cref{fig:NSA}.
The automaton is a safety automaton as all states are accepting. 


\begin{figure}[ht]
\centering
\begin{tikzpicture} [->, node distance = 2cm, auto,initial text = {}]
    \node[initial left,state, accepting] (q0) at (0,1) {$q_0$};
    \node[state, accepting] (q1) at (2,1) {$q_1$};
    \node[state, accepting] (q2) at (4,1) {$q_2$};
    \node[state, accepting] (q3) at (6,1) {$q_3$};
    \node (qi) at (7,1) {$\ldots$};
    \node[state, accepting] (qn) at (8,1) {$q_n$};

    \path (q0) edge [loop above] node {$0,1$} (q0);
    \path (q0) edge node {$1$} (q1);
    \path (q1) edge node {$0,1$} (q2);
    \path (q2) edge node {$0,1$} (q3);
    \path (qn) edge [loop above] node {$\$$} (qn);

\end{tikzpicture}
\caption{The safety NBW $\mathcal S_n$ recognising $\lang_n^s$.}
\label{fig:NSA}
\end{figure}

Next, we show that a GFM automaton for these languages needs to have $2^n$ states, not counting accepting or rejecting sinks.
For that, we construct a family of MCs as shown in \cref{fig:MC-big-GfM}.
\rev{We show that every automaton that recognises $\lang_n^r$ or $\lang_n^s$ and is good for this family of MCs has at least $2^n$ states.} 

\begin{figure}[ht]
\centering
\begin{tikzpicture} [->, node distance = 2cm, auto,initial text = {}]
    \node[initial left,state] (q0) at (0,1) {$\sigma_1$};
    \node[state] (q1) at (2,1) {$\sigma_2$};
    \node[state] (q2) at (4,1) {$\sigma_3$};
    \node (qi) at (5,1) {$\ldots$};
    \node[state] (q3) at (6,1) {$\sigma_n$};
    \node[state] (q00) at (8,2.45) {$0$};
    \node[state] (q11) at (8,-.45) {$1$};

    \node[state] (qf) at (8,1) {$\$$};

    \path (q0) edge node {$1$} (q1);
    \path (q1) edge node {$1$} (q2);
    \path (q3) edge [above] node {$\frac{1}{4}$} (q00);
    \path (q3) edge [below] node {$\frac{1}{4}$} (q11);

    \path (q3) edge node {$\frac{1}{2}$} (qf);
    \path (q00) edge node {$\frac{1}{2}$} (qf);
    \path (q00) edge [loop above] node {$\frac{1}{4}$} (q00);
    \path (q11) edge [right] node {$\frac{1}{2}$} (qf);
    \path (q11) edge [loop below] node {$\frac{1}{4}$} (q11);
    \path[<->] (q00) edge [out=5, in=-5,looseness=1.5] node {$\frac{1}{4}$} (q11);

    \path (qf) edge [loop right] node {$1$} (qf);

\end{tikzpicture}
\caption{A family of MCs.
A state $s$ is labelled with $s$ itself; in particular, a state $\sigma_i$ is labelled with $\sigma_i \in \{0,1\}$.
}
\label{fig:MC-big-GfM}
\end{figure}

\begin{thm}\label{theorem:GFM-exponential}
 A GFM automaton that recognises $\lang_n^r$ or $\lang_n^s$ has at least $2^n$ states.   
\end{thm}

\begin{proof}
Let $n > 0$ and $\sigma_1,\ldots,\sigma_n \in \{0,1\}$. 
The family of MCs is shown in \cref{fig:MC-big-GfM}.
Since there are $2^n$ combinations of $\sigma_1,\ldots,\sigma_n$, we have $2^n$ MCs in this family. 

The chance that an MC in this family produces a word in $\lang_n^r$ (and $\lang_n^s$) is $2^{-(n+1)}+ \sum_{i=1}^{n} \sigma_i 2^{-i}$.
We briefly explain how we obtain this value.
The chance that the first $\$$ is $i \leq n$ steps after reaching $\sigma_n$ is $2^{-i}$, and the relative chance of winning then is $\sigma_i 2^{-i}$; this is the $\sum_{i=1}^{n} \sigma_i 2^{-i}$ part.
The chance that there never is a $\$$ is $0$.
The chance that there is a $\$$ with more than $n$ steps after reaching $\sigma_n$ is $2^{-n}$. The chance of winning in this case is $\frac{1}{2}$ (as a $0$ and $1$ exactly $n$ steps earlier are equally likely); this is the $2^{-(n+1)}$ part.

Let $\G$ be a GFM automaton that recognises $\lang_n^r$ (or $\lang_n^s$).
Now consider its product with one of the MCs $\M$. 
The chance that $\M$ produces a word in $\lang_n^r$ (and $\lang_n^s$) should also be the value of the state (a pair of $\sigma_n$ and a state in $\G$) in the product $\M \times \G$ when $\sigma_n$ is first reached.

\rev{The sub Markov chain from $\sigma_n$ onward -- comprising the states $0$, $1$, and $\$$, together with their transition probabilities -- is exactly the same for every MC in this family, so the probability distribution over what happens \emph{after} $\sigma_n$ is identical across the whole family, regardless of which values $\sigma_1,\ldots,\sigma_n$ were emitted to get there.
Suppose, towards a contradiction, that $\G$ were in the \emph{same} state $q$ when $\sigma_n$ is first reached for two MCs $\M$ and $\M'$ in the family whose emitted values $\sigma_1,\ldots,\sigma_n$ differ in at least one position.
Since future behaviour of $\G$ depends only on its current state and on the (here, identically distributed) future input, its syntactic probability of acceptance from this point onward would then have to be exactly the same for both $\M$ and $\M'$.
But the semantic probability of acceptance required from this point onward does depend on $\sigma_1,\ldots,\sigma_n$ -- by the formula above, it depends on which of $\sigma_1,\ldots,\sigma_n$ matches the state the random walk on $0$ and $1$ is in when $\$$ is finally reached -- so $\M$ and $\M'$ generally require \emph{different} continuation values to be optimal.
As $\G$ is GFM, it must match both targets exactly, which is impossible from a single shared state $q$.
Hence, in the product MDP $\M \times \G$, the state (a pair of $\sigma_n$ and a state in $\G$) when $\sigma_n$ is first reached must be different for every one of the $2^n$ MCs in the family.}
That is, $\sigma_n$ must be paired with different states of $\G$ for different MCs in this family.
As there are $2^n$ different MCs, there are at least $2^n$ states in $\G$.
\end{proof}

It can be seen from \cref{fig:NRA} (resp.\ \cref{fig:NSA}) that $\lang_n^r$ (resp.\ $\lang_n^s$) is recognised by a nondeterministic reachability (resp.\ safety) automaton with $n+2$ (resp.\ $n+1$) states. 
With \cref{theorem:GFM-exponential}, we have:
\begin{cor}\label{cor:nba-succinct}
General nondeterministic automata are exponentially more succinct than GFM automata.    
\end{cor}

We now consider the $\mathcal R_n$ and $\mathcal S_n$ more closely to see that we have used small subclasses of nondeterministic B\"uchi automata: $\mathcal R_n$ are unambiguous reachability automata, while $\mathcal S_n$ are run-unambiguous separating safety automata. 
Let us start with the $\mathcal R_n$. 
Note that it is not separating as $f$ is an \rev{accepting sink state} recognising all $\omega$-words.

\begin{restatable}{lem}{lemRnUnambiguous}
    For each $n\in \mathbb N$, $\mathcal R_n$ is 
    a run-unambiguous reachability automaton.
\end{restatable}
\begin{proof}
It is obviously a reachability automaton, as it has only a single accepting state, which is a sink.

To see that $\mathcal R_n$ is unambiguous, we look at a word in the language $\lang_n^r$ of the automaton. A word in $\lang_n^r$ has the form $\{0,1\}^m1\{0,1\}^{n-1}\}\$\{0,1,\$\}^\omega$, and its only accepting run is $q_0^{m+1},q_1,q_2,\ldots,q_n,f^\omega$:
any run that would move to $q_1$ earlier would block on reading a $0$ or $1$ while in $q_n$, while any run that would move to $q_1$ either later or not at all would block when reading a $\$$ in a state $q_i$ with $i<n$.

To see that $\mathcal R_n$ is run-unambiguous, we first observe that it only has two types of runs:
$q_0^\omega$ and $q_0^{m+1},q_1,q_2,\ldots,q_n,f^\omega$.
$q_0^\omega$ is a run of any word that does not contain a $\$$, while $q_0^{m+1},q_1,q_2,\ldots,q_n,f^\omega$ is a run of a word of which the first $\$$ is after exactly $m+n$ letters of $0$ or $1$, i.e., that the word is of the form \rev{$\{0,1\}^{m}1\{0,1\}^{n-1}\$\{0,1,\$\}^\omega$}.
These sets of words are pairwise disjoint, so that no word has more than one run.
\end{proof}

\begin{restatable}{lem}{lemSnStrongUnambiguous}
For each $n\in \mathbb N$, $\mathcal S_n$ is a run-unambiguous separating safety automaton.
\end{restatable}
\begin{proof}
Let $n\in \mathbb N$.
$\mathcal S_n$ is obviously a safety automaton, as all states are final.

To see that $\mathcal S_n$ is separating, for all $i>0$ the language accepted from state $q_i$ is $\{0,1\}^{i-1}\$^\omega$,
while the language from state $q_0$ contains only words, whose first $n$ letters are in $\{0,1\}$ (as it would block on any word containing a $\$$ in its first $n$ letters).

To see that $\mathcal S_n$ is run-unambiguous, we distinguish two types of words: those that do, and those that do not contain a $\$$. For the latter, $q_0^\omega$ is the only run, as all other runs that leave $q_0$ continue with $q_1,q_2,\ldots,q_n$ and then block.
For the former, as the automaton blocks when reading a $\$$ from every state except $q_n$, the only possible run is to move to $q_1$ exactly $n$ steps before the first $\$$ and then continue with $q_2,q_3,\ldots,q_{n-1},q_n^\omega$. (This will only be a run if the word is in $\lang_n^s$, but we only have to show that there is no other run.)
\end{proof}

With \cref{theorem:GFM-exponential}, this lemma provides:

\begin{cor}\label{cor:separated-succinct}
Run-unambiguous separating safety automata are exponentially more succinct than GFM automata.    
\end{cor}

As both $\mathcal R_n$ and $\mathcal S_n$ are run-unambiguous automata by the previous lemmas, \cref{theorem:GFM-exponential} provides:

\begin{cor}\label{cor:unambiguous-succinct}
Run-unambiguous reachability and safety automata are exponentially more succinct than GFM automata.    
\end{cor}

Finally, we can create a separating automaton with $n+2$ states that recognises the reachability language $\lang_n^r$, shown in \cref{fig:SRA}. However, one cannot expect such automata to be (syntactic) reachability automata, because all words are accepted from an accepting sink state.

\begin{figure}[ht]
\centering
\begin{tikzpicture} [->, node distance = 2cm, auto,initial text = {}]
    \node[initial left,state] (q0) at (0,1) {$q_0$};
    \node[state] (q1) at (2,1) {$q_1$};
    \node[state] (q2) at (4,1) {$q_2$};
    \node[state] (q3) at (6,1) {$q_3$};
    \node (qi) at (7,1) {$\ldots$};
    \node[state, accepting] (qn) at (8,1) {$q_n$};
    \node[state, accepting] (qf) at (10,1) {$f$};

    \path (q0) edge [loop above] node {$0,1$} (q0);
    \path (q0) edge [below] node {$1$} (q1);
    \path (q1) edge [below] node {$0,1$} (q2);
    \path (q2) edge [below] node {$0,1$} (q3);
    \path (qn) edge node {$\$$} (qf);
    \path (qn) edge [bend right=40] node [near end,below] {$\$$} (q3);
    \path (qn) edge [bend right=40] node [near end,below] {$\$$} (q2);
    \path (qn) edge [bend right=40] node [near end,below] {$\$$} (q1);
    \path (qn) edge [bend right=40] node [near end,below] {$\$$} (q0);
    \path (qf) edge [loop above] node {$0,1$} (qf);
    \path (qn) edge [loop above] node {$\$$} (qn);
    \path (qf) edge [bend left=27] node [below] {$0$} (q1);

\end{tikzpicture}
\caption{The separating NBW $\mathcal R_n'$ recognising the reachability language $\lang_n^r$.}
\label{fig:SRA}
\end{figure}

\begin{restatable}{lem}{lemRnpSeparating}
\label{lem:Rn'-separating}
For each $n\in \mathbb N$, $\mathcal R_n'$ is a separating automaton that recognises $\lang_n^r$.
\end{restatable}
\begin{proof}
To see that $\mathcal R_n'$ is separating, we note that
\begin{itemize}
    \item words starting with $\$$ can only be accepted from $q_n$, as the automaton blocks when reading $\$$ from any other state,
    \item words starting with $\{0,1\}^i\$$, with $0 < i < n$, can only be accepted from $q_{n-i}$, because
    \begin{itemize}
        \item when starting at $q_j$ with $j > n-i$, the automaton would block when reading a $0$ or $1$ after reaching $q_n$, while
        \item for all other states, the automaton would block when reading the first $\$$ while not having reached $q_n$,
    \end{itemize}
        \item words starting with $\{0,1\}^i1\{0,1\}^{n-1}\$$ can only be accepted from $q_0$, because
        \begin{itemize}
            \item from $q_i$ with $0<i \leq n$, the automaton would block when reading the $(n+1-i)$-th letter which is a $0$ or a $1$, while
            \item from $f$, it would block when moving upon reading the $j$-th letter with $j\leq i$ to $q_1$  when, reading the $(n+j)$-th letter (a $0$ or $1$) in state $q_n$ or, where $j>i+1$ (or where the automaton never moves to $q_1$), when reading the $(1+i+n)$-th letter, a $\$$, in a state in $q_k$ with $k<n$ (or $f$), 
        \end{itemize} 
        \item words starting with $\{0,1\}^i0\{0,1\}^{n-1}\$$ can only be accepted from $f$, using a similar argument, and 
        \item finally, words containing no $\$$ can only be accepted from $f$, because the only final state reachable from any other state of the automaton is $q_n$, and from $q_n$ the automaton would block, as it only accepts a $\$$ from there.
\end{itemize}
This partitions the set of words.
To see that $\mathcal R_n'$ recognises $\lang_n^r$, we note that, for all of the words that do contain a $\$$, this provides a recipe to be in $q_n$ when this $\$$ arrives from \emph{some} state, while for a word that does not contain a $\$$, we accept from $f$ by simply staying there for ever.
For a word in $\lang_n^r$, we can therefore follow this recipe until a $\$$ arrives, and then pick a successor state $q$ so that the tail of the word (after that $\$$) is in the language that can only be accepted from $q$, and so forth.
This constructs an accepting run.
(Note that we have already shown that words not in $\lang_n^r$ cannot be accepted from the only initial state, $q_0$.)
\end{proof}

With \cref{theorem:GFM-exponential}, this lemma provides:

\begin{cor}\label{cor:separated-reach-succinct}
Separating automata that recognise reachability languages are exponentially more succinct than GFM automata.    
\end{cor}
	

\section{Discussion}

We first established that deciding GFM-ness is {\sf PSPACE}-hard via a reduction from the NFA universality problem. 
Next, we introduced the seemingly simpler auxiliary concept of \emph{qualitative} GFM (QGFM) and developed an algorithm to check QGFM-ness in {\sf EXPTIME}. 
We then characterised GFM-ness using extensive QGFM tests, only to discover that this characterisation also forms a necessary condition for QGFM-ness itself, resulting in a collapse of the qualitative and full notions of GFM-ness. 
Consequently, the hardness results for GFM-ness carry over to QGFM-ness, and the {\sf EXPTIME} decision procedure for QGFM-ness serves as a decision procedure for GFM-ness.



We have established that GFM automata are exponentially more succinct than GFG automata, while general nondeterministic automata are exponentially more succinct than GFM automata. For the latter, we have shown that this is still the case when we restrict the class of general nondeterministic automata further, requiring them to be run-unambiguous reachability or safety automata, run-unambiguous separating safety automata, or separating automata that recognise a reachability language.
For all of these automata, we have provided very simple families of automata that witness the \rev{exponential} succinctness advantage.

Closing the remaining gap between the lower and upper complexity bounds remains an intriguing challenge for future work.

\section*{Acknowledgment}
  \noindent
  We thank an anonymous reviewer for raising the excellent question (to the original version) of whether or not QGFM and GFM are different.	They proved not to be. Without their clever question, we would not have considered this question, and thus not strengthened this paper accordingly.
  We also thank an anonymous reviewer for pointing out a simplification of the proof that GFM NBWs are exponentially more succinct than DBWs.
  
  This work was supported by the EPSRC through grants EP/X03688X/1 and EP/X042596/1.
  This project has received funding from the European Union's Horizon 2020 research and innovation programme under grant agreement 956123 (FOCETA).

\bibliographystyle{alphaurl}
\bibliography{paper}
\newpage
\end{document}